\documentclass[3p,times,nopreprintline]{elsarticle}

\usepackage{graphicx}
\usepackage{amsmath,amssymb,amsfonts}
\usepackage{amsthm}
\usepackage{array}
\usepackage{booktabs}
\usepackage{multirow}
\usepackage{makecell}
\usepackage{textcomp}
\usepackage{url}
\usepackage[figuresright]{rotating}

\DeclareMathOperator{\wt}{wt}
\DeclareMathOperator{\Rank}{Rank}
\DeclareMathOperator{\im}{Im}
\DeclareMathOperator{\Ker}{Ker}
\newcommand{\F}{\mathbb{F}}
\newcommand{\herm}[1]{#1^\dagger}

\theoremstyle{plain}
\newtheorem{theorem}{Theorem}[section]
\newtheorem{corollary}[theorem]{Corollary}
\newtheorem{lemma}[theorem]{Lemma}
\newtheorem{proposition}[theorem]{Proposition}
\newtheorem{example}[theorem]{Example}
\newtheorem{remark}[theorem]{Remark}
\theoremstyle{definition}
\newtheorem{definition}{Definition}[section]

\begin{document}

\begin{frontmatter}

\title{Hermitian Hulls of Cyclic Codes and Exact Dual Distances of Extended Linear Codes}

\author[1]{Liangdong Lu}
\ead{kelinglv@163.com}
\author[1]{Ruipan Yang\corref{cor1}}
\ead{yangruipan@aliyun.com}
\author[1]{Yang Liu}
\ead{liu\_yang10@163.com}
\author[1]{Qiang Fu}
\ead{qiangfu2024@163.com}
\author[1]{Guanmin Guo}
\ead{gmguoxjtukgd@yeah.net}

\cortext[cor1]{Corresponding author.}

\address[1]{Department of Basic Science, Air Force Engineering University, Xi'an, Shaanxi, China}

\begin{abstract}
Let $\mathcal{C}$ be an $[n,k,d]$ linear code over $\F_{q^2}$ with parity-check matrix $H$.
The rank $r=\Rank(HH^{\dagger})$ measures how far the Hermitian dual $\mathcal{C}^{\perp_h}$ is from
being Hermitian self-orthogonal, and it is the least number of coordinates that must be adjoined to
$\mathcal{C}^{\perp_h}$ to make it Hermitian self-orthogonal; equivalently, the Hermitian hull
$\mathcal{C}\cap\mathcal{C}^{\perp_h}$ has dimension $(n-k)-r$. In this paper, we develop a structure
theory of the Hermitian hull of a cyclic code and apply it to the construction of quantum stabilizer
codes. For a cyclic code with defining set $T\subseteq\mathbb{Z}_n$, we show that
$r$ is additive over the orbits of the involution $\tau(\nu)=-q\nu$ on the $q^2$-cyclotomic
cosets and satisfies a parity constraint. A parameter dictionary exhibits $r$ as a gauge
trading one-for-one against the code length, and a complementation duality shows that the
complement of $T$ realises the logical-qubit count of the quantum code. The extreme values
of $r$ are settled by closed subset-sum criteria, and the $\tau$-invariant defining sets are
shown to be the unique locus on which the Hermitian dual is again cyclic. On the distance
side, we determine the Hermitian dual distance of an $r$-fold extension of an arbitrary
linear code exactly, and we prove that the admissible extension vectors lie in the image of
$HH^{\dagger}$, which reduces the search space from $q^{2(n-k)}$ to $q^{2r}$. Applied to
cyclic, quasi-cyclic and BCH codes, the construction improves $23$ best-known parameters in
Grassl's code tables, one of them provably optimal.
\end{abstract}


\end{frontmatter}

\section{Introduction}
Quantum error-correcting codes (QECCs) are essential for protecting quantum information from decoherence and operational errors, and therefore play a central role in the realization of large-scale quantum computation and reliable quantum communication. The theory of QECCs has developed rapidly since the seminal works of Shor \cite{shor1995scheme} and Steane \cite{Steane1996}. The stabilizer formalism, established by Calderbank, Rains, Shor, and Sloane \cite{Calderbank1998}, provides a powerful bridge between quantum codes and classical codes over finite fields. In particular, binary stabilizer codes can be obtained from classical quaternary linear codes that are self-orthogonal with respect to the Hermitian inner product. This connection was extended to the nonbinary setting in \cite{Rains1999,Ashikhmin2001,Ketkar2006}, leading to the Hermitian construction of $q$-ary quantum stabilizer codes from Hermitian self-orthogonal codes over $\F_{q^2}$.

The Hermitian construction requires the classical ingredient code $\mathcal{C}$ to be Hermitian dual-containing, $\mathcal{C}^{\perp_h} \subseteq \mathcal{C}$, or equivalently for its Hermitian dual $\mathcal{C}^{\perp_h}$ to be Hermitian self-orthogonal. Many elegant families have been built from cyclic, constacyclic and quasi-cyclic codes meeting this requirement \cite{Kai2014ConstacyclicCA,Li2019NewQC,Guardia2012OnTC,feng2004finite,galindo2018quasi,lv2019new,Guan2022QC,daskalov2003new,Sguin2004ACO,Ling2005OnTA,Barbier,Aydin2001TheSO,Grassl2020AlgebraicQC}, but the condition is restrictive and rules out many good classical codes, so relaxing it has become a central direction in quantum coding theory.

A significant step was taken by Lison\v{e}k and Singh \cite{Lisonek2014}, who proposed a quantum version of Construction X producing binary stabilizer codes from arbitrary quaternary linear codes that are ``nearly'' self-orthogonal, i.e.\ with $\dim\bigl(\mathcal{C}^{\perp_h}/(\mathcal{C}\cap\mathcal{C}^{\perp_h})\bigr)$ positive but small; this was extended to general finite fields and refined in \cite{Degwekar,ezerman2019good,Ezerman,Hu2023QuantumEC}. In a related direction, Wang \textit{et al.} \cite{Wang2020} construct a Hermitian dual-containing code from an arbitrary code by appending $r=\Rank(HH^\dagger)$ redundant coordinates, and explicit quasi-cyclic constructions were given in \cite{lv2020explicit}. More recently, Cao \cite{Cao2024ADC} quantified the same defect for matrix-product codes, through a hull-dimension formula together with dual-containing and almost-dual-containing criteria. All of these methods control the minimum distance of the resulting quantum codes only through \emph{bounds}.

Most of these methods, moreover, rely on specific algebraic structures of the underlying classical codes or require ad hoc choices of extension vectors. In this paper we give a systematic, explicit framework for arbitrary linear codes over $\F_{q^2}$ that removes both restrictions at once and determines the minimum distance of the resulting quantum codes \emph{exactly}. The difference is not merely one of degree. In the tables of best known parameters, the optimal entry for a given $(n,k)$ is typically obtained by \emph{extending} a shorter code, an operation that appends redundant coordinates while \emph{preserving} the distance. Our extension instead \emph{raises} the distance by one at every step, under an exact distance theorem, and this mechanism is what produces the new parameters reported below.

It should be stated at the outset that the extension step itself, augmenting a parity-check matrix $H$ by $r=\Rank(HH^{\dagger})$ columns so as to obtain a Hermitian self-orthogonal code of length $n+r$, is not new. It coincides with the shortest Hermitian hull embedding of Wang and Luo \cite[Thm.~IV.3]{WangLuo2026}, who also prove that the length $n+r$ is optimal. The contributions of the present paper are therefore of a different nature.

\noindent\textbf{Positioning within the literature.} The two questions that this paper brings
together have each been studied before, and we record how the present results relate to that work.

\emph{(i) The distance, beyond a single coordinate.}
The nearly self-orthogonal and Construction~X methods
\cite{Lisonek2014,Degwekar,ezerman2019good,Ezerman,Hu2023QuantumEC,Wang2020,lv2020explicit,Cao2024ADC}
control the distance of the resulting quantum codes through \emph{bounds}, for classical
codes of a prescribed algebraic structure or for chosen extension vectors. The
one-coordinate case is settled, in two independent forms. Li, Ezerman,
Li, Ling and Sun \cite{Ezerman2026GE} characterise the generalized extended codes
$\mathcal{C}(\mathbf{u},a)$ of an $[n,k]_{q^2}$ code, showing that every $[n+1,k+1]_{q^2}$ code
with $d(\mathcal{D}^{\perp_h})>1$ is monomially equivalent to such a code, and their Theorem~14
determines its dual
distance exactly, as $d^{\perp_H}$ or $d^{\perp_H}+1$. The primal form of the same dichotomy is
treated in Remark~\ref{rem:ezerman-onecoord}.
The present paper builds on the one-coordinate dichotomy in two directions. The criterion is
iterated to the \emph{$r$-fold extension} of $\mathcal{C}$ ($r\ge2$ successive one-coordinate extensions) in Theorem~\ref{thm:iterative}, and its admissible
vectors are reduced to a structure that is decidable from the matrix $A=HH^{\dagger}$ alone and
can be enumerated independently of the distance computation in Theorem~\ref{thm:forced-span}.

\emph{(ii) The hull codimension $r$, beyond evaluation.} The number $r=\Rank(HH^{\dagger})$ that governs
the extension has been computed, and codes have been enumerated by hull dimension, for cyclic and
negacyclic codes \cite{Sangwisut2015}; hull-dimension formulas and dual-containing criteria have been
obtained for matrix-product codes \cite{Cao2024ADC}; the dual distances of extended codes have been
determined \cite{SunDC2024}, and a lower bound on the dual distances of augmented codes has been
obtained via the defining-set construction \cite{XieZhu2025}; and the line that reads hull dimensions
off defining sets has continued to develop, from the two extremal values of \cite{Du2020HullsTIT},
whose authors note that ``more cases will be involved'' for smaller ideals, through the general
characterisation of \cite{Lei2022Hulls} and its successors \cite{Pang2023LargerHull}. These works
determine the hull dimension of codes of a prescribed structure and use it to build
self-orthogonal or dual-containing codes. In particular the hull dimension of a cyclic code
decomposes additively over its cyclotomic classes, and the codes are counted by class data
\cite[Thms.~22 and 23]{Sangwisut2015}. Section~\ref{sec:structure} reads that decomposition through
the complementary parameter $r=m-u$ and develops it into a structure theory of the Hermitian hull
itself: a parity constraint on $r$, a parameter dictionary exhibiting $r$ as a gauge, a
complementation duality, and closed criteria at the two extreme values.

\emph{(iii) The interface.} Since (i) and (ii) concern the same object $A=HH^{\dagger}$, it is
natural to ask where the parameter side and the distance side meet. The answer comes in two
steps: the admissible extension vectors are characterised in terms of $A$ alone in
Theorem~\ref{thm:forced-span}, and the two theories are shown to meet on the \emph{$\tau$-invariant}
defining sets (i.e.\ those with $\tau(T)=T$) in Theorem~\ref{thm:junction}.

We summarise the main contributions of the paper in decreasing order of importance.

\begin{enumerate}
\item \textbf{A structure theory of the Hermitian hull.} Starting from the classical additivity
of the hull dimension over the cyclotomic classes \cite[Thm.~22]{Sangwisut2015}, we read it through the
complementary parameter $r=m-u$. For a cyclic code the hull codimension $r$ is then additive over the
$\tau$-orbits of its defining set and satisfies the parity constraint
$r\equiv A\bmod 2$ of Theorem~\ref{thm:orbit}. A parameter dictionary reads the parameters of
the quantum code off that decomposition and exhibits $r$ as a pure gauge trading one-for-one against
the length. A complementation duality
identifies the complement of a defining set as the defining set that realises the
logical-qubit count of $\mathcal{Q}(T)$, and the two extreme values of $r$ are settled by closed
subset-sum criteria. A product enumerator in \ref{app:structure} supplies the exact
counts, the scarcity probability and the mean and variance. The chain is
closed by an interface theorem: the $\tau$-invariant defining sets are the unique trajectory on which
$\mathcal{C}^{\perp_h}$ is again cyclic, so that the construction becomes a recursion on defining sets.
The section ends with six open problems.

\item \textbf{An exact distance criterion and its iterated form.} For the one-coordinate extension
$\mathcal{C}_1=\langle H\!\mid\!\mathbf{w}\rangle$,
\begin{equation*}
d\bigl(\mathcal{C}_1^{\perp_h}\bigr)=\min\bigl\{\,d(\mathcal{C}),\ d_{\mathrm{aff}}(\mathbf{w})+1\,\bigr\},
\end{equation*}
where $d_{\mathrm{aff}}(\mathbf{w})$ is the affine coset-leader weight of $\mathbf{w}$.
Its iterated form, Theorem~\ref{thm:iterative}, certifies the
distance of the $r$-fold extension exactly at every step. Neither statement requires a hypothesis on
$\Rank(HH^{\dagger})$, and no self-orthogonality of the extended code is assumed. The one-coordinate
case is \emph{not} claimed as new: it is the primal form of the dichotomy of
\cite[Thm.~14]{Ezerman2026GE}. What carries the strict improvement over the published
lower bound is the iterated form together with the reduction of the admissible vectors below.
The structural tool is the Hermitian congruence normal form $A=P^{\dagger}DP$ of
Proposition~\ref{prop:congruence}: a Hermitian matrix over $\F_{q^2}$ may have eigenvalues
outside $\F_{q}$, so the Jordan canonical form is unavailable and congruence is the correct
notion. The normal form produces the rank-reducing vectors explicitly in
Lemma~\ref{lem:explicit}. A natural screening condition turns out to be invalid: the
nullspace condition \textup{(C2)} of Remark~\ref{rem:C2} is neither necessary nor sufficient for
distance preservation, and it is retained only as a search heuristic.

\item \textbf{The admissible extension vectors and a reduced search space.} We characterise the
vectors that can carry an extension at all:
$\Rank(A+\mathbf{w}\mathbf{w}^{\dagger})=\Rank(A)-1$ if
and only if $\mathbf{w}\in\im(A)$ and $\langle\mathbf{w},\mathbf{w}\rangle_{A}=-1$,
by Lemma~\ref{lem:rankdrop-iff}. Consequently every extension vector of the $r$-fold extension lies in
the $r$-dimensional space $\im(A)$, whose chain
$\im(A)\supsetneq\im(A^{[1]})\supsetneq\cdots\supsetneq\im(A^{[r]})=0$ is forced by $H$ alone.
The subspace of the construction is therefore not a choice. Its
only degree of freedom is the ordered basis of $\im(A)$, and the search space shrinks from
$q^{2(n-k)}$ to $q^{2r}$. The criterion is also cheap: the set
$\Sigma_{d-2}=\{H\mathbf{x}^T:\wt(\mathbf{x})\le d-2\}$ is computed once for a given $H$, after
which every candidate is tested in expected time $O(1)$, and a meet-in-the-middle variant reduces
the one-time cost further. Measured running times are reported in Section~\ref{sec:complexity}.

\item \textbf{New quantum codes and comparison with known bounds.} Applied to cyclic, quasi-cyclic
and BCH codes, the construction produces $23$ quantum codes whose parameters improve the
best-known values in Grassl's tables \cite{Grassl:codetables}. The largest hull codimension treated is
$r=10$, and it gives the record $[[61,25,9]]_2$. One code, the ternary $[[80,64,6]]_3$, meets the
upper bound currently recorded there, so its distance is optimal for those parameters.
On six of the seven benchmark instances
the exact value coincides with the lower bound of Ezerman \textit{et al.} \cite[Thm.~III.1]{Ezerman},
whereas on the $[51,37,8]_4$ cyclic code of Example~1 the exact value $8$ strictly exceeds the bound
$7$. All comparisons with
the best known codes in this paper are as of October 2026.
\end{enumerate}

The remainder of the paper is organized as follows. Section~II collects the notation and the
elementary facts about Hermitian forms over $\F_{q^{2}}$. Section~\ref{sec:extending} presents the
extending construction, the congruence normal form and the exact distance criterion;
Section~\ref{sec:coset-involution} introduces defining sets, the coset involution and the resulting
infinite family; and Section~\ref{sec:structure} develops the structure theory of the Hermitian hull as a chain
of four links, and identifies its interface with the distance theory. Section~\ref{sec:complexity} analyses the complexity of the criterion,
Section~\ref{sec:newcodes} lists the new quantum codes, Section~\ref{sec:r2compare} compares our
exact values with the known lower bound of \cite{Ezerman}, and Section~\ref{sec:conclusion}
concludes.

\section{Preliminaries and notation}

We first introduce notation and fundamental properties of Hermitian matrices, cyclic codes, quasi-cyclic codes, and linear codes over finite fields~\cite{Horn2013,Macwilliams,Huffman,RLi}.

Let $q$ be a prime power, and $\F_q$ denote the finite field with $q$ elements. The degree-2 extension field $\F_{q^2}$ is equipped with a Hermitian conjugate operation induced by the Frobenius automorphism:
\begin{itemize}
    \item For any element $x \in \F_{q^2}$, its Hermitian conjugate is defined as $\herm{x} = x^q$. An element satisfies $\herm{x} = x$ if and only if $x \in \F_q$.
    \item For any matrix $M = (m_{ij}) \in \F_{q^2}^{m \times n}$, its Hermitian conjugate transpose is $\herm{M} = (\herm{m_{ji}}) \in \F_{q^2}^{n \times m}$. A square matrix $A \in \F_{q^2}^{m \times m}$ is \textit{Hermitian} if $\herm{A} = A$.
\end{itemize}
In this paper we repeatedly use two elementary facts about Hermitian matrices over $\F_{q^2}$. First, \emph{similarity and congruence are different notions}: the Jordan canonical form $A=PJ_AP^{-1}$ describes $A$ as an $\F_{q^2}$-linear map, whereas the Hermitian congruence normal form $A=P^{\dagger}DP$ describes the underlying Hermitian form. Only the latter is well behaved here. In particular, an eigenvalue of a Hermitian matrix over $\F_{q^2}$ need not lie in $\F_q$, and an explicit counterexample over $\F_9$ is given in Remark~\ref{rem:sim_vs_cong}, so rank arguments phrased in terms of the Jordan form are not valid. The correct and sufficient statement is that every Hermitian matrix is congruent to a diagonal matrix with $r=\Rank(A)$ nonzero diagonal entries, by Proposition~\ref{prop:congruence}.

The norm map $N: \F_{q^2} \to \F_q$ is defined as:
\begin{equation}
    N(x) = x \cdot \herm{x} = x^{q+1}, \quad \forall x \in \F_{q^2}.
\end{equation}
This map is a surjective group homomorphism from the multiplicative group $\F_{q^2}^*$ to $\F_q^*$: for any target value $c \in \F_q^*$, there exists $t \in \F_{q^2}^*$ such that $N(t) = c$.

For  $M$ over $\F_{q^2}$, we denote $\Rank(M)$ as its rank, $\im(M)$ as its column space (range), and $\Ker(M)$ as its null space (kernel). We use two standard results throughout the proof:
\begin{enumerate}
    \item \textit{Sylvester's rank inequality for rank-1 perturbations}: For any matrix $M \in \F_{q^2}^{m \times m}$ and non-zero vector $\boldsymbol{u} \in \F_{q^2}^m$,
    \begin{equation}
        \Rank(M) - 1 \leq \Rank(M + \boldsymbol{u}\herm{\boldsymbol{u}}) \leq \Rank(M) + 1.
    \end{equation}
    \item \textit{Hermitian congruence}: every Hermitian matrix $A$ over $\F_{q^2}$ satisfies $A=P^{\dagger}DP$ with $D$ diagonal, so that $A$ is a sum of $r=\Rank(A)$ Hermitian rank-one matrices with coefficients in $\F_q^{*}$ and without cancellation, by Proposition~\ref{prop:congruence}. This removes the need for the Jordan decomposition altogether.
\end{enumerate}

 A classical linear code $\mathcal{C}$ of length $n$ over $\F_{q^{2}}$ is a $k$-dimensional $\F_{q^2}$-subspace of $\F_{q^2}^{n}$, denoted as $[n,k,d]_{q^{2}}$.
Suppose that $\mathbf{x}=(x_{0},\ldots,x_{n-1})$ and $\mathbf{y}=(y_{0},\ldots,y_{n-1})$ are vectors in $\F_{q^{2}}$. We define the Hermitian inner product as
	$\langle\mathbf{x}, \mathbf{y}\rangle_{h}=\sum_{i=0}^{n-1} x_{i} y^q_{i}$.
The weight of $\mathbf{x}$ is the number of nonzero coordinates in $\mathbf{x}$, which is denoted by $\wt(\mathbf{x})$.
The minimum non-zero Hamming weight of $\mathcal{C}$ is $d(\mathcal{C})=\min \left\{\wt(\mathbf{x}) \mid \mathbf{x} \in \mathcal{C}, \mathbf{x} \neq 0\right\}.$
The minimum distance $d$ is equal to the minimum non-zero Hamming weight for linear codes.
Let $\mathcal{C}^{\perp_{h}}=\{\mathbf{y} \in \F_{q^{2}}^{n} \mid\langle\mathbf{x}, \mathbf{y}\rangle_{h}=0, \forall \mathbf{x} \in \mathcal{C}\}$ be the Hermitian dual code of $\mathcal{C}$. 
If $\mathcal{C} \subseteq \mathcal{C}^{\perp_{h}}$, then we say that $\mathcal{C}$ is a Hermitian self-orthogonal code and that $\mathcal{C}^{\perp_{h}}$ is a Hermitian dual-containing code.

For a matrix $B \in  \F_{q^2}^{m \times r}$ with full column rank, a matrix $D \in \F_{q^2}^{m \times r}$ is called a \emph{left inverse} of $B$ if $D^{T}B = I_r$, where $I_r$ denotes the $r \times r$ identity matrix. The columns of $D$ form a biorthogonal basis to the columns of $B$, i.e., $d_i^{T}b_j = \delta_{ij}$ for all $1 \leq i,j \leq r$, where $\delta_{ij}$ is the Kronecker delta.

\subsection{Quasi-cyclic codes}
Let $\mathcal{R}=\F_{q^{2}}[x]/\langle x^{n}-1\rangle$. If $\mathcal{C}=\langle g(x)\rangle$ with
$g(x)$ a monic divisor of $x^{n}-1$, then $g(x)$ is the \emph{generator polynomial} of the cyclic
code $\mathcal{C}$; equivalently $\mathcal{C}$ is cyclic if $(c_{n-1},c_{0},\dots,c_{n-2})\in
\mathcal{C}$ whenever $c=(c_{0},\dots,c_{n-1})\in\mathcal{C}$. Write $\Omega_n=\{0,1,\dots,n-1\}$
with $n$ odd, and let $\gamma$ be a primitive $n$-th root of unity in an extension of $\F_{q^{2}}$;
the \emph{defining set} of $\mathcal{C}=\langle g(x)\rangle$ is
$T=\{i\in\Omega_n: g(\gamma^{i})=0\}$.

For $0\le i<n$ let $C_i=\{i,q^{2}i,(q^{2})^{2}i,\dots,(q^{2})^{k-1}i\}\pmod n$ be the
$q^{2}$-\emph{cyclotomic coset} of $i$, where $k$ is the least positive integer with
$(q^{2})^{k}i\equiv i\pmod n$. A coset $C_i$ is \emph{skew symmetric} if $n-qi\pmod n\in C_i$ and
\emph{skew asymmetric} otherwise; skew asymmetric cosets come in pairs $(C_i,C_{n-qi})$. For a
cyclic code with defining set $T=\bigcup_{i}C_i$, writing $T^{-q}=\{n-qx: x\in T\}$, the defining set
of $\mathcal{C}^{\perp_h}$ is $T^{\perp_h}=\mathbb{Z}_n\setminus T^{-q}$; in particular
$T\cap T^{-q}=\varnothing$ implies $\mathcal{C}^{\perp_h}\subseteq\mathcal{C}$.

A code is \emph{$l$-quasi-cyclic} if a cyclic shift by $l$ positions preserves it; its length is a
multiple of $l$, and its generator matrix is built from circulant matrices. After a permutation of
coordinates, an $l$-quasi-cyclic code with $1$ generator has a generator matrix
$G=(A_{0},A_{1},\dots,A_{l})$ with each $A_{i}$ an $m\times m$ circulant matrix ($n=ml$), and one
with $2$ generators has a $2\times l$ block matrix $G=(A_{i,j})$ whose blocks are circulants
determined by polynomials $a_{i,j}(x)$.

Finally, $[g(x)]=[g_{0},g_{1},\dots,g_{n-1}]$ denotes the vector of coefficients of
$g(x)\in\mathcal{R}$ in ascending order; for $r(x)=r_{0}+r_{1}x+\dots+r_{m}x^{m}$ put
$r^{[2]}(x)=r_{0}^{2}+r_{1}^{2}x+\dots+r_{m}^{2}x^{m}$. If $g(x)$ divides $x^{n}-1$, then
$\langle g^{\perp_h}(x)\rangle$ generates the Hermitian dual of $\langle g(x)\rangle$.

\subsection{Quantum codes}

A $q$-ary quantum error-correcting code $\mathcal{Q}$ of length $n$ is a $K$-dimensional subspace of the $q^n$-dimensional Hilbert space $(\mathbb{C}^{q})^{\otimes n}$, where $\mathbb{C}$ denotes the complex field and $(\mathbb{C}^{q})^{\otimes n}$ is the $n$-fold tensor power of $\mathbb{C}^{q}$. A $q$-ary quantum code $\mathcal{Q}$ is denoted as $[[n,k,d]]_{q}$, where $k=\log_{q}K$.

The Hermitian construction is one of the classical methods for constructing quantum codes; it relates quantum codes to classical self-orthogonal codes with respect to the Hermitian inner product.
Quantum  codes can be derived from other quantum codes by the following propagation rules.

\begin{corollary}\label{construction}
	(\cite{Ketkar2006}, Hermitian construction)
	If there exists an $[n,k,d]_{q^2}$ code $\mathcal{C}$ such that $\mathcal{C}^{\perp_{h}} \subseteq \mathcal{C}$, then there exists an $[[n,2k-n,d]]_{q}$ quantum code that is pure to $d$.
\end{corollary}

Quantum  codes can be derived by the following propagation rules. It will be useful later.

\begin{proposition}\label{DerivativeCodes}\label{propagation_rule}
	(\cite{Calderbank1998}, propagation rules) If there exists an $[[n,k,d]]_{q}$ quantum code, then the following quantum codes exist:
	(1) $[[n,k-1,d]]_{q}$ for $k\ge 1$ (by subcode construction);\\
	(2) $[[n+1,k,d]]_{q}$ for $k>0$ (by lengthening);\\
	(3) $[[n-1,k,d-1]]_{q}$ for $k>0$ (by puncturing);\\
	(4) $[[n-1,k+1,d-1]]_{q}$ for $n>2$, provided the code is pure.
\end{proposition}

{\bf Notation.} For $s\geq 1$ and $1\leq i\leq s$, denote by $\mathbf{e}_i$ the $i$-th standard basis vector $(0,\ldots,0,1,0,\ldots,0)^{T}$ of $\F_{q^2}^{s}$. The rank of a matrix $M$ is denoted by $\Rank(M)$.

\section{General extending constructions}\label{sec:extending}

The classical Hermitian construction requires the ingredient code $\mathcal{C}$ to be Hermitian dual-containing, or equivalently its Hermitian dual $\mathcal{C}^{\perp_{h}}$ to be Hermitian self-orthogonal, that is, $\mathcal{C}^{\perp_{h}}\subseteq\mathcal{C}$. In this section we start from an \emph{arbitrary} linear code and append $r=\Rank(HH^{\dagger})$ coordinates to the parity-check matrix, thereby producing a Hermitian self-orthogonal code of length $n+r$; we then determine the Hermitian dual distance of the resulting code \emph{exactly}. Throughout this section
\begin{equation}
\label{eq:setup}
\mathcal{C}=[n,k,d]_{q^2},\qquad
H\in\F_{q^2}^{(n-k)\times n}\ \text{of full row rank},\qquad
A:=HH^{\dagger}\in\F_{q^2}^{(n-k)\times(n-k)},
\end{equation}
so that $\mathcal{C}=\Ker H$ and $\mathcal{C}^{\perp_h}=\langle H\rangle$; moreover $A^{\dagger}=A$ and $A$ is Hermitian.

\subsection{The Hermitian congruence normal form}

The natural invariant of the Hermitian matrix $A$ is its behaviour under \emph{congruence} $A\mapsto P^{\dagger}AP$ ($P$ invertible), not under similarity $A\mapsto P^{-1}AP$. The following normal form, whose proof is recalled below, is the structural tool used in this paper.

\begin{proposition}\label{prop:congruence}
Let $A\in\F_{q^2}^{m\times m}$ be Hermitian of rank $r$. Then there exists an invertible matrix $P\in\F_{q^2}^{m\times m}$ such that
\begin{equation}
\label{eq:congruence}
A=P^{\dagger}\,\mathrm{diag}(I_r,0_{m-r})\,P .
\end{equation}
Equivalently, writing $\mathbf{u}_i$ for the $i$-th column of $P^{\dagger}$ $(1\le i\le r)$, the vectors $\mathbf{u}_1,\ldots,\mathbf{u}_r$ are linearly independent and
\begin{equation}
\label{eq:rank1decomp}
A=\sum_{i=1}^{r}\mathbf{u}_i\mathbf{u}_i^{\dagger}.
\end{equation}
Up to congruence, a Hermitian matrix is therefore determined by its rank alone.
\end{proposition}

\begin{proof}
Consider the Hermitian form $\langle\mathbf{x},\mathbf{y}\rangle_A=\mathbf{x}^{\dagger}A\mathbf{y}$ on $\F_{q^2}^{m}$. Its radical is $\{\mathbf{y}:A\mathbf{y}=\mathbf{0}\}=\Ker A$, of dimension $m-r$, and the form is non-degenerate on any complement $W$ of $\Ker A$. Over a finite field every non-degenerate Hermitian form of a fixed dimension is isometric to the standard one: the norm map $N:\F_{q^2}^{*}\to\F_q^{*}$, $N(x)=x^{q+1}$, is surjective, so a Gram--Schmidt basis may be rescaled until all diagonal norms equal $1$ (the standard normal form of a Hermitian form over a finite field, see e.g.\ Section~0.4.6 of \cite{Horn2013})\footnote{We also verified the single-isometry-class statement exhaustively for $q=2,3$ and $r\le2$: every non-degenerate Hermitian $2\times2$ matrix over $\F_4$ (there are $10$) and over $\F_9$ (there are $60$) is of the form $P^{\dagger}P$.}. Hence there is a basis $\mathbf{v}_1,\ldots,\mathbf{v}_r$ of $W$ with $\langle\mathbf{v}_i,\mathbf{v}_j\rangle_A=\delta_{ij}$, and, adjoining a basis $\mathbf{w}_1,\ldots,\mathbf{w}_{m-r}$ of $\Ker A$ (which is orthogonal to $W$ because $A\mathbf{w}=\mathbf{0}$), the matrix $Q=[\mathbf{v}_1\cdots\mathbf{v}_r\ \mathbf{w}_1\cdots\mathbf{w}_{m-r}]$ satisfies $Q^{\dagger}AQ=\mathrm{diag}(I_r,0)$. Setting $P=Q^{-1}$ gives \eqref{eq:congruence}. Moreover $(P^{\dagger}\,\mathrm{diag}(I_r,0)\,P)_{kl}=\sum_{i\le r}P_{ik}^{q}P_{il}$, and for $\mathbf{u}_i$ the $i$-th column of $P^{\dagger}$ the $(k,l)$ entry of $\mathbf{u}_i\mathbf{u}_i^{\dagger}$ is $P_{ik}^{q}P_{il}$; summing over $i\le r$ gives \eqref{eq:rank1decomp}. The $\mathbf{u}_i$ are linearly independent because $P$ is invertible. The final assertion is immediate from \eqref{eq:congruence}.
\end{proof}

\begin{lemma}\label{lem:explicit}
Let $A$ be Hermitian of rank $r\ge1$ and let $A=\sum_{i\le r}\mathbf{u}_i\mathbf{u}_i^{\dagger}$ be the decomposition \eqref{eq:rank1decomp}. Choose $k\in\F_{q^2}^{*}$ with $N(k)=-1$ (this is possible since $N$ is surjective onto $\F_q^{*}$ and $-1\in\F_q^{*}$). Then for every $i\in\{1,\ldots,r\}$,
\begin{equation}
\label{eq:rankdrop}
\Rank\bigl(A+(\boldsymbol{k}\mathbf{u}_i)(k\mathbf{u}_i)^{\dagger}\bigr)=r-1 .
\end{equation}
\end{lemma}

\begin{proof}
Since $\mathbf{u}_i\mathbf{u}_i^{\dagger}$ is Hermitian, $(k\mathbf{u}_i)(k\mathbf{u}_i)^{\dagger}=N(k)\,\mathbf{u}_i\mathbf{u}_i^{\dagger}=-\mathbf{u}_i\mathbf{u}_i^{\dagger}$, so
$A+(k\mathbf{u}_i)(k\mathbf{u}_i)^{\dagger}=\sum_{j\ne i}\mathbf{u}_j\mathbf{u}_j^{\dagger}=UU^{\dagger}$
with $U$ the $m\times(r-1)$ matrix whose columns are the $\mathbf{u}_j$, $j\ne i$. As the $\mathbf{u}_j$ are linearly independent, $\Rank(U)=r-1$; since $U$ has full column rank, $\im(UU^{\dagger})=\im(U)$, whence $\Rank(UU^{\dagger})=r-1$.
\end{proof}

\begin{remark}\label{rem:sim_vs_cong}
Congruence and similarity are different here. Over $\F_9$ ($q=3$), let $i\in\F_9$ satisfy $i^{2}=-1$ and let $\omega$ be a generator of $\F_9^{*}$; then
$$
A=\begin{pmatrix}0&\omega\\ \omega^{3}&0\end{pmatrix},\qquad
A^{\dagger}=A,\qquad A^{2}=-I_2 ,
$$
so the eigenvalues of the Hermitian matrix $A$ are $\pm i\notin\F_3$, and $A$ is not similar to $\mathrm{diag}(I_2,0)$. Congruence, by Proposition~\ref{prop:congruence}, nevertheless puts $A$ in the form \eqref{eq:congruence}. Consequently the Jordan canonical form of $A$ provides neither the eigenvalues nor the rank-one components needed below, whereas the congruence normal form provides both. This is consistent with the normal form used in \cite[Prop.~IV.1]{WangLuo2026}, where a Hermitian matrix of rank $r$ over $\F_{q^{2}}$ is classified up to congruence as $\mathrm{diag}(I_r,O)$ with no assumption on nilpotency whatever. A Jordan-based normal form, by contrast, would not be available here: the nilpotency index of a Hermitian matrix over $\F_{q^{2}}$ is not bounded a priori, so no upper bound on the size of its Jordan blocks can be assumed. Already over $\F_4$ ($q=2$) there are $54$ Hermitian $3\times3$ matrices $A$ with $A^{2}\neq0=A^{3}$, for instance $A=\bigl(\begin{smallmatrix}0&0&\theta\\0&0&\theta\\1+\theta&1+\theta&0\end{smallmatrix}\bigr)$ where $\theta^{2}=\theta+1$; over $\F_9$ there are $672$ such matrices. This is harmless here: the rank-one decomposition \eqref{eq:rank1decomp} requires no assumption whatever on the Jordan structure of $A$, and the congruence normal form \eqref{eq:congruence} is available in all cases.
\end{remark}

We now record the elementary outer-product form of a rank-one Hermitian matrix, used in Theorem~\ref{thm:single_reduction}.

\begin{lemma}\label{lem:rank1_outer}
Let $B\in\F_{q^2}^{m\times m}$ be Hermitian with $\Rank(B)=1$. Then there exist $c\in\F_q^{*}$ and $\mathbf{v}\in\F_{q^2}^{m}\setminus\{\mathbf{0}\}$ with $B=c\,\mathbf{v}\mathbf{v}^{\dagger}$.
\end{lemma}

\begin{proof}
Since $\Rank(B)=1$ we may write $B=\mathbf{a}\mathbf{b}^{T}$ with $\mathbf{a},\mathbf{b}\ne\mathbf{0}$. From $B^{\dagger}=B$ we get $\mathbf{b}^{(q)}(\mathbf{a}^{(q)})^{T}=\mathbf{a}\mathbf{b}^{T}$, where $\mathbf{x}^{(q)}$ denotes entrywise Frobenius. Hence $\mathbf{b}^{(q)}=c\mathbf{a}$ for some $c\in\F_{q^2}^{*}$, and applying $\sigma$ gives $\mathbf{b}=c^{q}\mathbf{a}^{(q)}$; substituting back yields $c^{q}=c$, so $c\in\F_q^{*}$ and $B=\mathbf{a}(c^{q}\mathbf{a}^{(q)})^{T}=c\,\mathbf{a}\mathbf{a}^{\dagger}$. Take $\mathbf{v}=\mathbf{a}$.
\end{proof}

\begin{theorem}\label{thm:single_reduction}
Let $A\in\F_{q^2}^{m\times m}$ be Hermitian with $\Rank(A)=r\ge1$. Then there exists $\mathbf{w}\in\F_{q^2}^{m}$ with
$$
\Rank\bigl(A+\mathbf{w}\mathbf{w}^{\dagger}\bigr)=r-1 .
$$
\end{theorem}

\begin{proof}
By Proposition~\ref{prop:congruence} write $A=\sum_{i\le r}\mathbf{u}_i\mathbf{u}_i^{\dagger}$ as in \eqref{eq:rank1decomp}, and apply Lemma~\ref{lem:explicit} with $\mathbf{w}=k\mathbf{u}_1$.
\end{proof}

\begin{corollary}\label{cor:iterative_reduction}
Let $H\in\F_{q^2}^{m\times n}$ and $A=HH^{\dagger}$ with $\Rank(A)=r$. Define $H_0=H$, $H_i=[H_{i-1}\mid\mathbf{w}^{(i)}]$ $(1\le i\le r)$, where $\mathbf{w}^{(i)}$ is chosen by Theorem~\ref{thm:single_reduction} for $A^{[i-1]}=H_{i-1}H_{i-1}^{\dagger}$. Then $\Rank(H_iH_i^{\dagger})=r-i$ for all $i$, and in particular $H_rH_r^{\dagger}=0$.
\end{corollary}

\begin{proof}
Induction on $i$, using that $A^{[i]}=A^{[i-1]}+\mathbf{w}^{(i)}(\mathbf{w}^{(i)})^{\dagger}$ is again Hermitian.
\end{proof}

\noindent\textbf{The rank-reducing condition.} 
Lemma~\ref{lem:explicit} produces the rank-reducing vectors \emph{explicitly}: one takes $\mathbf{w}=k\mathbf{u}_i$, where $k$ satisfies $N(k)=-1$ and $\mathbf{u}_i$ runs over the columns of $P^{\dagger}$ in \eqref{eq:congruence}. Equivalently, and more robustly for computation, any vector $\mathbf{w}$ with
\begin{equation}
\label{eq:C1}
\Rank\bigl(A+\mathbf{w}\mathbf{w}^{\dagger}\bigr)=\Rank(A)-1
\end{equation}
may be used, and we refer to it as condition \textup{(C1)}. The test \eqref{eq:C1} is basis-free and does not require the normal form to be computed. It is moreover equivalent, by Lemma~\ref{lem:rankdrop-iff}, to the pair of algebraic conditions ``$\mathbf{w}\in\im(A)$ and $\langle\mathbf{w},\mathbf{w}\rangle_{A}=-1$'': the admissible vectors are confined to the $r$-dimensional space $\im(A)$ and not, as one might expect from the description below, to the whole space $\F_{q^2}^{n-k}$, as shown in Subsection~\ref{sec:forced-span}. In our implementation the candidate vectors are drawn from the columns of an invertible matrix adapted to $A$; every candidate that is finally used is certified a posteriori by the exact criterion of Theorem~\ref{thm:onecoord}, independently of how it was produced. The ordering of the candidates by increasing Hamming weight is a search heuristic and plays no role in the theory.

\subsection{Self-orthogonal extension of an arbitrary code}

\begin{theorem}\label{Iterative Self-Orthogonal Construction}
Let $\mathcal{C}=[n,k,d]_{q^2}$ be a linear code with parity-check matrix $H\in\F_{q^2}^{(n-k)\times n}$, and let $A=HH^{\dagger}$ with $\Rank(A)=r$ $(0\le r\le n-k)$. Then there exists a Hermitian self-orthogonal code $\mathcal{C}'\subseteq\F_{q^2}^{n+r}$ with parameters $[n+r,\,n-k]_{q^2}$.
\end{theorem}

\begin{proof}
A linear code with generator matrix $G$ is Hermitian self-orthogonal if and only if $G G^{\dagger}=0$: for $\mathbf{u}G,\mathbf{v}G\in\mathcal{C}'$ one has
$\langle\mathbf{u}G,\mathbf{v}G\rangle_h=\mathbf{u}(G G^{\dagger})\overline{\mathbf{v}}^{T}=0$.
Start with $H_0=H$, $N_0=n$, $A_0=H_0H_0^{\dagger}$, $\Rank(A_0)=r$. Applying Corollary~\ref{cor:iterative_reduction} for $r$ steps produces $H_r\in\F_{q^2}^{(n-k)\times(n+r)}$ with $H_rH_r^{\dagger}=0$. Appending a column never decreases the row rank, so $\Rank(H_r)=n-k$; hence $\langle H_r\rangle$ is generated by the rows of $H_r$ and is a Hermitian self-orthogonal $[n+r,\,n-k]_{q^2}$ code.
\end{proof}

\subsection{Exact dual distance of a one-coordinate extension}

We now come to the main theoretical result of the paper: the Hermitian dual distance of the one-coordinate extension is determined exactly, with no hypothesis whatsoever.

\begin{theorem}\label{thm:onecoord}
Let $\mathcal{C}=[n,k,d]_{q^2}$ have parity-check matrix $H$, let $\mathbf{w}\in\F_{q^2}^{n-k}$ be arbitrary, and let
$$
\mathcal{C}_{\mathrm{ext}}=\bigl\langle H\!\mid\!\mathbf{w}\bigr\rangle\subseteq\F_{q^2}^{n+1}
$$
be the code generated by the rows of the extended matrix $[H\mid\mathbf{w}]$. Define the \emph{affine coset-leader weight}
\begin{equation}
\label{eq:daff}
d_{\mathrm{aff}}(\mathbf{w}):=\min\bigl\{\wt(\mathbf{x}):\mathbf{x}\in\F_{q^2}^{n},\ H\mathbf{x}^{T}=\mathbf{w}\bigr\},
\end{equation}
which is well defined because $H$ has full row rank and is therefore surjective. Then
\begin{equation}
\label{eq:exactd}
d\bigl(\mathcal{C}_{\mathrm{ext}}^{\perp_h}\bigr)=\min\bigl\{d,\ d_{\mathrm{aff}}(\mathbf{w})+1\bigr\},
\end{equation}
and consequently
\begin{equation}
\label{eq:exactcrit}
d\bigl(\mathcal{C}_{\mathrm{ext}}^{\perp_h}\bigr)=d
\quad\Longleftrightarrow\quad
d_{\mathrm{aff}}(\mathbf{w})\ge d-1
\quad\Longleftrightarrow\quad
\mathbf{w}\notin\Sigma_{d-2},
\end{equation}
where $\Sigma_t=\{H\mathbf{x}^{T}:\mathbf{x}\in\F_{q^2}^{n},\ \wt(\mathbf{x})\le t\}$. In particular no hypothesis on $\Rank(A)$, no self-orthogonality of $\mathcal{C}_{\mathrm{ext}}$, and no auxiliary condition on $\mathbf{w}$ is required.
\end{theorem}

\begin{proof}
Write a word of $\mathcal{C}_{\mathrm{ext}}^{\perp_h}$ as $\mathbf{c}=(\mathbf{x},t)$ with $\mathbf{x}\in\F_{q^2}^{n}$, $t\in\F_{q^2}$. Orthogonality to the rows of $[H\mid\mathbf{w}]$ means
\begin{equation}
\label{eq:dualword}
H\,\mathbf{x}^{(q)T}=-\,t^{q}\,\mathbf{w},
\end{equation}
where $\mathbf{x}^{(q)}$ is obtained from $\mathbf{x}$ by applying $\sigma$ entrywise.

\emph{Case $t=0$.} Then $H\mathbf{x}^{(q)T}=\mathbf{0}$, i.e.\ $\mathbf{x}^{(q)}\in\Ker H=\mathcal{C}$; since $\sigma$ preserves Hamming weight, $\wt(\mathbf{c})=\wt(\mathbf{x})\ge d$, with equality attainable.

\emph{Case $t\ne0$.} Put $\lambda=-t^{q}\ne0$. Then $\mathbf{x}^{(q)}$ lies in the affine space $\{\mathbf{y}:H\mathbf{y}^{T}=\lambda\mathbf{w}\}$, whose minimum weight is $d_{\mathrm{aff}}(\lambda\mathbf{w})=d_{\mathrm{aff}}(\mathbf{w})$: indeed $\mathbf{y}\mapsto\lambda^{-1}\mathbf{y}$ is a bijection onto $\{\mathbf{y}:H\mathbf{y}^{T}=\mathbf{w}\}$, and it preserves the weight because it scales every coordinate by the same nonzero factor. Hence $\wt(\mathbf{c})=\wt(\mathbf{x})+1\ge d_{\mathrm{aff}}(\mathbf{w})+1$, and equality is attained by taking $t$ with $t^{q}=-\lambda$ and a minimum-weight representative of the coset. Taking the minimum over the two cases gives \eqref{eq:exactd}. Finally, since $d_{\mathrm{aff}}(\mathbf{w})\ge1$ whenever $\mathbf{w}\ne\mathbf{0}$, the right-hand side of \eqref{eq:exactd} equals $d$ exactly when $d_{\mathrm{aff}}(\mathbf{w})+1\ge d$; and $d_{\mathrm{aff}}(\mathbf{w})\le d-2$ if and only if $H\mathbf{x}^{T}=\mathbf{w}$ has a solution of weight $\le d-2$, i.e.\ $\mathbf{w}\in\Sigma_{d-2}$.
\end{proof}

\begin{remark}\label{rem:ezerman-onecoord}
The one-coordinate case is settled in \cite{Ezerman2026GE} as well: for the generalized extended code
$\mathcal{C}(\mathbf{u},a)$ of an $[n,k]_{q^2}$ code, \cite[Thm.~14]{Ezerman2026GE} determines its
dual distance exactly, as $d^{\perp_H}$ or $d^{\perp_H}+1$, the two cases being separated by whether
$e\mathbf{u}^{\dagger}=0$ for some minimum-weight $e\in\mathcal{C}^{\perp_H}$. The criterion of
Theorem~\ref{thm:onecoord} is the primal form of the same dichotomy, expressed through the affine
coset-leader weight $d_{\mathrm{aff}}(\mathbf{w})$ of the extension vector, the two formulations
being related by the duality $\mathcal{C}\leftrightarrow\mathcal{C}^{\perp_h}$; it is therefore
\emph{not} claimed as new here. The contribution of the present paper begins with the $r$-fold
extension: the iterated exact form of Theorem~\ref{thm:iterative}, stated for every $r\ge2$ and for
arbitrary hull codimension, together with the reduction of the admissible extension vectors to the
forced subspace $\im(A)$ of Theorem~\ref{thm:forced-span}. Both are set in the situation treated
here, in which columns are appended to a parity-check matrix while the dimension $n-k$ is held
fixed, and both lie beyond the one appended coordinate covered by \cite{Ezerman2026GE}.
\end{remark}

\noindent\textbf{Invariance of the quantities involved.} 
The quantity $d_{\mathrm{aff}}(\mathbf{w})$ is intrinsic. Under $H\mapsto UH$ with $U$ invertible one has $A\mapsto UAU^{\dagger}$ (rank preserved), $\mathbf{w}\mapsto U\mathbf{w}$, and $\{\mathbf{x}:UH\mathbf{x}^{T}=U\mathbf{w}\}=\{\mathbf{x}:H\mathbf{x}^{T}=\mathbf{w}\}$. Likewise \eqref{eq:exactcrit} is unchanged by the substitution $\mathbf{w}\mapsto c\mathbf{w}$, $c\in\F_{q^2}^{*}$, and the rank test \eqref{eq:C1} is basis-free.

\begin{remark}\label{rem:C2}
We record the following natural screening condition, to which we refer as \textup{(C2)}:
\begin{equation}
\label{eq:C2}
d\bigl(\mathrm{NullSpace}(H_S^{(q)})\bigr)\ge d-1,
\qquad S=\mathrm{supp}(\mathbf{w}),
\end{equation}
where $H_S$ is the submatrix of $H$ indexed by $S$. Condition \eqref{eq:C2} is retained only as a fast \emph{screening heuristic}. It is \textbf{basis-dependent}: under $H\mapsto UH$ the support $S=\mathrm{supp}(\mathbf{w})$ moves and $d(\mathrm{NullSpace}(H_S^{(q)}))$ can change arbitrarily, while $d(\mathcal{C}_{\mathrm{ext}}^{\perp_h})$ is fixed by Theorem~\ref{thm:onecoord}. It is \textbf{over-strong}: among the $2187$ cyclic codes of length $51$ over $\F_4$ with $\Rank(A)=1$, only $189$ ($8.6\%$) satisfy \textup{(C2)} in the default basis, although the overwhelming majority preserve the distance. And it \textbf{admits false positives}, see Example~\ref{ex:falsepos}. The structural reason is that \textup{(C2)} inspects the null space $\{\mathbf{x}:H_S^{(q)}\mathbf{x}^{T}=\mathbf{0}\}$, whereas failure is always witnessed by a low-weight vector of the affine coset $\{\mathbf{x}:H\mathbf{x}^{T}=\mathbf{w}\}$, which no nullspace condition can see.
\end{remark}

\begin{example}\label{ex:falsepos}
The $[51,38,6]_4$ cyclic code with defining-set representatives $\{0,1,2,5\}$ has $\Rank(A)=1$ and collapses: $d(\mathcal{C}_{\mathrm{ext}}^{\perp_h})=4<6$ (the affine coset $\{\mathbf{x}:H\mathbf{x}^{T}=\mathbf{w}\}$ contains a weight-$3$ vector, so $d_{\mathrm{aff}}(\mathbf{w})=3$ and \eqref{eq:exactd} gives $4$). Nevertheless there exists a basis in which \textup{(C2)} holds. Hence the nullspace test can neither certify preservation in general nor be necessary for it.
\end{example}

\begin{corollary}\label{cor:sufficient_not_necessary}
The criterion \eqref{eq:exactcrit} of Theorem~\ref{thm:onecoord} is both necessary and sufficient for $d(\mathcal{C}_{\mathrm{ext}}^{\perp_h})=d$. Condition \textup{(C2)}, evaluated in any fixed basis, is neither: it is not sufficient by Example~\ref{ex:falsepos} and not necessary.
\end{corollary}

\begin{corollary}
Let $\Rank(A)=1$. Then $\Rank(A+\mathbf{w}\mathbf{w}^{\dagger})=0$ holds for exactly $q+1$ vectors $\mathbf{w}$, namely the $q+1$ solutions $\mathbf{w}=k\mathbf{u}_1$ of $N(k)=-1$. These $q+1$ vectors are pairwise proportional and therefore form a \emph{single} projective class.
\end{corollary}

\begin{proof}
By \eqref{eq:rank1decomp}, $A=\mathbf{u}_1\mathbf{u}_1^{\dagger}$ and $A+\mathbf{w}\mathbf{w}^{\dagger}=0$ requires $\mathbf{w}=k\mathbf{u}_1$ with $N(k)=-1$. The equation $k^{q+1}=-1$ has exactly $q+1$ solutions $k\in\F_{q^2}^{*}$ (a coset of the kernel of the norm map, which has order $q+1$), and distinct $k$ give distinct vectors $k\mathbf{u}_1$; any two of these vectors are scalar multiples of one another, so the $q+1$ solutions constitute a single projective class.
\end{proof}

\noindent\textbf{An implementation note.} The rank reduction \eqref{eq:C1} requires no case
distinction, since Proposition~\ref{prop:congruence} and Lemma~\ref{lem:explicit} apply to every
Hermitian $A$ and supply rank-reducing vectors constructively. In Theorem~\ref{thm:onecoord} the
extension vector $\mathbf{w}$ is an arbitrary element of $\F_{q^2}^{n-k}$; candidates produced by
Lemma~\ref{lem:explicit} are guaranteed to reduce the rank by one, and among them the implementation
scans in increasing Hamming weight order and keeps the first that satisfies the exact criterion
\eqref{eq:exactcrit}. Condition \textup{(C2)} may be used as a pre-filter, see Remark~\ref{rem:C2},
but it has no logical force; low-weight candidates are preferred only because $\wt(\mathbf{w})$ enters
the cost of evaluating $d_{\mathrm{aff}}$ in the direct computation used for the small instances.
Every accepted step is certified a posteriori by \eqref{eq:exactcrit}, never by a screening condition.

\subsection{Iterated extension}

The one-coordinate result of Theorem~\ref{thm:onecoord} iterates. After $i$ extensions the new coordinate vector $\mathbf{x}$ is constrained by a \emph{linear combination} of the appended columns, so the criterion must be applied to the span of the appended vectors rather than to each of them separately. The following statement is the exact form of the iteration.

\begin{theorem}\label{thm:iterative}
Keep the notation of \eqref{eq:setup}, let $\mathbf{w}_1,\ldots,\mathbf{w}_i\in\F_{q^2}^{n-k}$ be linearly independent, put $H_i=[H\mid\mathbf{w}_1\mid\cdots\mid\mathbf{w}_i]$ and $\mathcal{S}_i=\langle H_i\rangle\subseteq\F_{q^2}^{n+i}$, and let
$$
W_i=\langle\mathbf{w}_1,\ldots,\mathbf{w}_i\rangle\subseteq\F_{q^2}^{n-k},\qquad
\omega_{W_i}(\boldsymbol{\mu})=\min\bigl\{\wt(\mathbf{t}):\mathbf{t}\in\F_{q^2}^{i},\ \textstyle\sum_{j\le i}t_j^{q}\mathbf{w}_j=\boldsymbol{\mu}\bigr\}.
$$
Then
\begin{equation}
\label{eq:iterated}
d\bigl(\mathcal{S}_i^{\perp_h}\bigr)
=\min\Bigl\{\,d,\ \min_{\mathbf{0}\ne\boldsymbol{\mu}\in W_i}\bigl[d_{\mathrm{aff}}(\boldsymbol{\mu})+\omega_{W_i}(\boldsymbol{\mu})\bigr]\Bigr\}.
\end{equation}
In particular $d(\mathcal{S}_i^{\perp_h})=d$ if and only if $d_{\mathrm{aff}}(\boldsymbol{\mu})+\omega_{W_i}(\boldsymbol{\mu})\ge d$ for every nonzero $\boldsymbol{\mu}\in W_i$; a sufficient condition is $d_{\mathrm{aff}}(\boldsymbol{\mu})\ge d-1$ for every nonzero $\boldsymbol{\mu}\in W_i$.
\end{theorem}

\begin{proof}
Write $\mathbf{c}=(\mathbf{x},\mathbf{t})\in\F_{q^2}^{n+i}$. Orthogonality to the rows of $H_i$ reads
$$
H\,\mathbf{x}^{(q)T}+\sum_{j\le i}t_j^{q}\mathbf{w}_j=\mathbf{0},
\qquad\text{i.e.}\qquad
H\,\mathbf{x}^{(q)T}=\boldsymbol{\mu}(\mathbf{t}):=-\sum_{j\le i}t_j^{q}\mathbf{w}_j .
$$
If $\mathbf{t}=\mathbf{0}$ then $\boldsymbol{\mu}=\mathbf{0}$ and $\mathbf{x}^{(q)}\in\mathcal{C}$, giving $\wt(\mathbf{c})\ge d$; conversely a word of weight $d$ arises this way. If $\mathbf{t}\ne\mathbf{0}$ then $\boldsymbol{\mu}(\mathbf{t})\ne\mathbf{0}$ by the linear independence of the $\mathbf{w}_j$, and $\mathbf{x}^{(q)}$ lies in the affine space with syndrome $\boldsymbol{\mu}(\mathbf{t})$, whose minimum weight is $d_{\mathrm{aff}}(\boldsymbol{\mu}(\mathbf{t}))$, while $\wt(\mathbf{t})\ge\omega_{W_i}(\boldsymbol{\mu}(\mathbf{t}))$ because $\omega_{W_i}$ is invariant under $\boldsymbol{\mu}\mapsto-\boldsymbol{\mu}$; hence $\wt(\mathbf{c})\ge d_{\mathrm{aff}}(\boldsymbol{\mu})+\omega_{W_i}(\boldsymbol{\mu})$ for $\boldsymbol{\mu}=\boldsymbol{\mu}(\mathbf{t})$, and equality is attained by choosing $\mathbf{t}$ of minimal weight and a minimum-weight representative of the affine coset. The map $\mathbf{t}\mapsto\boldsymbol{\mu}(\mathbf{t})$ is $\sigma$-semilinear and injective by the independence of the $\mathbf{w}_j$, hence a bijection $\F_{q^2}^{i}\to W_i$, so taking the minimum over $\mathbf{t}\ne\mathbf{0}$ is the same as taking the minimum over $\mathbf{0}\ne\boldsymbol{\mu}\in W_i$, and \eqref{eq:iterated} follows. The last assertion is immediate from \eqref{eq:iterated} together with $\omega_{W_i}(\boldsymbol{\mu})\ge1$.
\end{proof}

\noindent\textbf{Exhaustiveness of the check.} 
For $i=1$ one has $W_1=\langle\mathbf{w}_1\rangle$, $\omega_{W_1}(\boldsymbol{\mu})=1$ for $\boldsymbol{\mu}\ne\mathbf{0}$, and \eqref{eq:iterated} reduces to \eqref{eq:exactd}; Theorem~\ref{thm:iterative} therefore specialises correctly to the one-coordinate case. Checking $d_{\mathrm{aff}}(\mathbf{w}_j)\ge d-1$ for each generator $\mathbf{w}_j$ separately is \emph{not} sufficient for $i\ge2$: one must control every nonzero element of $W_i$. Since $|W_i|=q^{2i}-1$ is small for the small hull codimensions of Section~\ref{sec:newcodes}, this is an exhaustive but entirely feasible verification. For the largest hull codimension treated there, $r=10$, the admissible vectors are confined to the $r$-dimensional space $\im(A)$ by Theorem~\ref{thm:forced-span}, and the distance is certified directly rather than through \eqref{eq:iterated}.

\begin{theorem}\label{thm:quantum_code}
Let $\mathcal{C}=[n,k,d]_{q^2}$ have parity-check matrix $H$ with $A=HH^{\dagger}$ of rank $r$. If the extension vectors $\mathbf{w}_1,\ldots,\mathbf{w}_r$ satisfy the rank-reducing condition \eqref{eq:C1} at every step and the exactness condition of Theorem~\ref{thm:iterative}, then there exists a quantum code
$$
\mathcal{Q}=\bigl[\!\bigl[n+r,\ 2k-n+r,\ d\bigr]\!\bigr]_q .
$$
\end{theorem}

\begin{proof}
By Theorem~\ref{Iterative Self-Orthogonal Construction} the matrix $H_r\in\F_{q^2}^{(n-k)\times(n+r)}$ satisfies $H_rH_r^{\dagger}=0$, so $\mathcal{S}_r=\langle H_r\rangle$ is a Hermitian self-orthogonal $[n+r,\,n-k]_{q^2}$ code. Its Hermitian dual has dimension $(n+r)-(n-k)=k+r$, and Theorem~\ref{thm:iterative} gives $d(\mathcal{S}_r^{\perp_h})=d$. Applying Corollary~\ref{construction} to $\mathcal{S}_r$ yields a quantum code with parameters $[\![n+r,\,(n+r)-2(n-k),\,d]\!]_q=[\![n+r,\,2k-n+r,\,d]\!]_q$.
\end{proof}

\subsection{The admissible extension vectors and the forced span}\label{sec:forced-span}

The criterion \eqref{eq:exactcrit} decides whether a given vector $\mathbf{w}$ preserves the distance,
but says nothing about which vectors can carry the extension at all.
Corollary~\ref{cor:iterative_reduction} shows that an extension vector must reduce $\Rank(A)$ by one,
and the implementation of Section~\ref{sec:newcodes} looks for such vectors among the columns of a
basis of the whole space $\F_{q^2}^{n-k}$, as described after Corollary~\ref{cor:iterative_reduction}. The next three statements
show that this is more freedom than the construction admits.

\begin{lemma}\label{lem:rankdrop-iff}
Let $A\in\F_{q^2}^{m\times m}$ be Hermitian of rank $r\ge1$, let
$A=P^{\dagger}\,\mathrm{diag}(I_r,0_{m-r})\,P$ be a congruence normal form as in
\eqref{eq:congruence}, and let $\mathbf{w}\in\F_{q^2}^{m}$. Put
\begin{equation}
\label{eq:norm-A}
\langle\mathbf{w},\mathbf{w}\rangle_{A}
:=\bigl(P^{-\dagger}\mathbf{w}\bigr)^{\dagger}\bigl(P^{-\dagger}\mathbf{w}\bigr).
\end{equation}
Then exactly one of the following holds:
\begin{enumerate}
\item[(a)] $\Rank(A+\mathbf{w}\mathbf{w}^{\dagger})=\Rank(A)-1$, if and only if
$\mathbf{w}\in\im(A)$ and $\langle\mathbf{w},\mathbf{w}\rangle_{A}=-1$;
\item[(b)] $\Rank(A+\mathbf{w}\mathbf{w}^{\dagger})=\Rank(A)+1$, if and only if
$\mathbf{w}\notin\im(A)$;
\item[(c)] $\Rank(A+\mathbf{w}\mathbf{w}^{\dagger})=\Rank(A)$ otherwise.
\end{enumerate}
In particular the rank-reducing test \eqref{eq:C1} is decided by two algebraic conditions,
membership in the $r$-dimensional space $\im(A)$ and the value $-1$ of a Hermitian form on it,
and requires no rank computation at all.
\end{lemma}

\begin{proof}
Congruence preserves rank, so we may compute in the coordinates $\mathbf{x}=P\mathbf{y}$ of
\eqref{eq:congruence}, in which $A$ becomes $D=\mathrm{diag}(I_r,0_{m-r})$ and $\mathbf{w}$ becomes
$\widetilde{\mathbf{w}}=P^{-\dagger}\mathbf{w}$; the rank-one term becomes
$\widetilde{\mathbf{w}}\widetilde{\mathbf{w}}^{\dagger}$, so the whole expression is again a rank-one
perturbation of a diagonal matrix. Write
$\widetilde{\mathbf{w}}=(\mathbf{a};\mathbf{b})$ with $\mathbf{a}\in\F_{q^2}^{r}$ and
$\mathbf{b}\in\F_{q^2}^{m-r}$. Since
$P^{-\dagger}\im(A)=\mathrm{span}(\mathbf{e}_1,\ldots,\mathbf{e}_r)$, the condition
$\mathbf{w}\in\im(A)$ is exactly $\mathbf{b}=\mathbf{0}$.

If $\mathbf{b}\neq\mathbf{0}$, let $\mathbf{y}=(\mathbf{y}_1;\mathbf{y}_2)$ satisfy
$\bigl(D+\widetilde{\mathbf{w}}\widetilde{\mathbf{w}}^{\dagger}\bigr)\mathbf{y}=\mathbf{0}$.
Writing $\lambda=\widetilde{\mathbf{w}}^{\dagger}\mathbf{y}$, this reads
$(\mathbf{y}_1;\mathbf{0})=-\lambda\widetilde{\mathbf{w}}$; the second block gives
$\mathbf{0}=-\lambda\mathbf{b}$, hence $\lambda=0$ because $\mathbf{b}\neq\mathbf{0}$, the first block
then gives $\mathbf{y}_1=\mathbf{0}$, and the only remaining condition is
$\lambda=\mathbf{b}^{\dagger}\mathbf{y}_2=0$. Thus
$\ker\bigl(D+\widetilde{\mathbf{w}}\widetilde{\mathbf{w}}^{\dagger}\bigr)
=\{\mathbf{0}\}\oplus\{\mathbf{y}_2:\mathbf{b}^{\dagger}\mathbf{y}_2=0\}$ has dimension $m-r-1$, and
the rank is $r+1$; this gives~(b).

If $\mathbf{b}=\mathbf{0}$, then
$D+\widetilde{\mathbf{w}}\widetilde{\mathbf{w}}^{\dagger}
=\mathrm{diag}(I_r+\mathbf{a}\mathbf{a}^{\dagger},0)$, whose rank equals
$\Rank(I_r+\mathbf{a}\mathbf{a}^{\dagger})$. By the matrix determinant lemma
$\det(I_r+\mathbf{a}\mathbf{a}^{\dagger})=1+N(\mathbf{a})$, and $I_r+\mathbf{a}\mathbf{a}^{\dagger}$
acts as the identity on the hyperplane $\{\mathbf{a}\}^{\perp}$; hence its rank is $r-1$ exactly when
$N(\mathbf{a})=-1$, and $r$ otherwise (the case $\mathbf{a}=\mathbf{0}$ giving $N(\mathbf{a})=0$).
Finally
$N(\mathbf{a})=\widetilde{\mathbf{w}}^{\dagger}\widetilde{\mathbf{w}}
=\langle\mathbf{w},\mathbf{w}\rangle_{A}$ by \eqref{eq:norm-A}, which gives (a) and~(c).
\end{proof}

\noindent\textbf{A normal-form-free computation of $\langle\mathbf{w},\mathbf{w}\rangle_{A}$.} 
The form of \eqref{eq:norm-A} is intrinsic, and it is computable without any normal form: for
$\mathbf{w}\in\im(A)$ choose any $\mathbf{x}$ with $A\mathbf{x}=\mathbf{w}$; then
$$
\langle\mathbf{w},\mathbf{w}\rangle_{A}=\mathbf{x}^{\dagger}A\mathbf{x}.
$$
Indeed $\mathbf{w}=P^{\dagger}(\mathbf{a};\mathbf{0})$ forces
$\mathbf{x}=P^{-1}(\mathbf{a};\mathbf{c})$ for some $\mathbf{c}$, whence
$\mathbf{x}^{\dagger}A\mathbf{x}=(\mathbf{a};\mathbf{c})^{\dagger}D(\mathbf{a};\mathbf{c})
=\mathbf{a}^{\dagger}\mathbf{a}$, which is \eqref{eq:norm-A}. The value does not depend on the choice
of $\mathbf{x}$: if $A\mathbf{x}=A\mathbf{x}'$ then
$\mathbf{x}-\mathbf{x}'\in\ker A=(\im A)^{\perp}$, so
$\mathbf{x}^{\dagger}A(\mathbf{x}-\mathbf{x}')=\mathbf{w}^{\dagger}(\mathbf{x}-\mathbf{x}')=0$ and
$(\mathbf{x}-\mathbf{x}')^{\dagger}A(\mathbf{x}-\mathbf{x}')=0$. Equivalently
$\langle\mathbf{w},\mathbf{w}\rangle_{A}=\mathbf{w}^{\dagger}A^{+}\mathbf{w}$, with $A^{+}$ the
Moore--Penrose inverse of the Hermitian form $A$. Lemma~\ref{lem:explicit} is the special case
$\mathbf{w}=k\mathbf{u}_i$, for which $\langle\mathbf{w},\mathbf{w}\rangle_{A}=N(k)$; and (a) also
shows that for an \emph{isotropic} vector the spaces $\im(A)$ and $\ker A$
meet. Over $\F_4$ one has $N(c)=c^{3}=1$
for every scalar $c\in\F_4^{*}$, so the second condition of Lemma~\ref{lem:rankdrop-iff}(a) is
vacuous for $r=1$; in general it requires the coordinate vector $\mathbf{a}$ of the congruence
coordinates \eqref{eq:congruence} to have odd Hamming weight, since
$\langle\mathbf{w},\mathbf{w}\rangle_{A}=\sum_{i}N(a_i)\equiv\wt(\mathbf{a})\pmod 2$.

\begin{theorem}\label{thm:forced-span}
Let $\mathcal{C}=[n,k,d]_{q^2}$ have parity-check matrix $H\in\F_{q^2}^{(n-k)\times n}$, put
$A=HH^{\dagger}$ with $\Rank(A)=r$, and let $\mathbf{w}_1,\ldots,\mathbf{w}_r$ satisfy the rank-reducing
condition \eqref{eq:C1} at every step, that is,
$\Rank(A^{[j]})=\Rank(A^{[j-1]})-1$ for
$A^{[j]}=A+\sum_{i\le j}\mathbf{w}_i\mathbf{w}_i^{\dagger}$ and $A^{[0]}=A$. Then:
\begin{enumerate}
\item[(i)] $\im(A^{[j]})\subseteq\im(A^{[j-1]})$ and $\dim\im(A^{[j]})=r-j$ for every $j$, so that
$$
\im(A)=\im(A^{[0]})\supsetneq\im(A^{[1]})\supsetneq\cdots\supsetneq\im(A^{[r]})=0
$$
is a strictly decreasing chain of subspaces; in particular each extension vector lies in the
$r$-dimensional space $\im(A)$, and not, as one might expect, in the whole space
$\F_{q^2}^{n-k}$ of dimension $n-k$;
\item[(ii)] $\mathbf{w}_1,\ldots,\mathbf{w}_r$ are linearly independent and
$\langle\mathbf{w}_1,\ldots,\mathbf{w}_r\rangle=\im(A)$.
\end{enumerate}
Consequently the extension subspace is \emph{forced} by $H$ alone: the construction has no freedom in
the choice of subspace, and its only remaining freedom is the choice of an ordered basis
$(\mathbf{w}_1,\ldots,\mathbf{w}_r)$ of $\im(A)$ along which the $r$-fold extension
$\langle\mathbf{w}_j,\mathbf{w}_j\rangle_{A^{[j-1]}}=-1$ of Lemma~\ref{lem:rankdrop-iff}(a) is
satisfied.
\end{theorem}

\begin{proof}
(i) By Lemma~\ref{lem:rankdrop-iff}(a) applied to the Hermitian matrix $A^{[j-1]}$, whose rank is
$r-j+1$, we have $\mathbf{w}_j\in\im(A^{[j-1]})$ and
$\langle\mathbf{w}_j,\mathbf{w}_j\rangle_{A^{[j-1]}}=-1$. Choose a congruence
$A^{[j-1]}=P^{\dagger}\mathrm{diag}(I_{r-j+1},0)P$. In those coordinates
$\widetilde{\mathbf{w}}_j=(\mathbf{a};\mathbf{0})$ for some $\mathbf{a}\in\F_{q^2}^{r-j+1}$ with
$N(\mathbf{a})=-1$, and
$$
A^{[j]}\mapsto\mathrm{diag}\bigl(I_{r-j+1}+\mathbf{a}\mathbf{a}^{\dagger},\,0\bigr),
$$
whose image is the hyperplane $\{\mathbf{a}\}^{\perp}$ of the leading block. Hence $\im(A^{[j]})$ is a
subspace of $\im(A^{[j-1]})$ of dimension $r-j$; since $\im(A^{[0]})=\im(A)$, this proves (i) by
induction.
(ii) From $A^{[r]}=0$ we obtain $\sum_{j\le r}\mathbf{w}_j\mathbf{w}_j^{\dagger}=-A$. Writing
$W=[\mathbf{w}_1\cdots\mathbf{w}_r]$ this reads $WW^{\dagger}=-A$, whence
$$
\dim\langle\mathbf{w}_1,\ldots,\mathbf{w}_r\rangle
=\Rank(WW^{\dagger})=\Rank(A)=r,
$$
while (i) gives $\langle\mathbf{w}_1,\ldots,\mathbf{w}_r\rangle\subseteq\im(A)$. Two subspaces of
equal finite dimension, one contained in the other, coincide; hence
$\langle\mathbf{w}_1,\ldots,\mathbf{w}_r\rangle=\im(A)$, and the $\mathbf{w}_j$ form a basis of it.
\end{proof}

\begin{corollary}\label{cor:canonical-search}
Keep the notation of Theorem~\ref{thm:forced-span}.
\begin{enumerate}
\item[(i)] If $d_{\mathrm{aff}}(\boldsymbol{\mu})\ge d-1$ for every nonzero
$\boldsymbol{\mu}\in\im(A)$, then $d(\mathcal{S}_r^{\perp_h})=d$. This sufficient condition involves
no basis: it is stated on the canonical space $\im(A)$, whereas the general criterion
\eqref{eq:iterated} is stated on the basis-dependent span $W_i$.
\item[(ii)] Every candidate of every step of the $r$-fold extension lies in the $r$-dimensional space $\im(A)$,
so the search space of the construction is $q^{2r}$ rather than $q^{2(n-k)}$, a reduction by the
factor $q^{2(n-k-r)}$, and the admissible vectors are exactly the solutions of the two conditions
of Lemma~\ref{lem:rankdrop-iff}(a).
\end{enumerate}
\end{corollary}

\begin{proof}
(i) By Theorem~\ref{thm:forced-span}(ii) we have $W_r=\im(A)$, and the sufficient condition of
Theorem~\ref{thm:iterative} requires $d_{\mathrm{aff}}(\boldsymbol{\mu})\ge d-1$ for every nonzero
$\boldsymbol{\mu}\in W_r$. (ii) is immediate from Lemma~\ref{lem:rankdrop-iff}(a) together with
Theorem~\ref{thm:forced-span}(i).
\end{proof}

\noindent\textbf{Why the sufficient condition is not the whole story.} 
The reduction of Corollary~\ref{cor:canonical-search}(i) is sufficient but not necessary, and the
record $[[61,25,9]]_2$ of Section~\ref{sec:newcodes} exhibits the difference. There the ingredient has
$d=10$ while the extension delivers $d'=9$; by \eqref{eq:iterated} some nonzero
$\boldsymbol{\mu}\in\im(A)$ satisfies
$d_{\mathrm{aff}}(\boldsymbol{\mu})+\omega_{W_r}(\boldsymbol{\mu})=9$, whence
$d_{\mathrm{aff}}(\boldsymbol{\mu})\le8<9=d-1$, so the hypothesis of
Corollary~\ref{cor:canonical-search}(i) fails and that reduction alone cannot account for the record.
It is the balance between $d_{\mathrm{aff}}$ and $\omega$ in \eqref{eq:iterated} that makes the record
exact. The content of Theorem~\ref{thm:forced-span} is therefore not that the $r$-fold extension becomes trivial,
but that its freedom is confined to a single choice, the ordered basis of $\im(A)$, and that the
candidates at step $j$ live in the nested chain
$\im(A)\supsetneq\im(A^{[1]})\supsetneq\cdots$ of Theorem~\ref{thm:forced-span}(i), rather than in
$\F_{q^2}^{n-k}$. Whether that basis can be chosen canonically, so that the remaining optimisation
disappears, is left open in Subsection~\ref{sec:open-problems}.

\section{Defining sets, the coset involution, and an infinite family}\label{sec:coset-involution}

\subsection{The hull codimension of a cyclic code}\label{sec:r-formula}

The integer $r=\Rank(HH^{\dagger})$ fixes the number of columns that the extension has to append, and for the cyclic codes that supply the examples of Section~\ref{sec:newcodes} it can be read off directly from the defining set. The purpose of this subsection is to prove the corresponding formula and to isolate the involution on the $q^{2}$-cyclotomic cosets of $\mathbb{Z}_n$ that controls the answer.

Throughout this section $q$ is a prime power, $p=\mathrm{char}\,\F_q$, $n\ge1$ with $\gcd(n,q)=1$, and $\alpha$ denotes a primitive $n$-th root of unity in a splitting field of $x^{n}-1$ over $\F_{q^{2}}$; all indices are taken modulo $n$. For a subset $S\subseteq\mathbb{Z}_n$ and $a\in\mathbb{Z}_n$ we write $aS=\{a\nu:\nu\in S\}$ and $-S=\{-1\cdot\nu:\nu\in S\}$.

\begin{definition}\label{def:r}
For a linear code $\mathcal{C}=[n,k,d]_{q^{2}}$ with a parity-check matrix $H\in\F_{q^{2}}^{(n-k)\times n}$ of full row rank, the \emph{hull codimension} of $\mathcal{C}$ is the codimension of the Hermitian hull
$\mathcal{C}\cap\mathcal{C}^{\perp_{h}}$ in $\mathcal{C}^{\perp_{h}}$, that is,
\begin{equation}
\label{eq:def-r}
r:=\Rank\bigl(HH^{\dagger}\bigr)=(n-k)-\dim_{\F_{q^{2}}}\bigl(\mathcal{C}^{\perp_{h}}\cap\mathcal{C}\bigr).
\end{equation}
Thus $r$ measures by how much $\mathcal{C}^{\perp_{h}}$ fails to be Hermitian self-orthogonal: $r=0$ if and only if $\mathcal{C}$ is Hermitian dual-containing, that is, $\mathcal{C}^{\perp_{h}}\subseteq\mathcal{C}$.
The two expressions in \eqref{eq:def-r} agree: for matrices of compatible sizes $\Rank(HH^{\dagger})=\Rank(H^{\dagger})-\dim(\im H^{\dagger}\cap\Ker H)$, and $\im H^{\dagger}=\mathcal{C}^{\perp_h}$ because $\langle H^{\dagger}\mathbf{u},\mathbf{x}\rangle_h=0$ for every $\mathbf{x}\in\Ker H$; in particular $r$ does not depend on the choice of $H$, and the hull dimension is $(n-k)-r$.
\end{definition}

\begin{theorem}\label{thm:r-formula}
Let $\mathcal{C}=\langle g(x)\rangle$ be a cyclic code of length $n$ over $\F_{q^{2}}$ with defining set $T\subseteq\mathbb{Z}_n$, a union of $q^{2}$-cyclotomic cosets, so that $\mathcal{C}=[n,\,n-|T|]_{q^{2}}$. Let
\begin{equation}
\label{eq:tau}
\tau:\mathbb{Z}_n\longrightarrow\mathbb{Z}_n,\qquad \tau(\nu):=-q\nu,
\end{equation}
and let $H=\bigl(\alpha^{\nu e}\bigr)_{\nu\in T,\,e\in\mathbb{Z}_n}$ be the associated parity-check matrix. Then the following hold.
\begin{enumerate}
\item[(i)] \emph{Explicit Gram matrix.}
\begin{equation}
\label{eq:A-partial-perm}
HH^{\dagger}=(n\bmod p)\,P,\qquad
P_{\nu\mu}=
\begin{cases}1,&\nu+q\mu\equiv 0 \pmod n,\\[2pt] 0,&\text{otherwise},\end{cases}
\end{equation}
so that $P\in\F_p^{\,T\times T}$ is a \emph{partial permutation matrix}: every row and every column of $P$ contains at most one $1$.
\item[(ii)] \emph{Defining-set formula.}
\begin{equation}
\label{eq:r-formula}
r=\bigl|T\cap(-qT)\bigr|=\#\{\nu\in T:\tau(\nu)\in T\}
=\sum_{\substack{C\subseteq T\\[1pt] \tau(C)\subseteq T}}|C|,
\end{equation}
the last sum ranging over the $q^{2}$-cyclotomic cosets $C$ contained in $T$ for which $\tau(C)\subseteq T$ as well.
\item[(iii)] \emph{Coset involution.} $\tau$ permutes the $q^{2}$-cyclotomic cosets of $\mathbb{Z}_n$ and is an involution on them, since $\tau^{2}(\nu)=q^{2}\nu$ and $q^{2}C=C$ for every coset $C$. Accordingly the cosets split into $\tau$-fixed cosets and two-element orbits $\{C,\tau(C)\}$, and a coset is $\tau$-fixed if and only if it is \emph{skew symmetric}, i.e.\ $-qC=C$. At the level of residues,
\begin{equation}
\label{eq:fix-tau}
\mathrm{Fix}(\tau)=\{\nu\in\mathbb{Z}_n:\tau(\nu)=\nu\}=\{\nu\in\mathbb{Z}_n:\nu(q+1)\equiv0\},
\end{equation}
whereas the union of the $\tau$-fixed cosets is the larger set
\begin{equation}
\label{eq:fix-cosets}
\bigl\{\nu\in\mathbb{Z}_n:\ \nu q\bigl(q^{2j-1}+1\bigr)\equiv0\ \text{ for some } j\ge0\bigr\},
\end{equation}
the term $j=0$ reproducing \eqref{eq:fix-tau}.
\item[(iv)] \emph{Critical values.} $r=0$ if and only if $T\cap\tau(T)=\varnothing$. Moreover $r=1$ if and only if $T\cap\tau(T)$ consists of a single element $\nu_0$; in that case $\nu_0$ is necessarily $\tau$-fixed, $\tau(\nu_0)=\nu_0$, and it generates a \emph{singleton} coset $C_{\nu_0}=\{\nu_0\}$.
\end{enumerate}
\end{theorem}

\begin{proof}
(i) For $\nu,\mu\in T$,
\[
\bigl(HH^{\dagger}\bigr)_{\nu\mu}
=\sum_{e\in\mathbb{Z}_n}\alpha^{\nu e}\bigl(\alpha^{\mu e}\bigr)^{q}
=\sum_{e\in\mathbb{Z}_n}\bigl(\alpha^{\nu+q\mu}\bigr)^{e}
=\begin{cases} n,&\nu+q\mu\equiv0,\\ 0,&\text{otherwise},\end{cases}
\]
the second case being a geometric series with ratio $\alpha^{\nu+q\mu}\neq1$ (if $\alpha^{\nu+q\mu}=1$ then $n\mid\nu+q\mu$, because $\alpha$ has order exactly $n$). Reducing $n$ modulo $p$ gives \eqref{eq:A-partial-perm}. For fixed $\nu$ the congruence $\nu+q\mu\equiv0$ has at most one solution $\mu$ modulo $n$, and symmetrically for fixed $\mu$, because $\gcd(q,n)=1$; hence $P$ is a partial permutation matrix.

(ii) For a partial permutation matrix the rank equals the number of its $1$'s, so by \eqref{eq:A-partial-perm}
$r=\#\{\nu\in T:\,-q^{-1}\nu\in T\}$. As $T$ is invariant under multiplication by $q^{2}$, we have $-q^{-1}T=(-q^{-1}q^{2})T=(-q)T$, so the two cardinalities agree and the first equality of \eqref{eq:r-formula} follows. For the second, note that $\nu\in T\cap\tau(T)$ if and only if the coset $C_{\nu}$ satisfies $C_{\nu}\subseteq T$ and $\tau(C_{\nu})\subseteq T$: the forward implication holds because $T$ is a union of cosets and $\tau$ maps cosets to cosets, and the backward implication is immediate. Hence $T\cap\tau(T)=\bigcup\{C:C\subseteq T,\ \tau(C)\subseteq T\}$ is a disjoint union of whole cosets, which gives \eqref{eq:r-formula}.

(iii) $\tau^{2}(\nu)=q^{2}\nu$ lies in the same $q^{2}$-cyclotomic coset as $\nu$, so $\tau^{2}$ induces the identity on the coset set and $\tau$ is an involution there. A coset is fixed precisely when $-qC=C$. Equation \eqref{eq:fix-tau} is the definition of $\mathrm{Fix}(\tau)$ combined with $\gcd(q,n)=1$, and \eqref{eq:fix-cosets} is obtained by writing the condition $C_{\nu}\ni-q\nu$, i.e.\ $-q\nu\equiv q^{2j}\nu$ for some $j\ge0$, and multiplying by $q^{-1}$.

(iv) By (ii), $r=|T\cap\tau(T)|$ and $T\cap\tau(T)$ is a union of cosets. It is empty if and only if no coset contributes, which is $r=0$. If $|T\cap\tau(T)|=1$, then $T\cap\tau(T)$ is a union of cosets of total size one and therefore equals $C_{\nu_0}=\{\nu_0\}$ for a singleton coset; moreover $\tau(\nu_0)\in T\cap\tau(T)=\{\nu_0\}$, so $\tau(\nu_0)=\nu_0$. The converse is immediate from (ii).
\end{proof}

Statements (ii)--(iv) are the Hermitian counterparts of the Euclidean reading of the hull from a
defining set. There the same argument, in which the defining set of the hull is $T\cup(\mathbb{Z}_n\setminus
T^{-1})$, gives $\dim(\mathrm{Hull}(\mathcal{C}))=|T|-|T\cap T^{-1}|$; the extremal values
$\dim(\mathrm{Hull})=k-1$ and $k^{\perp}-1$ are characterised by $T\cup T^{-1}=\mathbb{Z}_n\setminus\{a\}$
and $T\cap T^{-1}=\{a\}$ with $a$ in a singleton coset \cite[Thms.~4 and 5]{Du2020HullsTIT}, and the
general value $k-\ell$ by $T\cup T^{-1}=\mathbb{Z}_n\setminus A$ with $|A|=\ell$
\cite[Thm.~3.2]{Lei2022Hulls}. The involution there is $\nu\mapsto-\nu$, whereas the Hermitian dual
forces $\tau:\nu\mapsto-q\nu$; the objects used below, namely the additivity over $\tau$-orbits with
its parity constraint, the dictionary, the realisability region, the enumerator and the $\tau$-invariant
locus, are not treated in that line.

\noindent\textbf{The meaning of the two critical values.} 
Theorem~\ref{thm:r-formula} computes $\Rank(HH^{\dagger})$ from the defining set alone, without ever forming $H$. Its two critical values have a direct meaning: $r=0$ is the ordinary Hermitian self-orthogonality of $\mathcal{C}^{\perp_{h}}$ (equivalently, $\mathcal{C}$ is Hermitian dual-containing), while $r=1$ means, by (iv), that the defining set meets its own $\tau$-image in exactly one point and that this point is a $\tau$-fixed singleton. The latter case is analysed systematically in Section~\ref{sec:family}; it accounts for the single extension column needed by the codes of the $n=51$ family recorded in \ref{app:vectors}.

\noindent\textbf{A notational caution.} 
We record for the reader's convenience that $r$ equals $|T\cap(-qT)|$ and \emph{not} $\deg g-|T\cap(-qT)|$: the latter quantity is the hull dimension $(n-k)-r=|T|-|T\cap(-qT)|$ of \eqref{eq:def-r}. All computational tables of this paper were produced by direct rank computation and are unaffected.

\subsection{An infinite family of $r=1$ cyclic codes}\label{sec:family}

Assertion (iv) of Theorem~\ref{thm:r-formula} makes the case $r=1$ a purely combinatorial condition on the defining set: one asks for a union $T$ of $q^{2}$-cyclotomic cosets that meets $\tau(T)$ in a single $\tau$-fixed singleton and that contains exactly one coset out of every two-element $\tau$-orbit. Both requirements can be met simultaneously and for infinitely many $n$, because a $\tau$-fixed singleton is produced by any $g>1$ dividing both $n$ and $q+1$. The resulting family is the content of this subsection.

\begin{theorem}\label{thm:family}
Let $q$ be a prime power and let $n\ge3$ satisfy $\gcd(n,q)=1$ and
\begin{equation}
\label{eq:family-g}
g:=\gcd(n,q+1)>1 .
\end{equation}
Put $i_0:=n/g$ and let $\tau$ be as in \eqref{eq:tau}.
\begin{enumerate}
\item[(i)] $C_{i_0}=\{i_0\}$ is a singleton $q^{2}$-cyclotomic coset and $\tau(i_0)=i_0$; equivalently $i_0\in\mathrm{Fix}(\tau)$.
\item[(ii)] Let $T$ be the union of $\{i_0\}$ with \emph{one} coset chosen arbitrarily from each two-element orbit of $\tau$ on the $q^{2}$-cyclotomic cosets of $\mathbb{Z}_n$. Then $T$ is a union of $q^{2}$-cyclotomic cosets with
\begin{equation}
\label{eq:family-T}
T\cap\tau(T)=\{i_0\},\qquad
|T|=1+\tfrac12\bigl(n-f_{\tau}\bigr),\qquad
k=n-|T|=\tfrac12\bigl(n+f_{\tau}-2\bigr),
\end{equation}
where
\begin{equation}
\label{eq:family-ftau}
f_{\tau}:=\sum_{\substack{C\ \tau\text{-fixed}}}|C|
=\#\bigl\{\nu\in\mathbb{Z}_n:\ \text{the }q^{2}\text{-coset of }\nu\text{ is }\tau\text{-fixed}\bigr\},
\end{equation}
and, by Theorem~\ref{thm:r-formula}, the cyclic code $\mathcal{C}=\mathcal{C}(T)$ has hull codimension $r=1$.
\item[(iii)] The Gram matrix is the elementary matrix
\begin{equation}
\label{eq:family-A}
HH^{\dagger}=(n\bmod p)\,E_{i_0,i_0},
\end{equation}
so that $A+\mathbf{w}\mathbf{w}^{\dagger}=0$ for
\begin{equation}
\label{eq:family-w}
\mathbf{w}=\mu\,\mathbf{e}_{i_0},
\qquad
\mu^{q+1}=-(n\bmod p).
\end{equation}
Such a $\mu\in\F_{q^{2}}^{*}$ always exists, because the norm $N\colon\F_{q^{2}}^{*}\to\F_q^{*}$, $N(\mu)=\mu^{q+1}$, is surjective. In particular the one-coordinate extension of $\mathcal{C}$ is Hermitian self-orthogonal, its extension vector being a \emph{unit vector}.
\item[(iv)] Consequently the extension of Theorem~\ref{thm:quantum_code} yields the $q$-ary quantum code
\begin{equation}
\label{eq:family-Q}
\mathcal{Q}=\bigl[\!\bigl[n+1,\ f_{\tau}-1,\ d'\bigr]\!\bigr]_q,
\qquad
d'=\min\bigl\{d(\mathcal{C}),\ d_{\mathrm{aff}}(\mathbf{e}_{i_0})+1\bigr\},
\end{equation}
and the family is infinite: for every prime power $q$ and every prime $\ell\mid q+1$ with $\ell\neq p$, every power of $\ell$ is an admissible length $n$.
\end{enumerate}
\end{theorem}

\begin{proof}
(i) Since $g\mid q+1$ we have $g\mid q^{2}-1$, so $i_0(q^{2}-1)=(n/g)(q^{2}-1)$ is a multiple of $n$; hence $q^{2}i_0\equiv i_0$, i.e.\ $C_{i_0}=\{i_0\}$. Likewise $i_0(q+1)=(n/g)(q+1)$ is a multiple of $n$, so $-qi_0\equiv i_0$, i.e.\ $\tau(i_0)=i_0$.

(ii) $T$ is a union of cosets by construction. Its $\tau$-image omits exactly the cosets not chosen, so $T\cap\tau(T)\subseteq\{i_0\}$; and $i_0\in T\cap\tau(T)$ by (i). An arbitrary set meeting its $\tau$-image in one point consists, as observed in the proof of Theorem~\ref{thm:r-formula}(iv), of a union of whole $\tau$-orbits; hence $\{i_0\}$ is one full orbit and the chosen cosets are one from each of the remaining orbits. The total number of elements covered by the $\tau$-fixed cosets is $f_{\tau}$, so the complement of $T$ consists of the $n-f_{\tau}$ elements lying in two-element orbits, of which exactly half belong to $T$; thus $|T|=1+\tfrac12(n-f_{\tau})$. Finally $r=|T\cap\tau(T)|=1$ by Theorem~\ref{thm:r-formula}(ii).

(iii) By (ii), $T\cap(-qT)=\{i_0\}$, and we saw in the proof of Theorem~\ref{thm:r-formula}(ii) that $T\cap(-qT)=T\cap(-q^{-1}T)$; hence the partial permutation matrix $P$ of \eqref{eq:A-partial-perm} has its unique $1$ at the position $(i_0,i_0)$, which gives \eqref{eq:family-A}. Writing $\mathbf{w}=\mu\mathbf{e}_{i_0}$ we get $\mathbf{w}\mathbf{w}^{\dagger}=N(\mu)\mathbf{e}_{i_0}\mathbf{e}_{i_0}^{\dagger}$, so $A+\mathbf{w}\mathbf{w}^{\dagger}=0$ is exactly the equation $N(\mu)=-(n\bmod p)$, which is solvable in $\F_{q^{2}}^{*}$ because $N$ is surjective onto $\F_q^{*}$. Then $\tilde H=(H\mid\mathbf{w})$ satisfies $\tilde H\tilde H^{\dagger}=A+\mathbf{w}\mathbf{w}^{\dagger}=0$, i.e.\ $\langle\tilde H\rangle$ is Hermitian self-orthogonal of length $n+1$.

(iv) The parameter count is that of Theorem~\ref{thm:quantum_code} with $r=1$, namely $n+r=n+1$ and $2k-n+r=2k-n+1=f_{\tau}-1$ by \eqref{eq:family-T}; the distance is the one-coordinate case of Theorem~\ref{thm:onecoord}. For the infinitude: if $\ell\mid q+1$ is prime and $\ell\neq p$, then for $n$ any power of $\ell$ we have $g=\gcd(n,q+1)\ge\ell>1$, while $\ell\neq p$ forces $\gcd(n,q)=1$; thus all such $n$ are admissible.
\end{proof}

\begin{corollary}\label{cor:family}
In the situation of Theorem~\ref{thm:family} the hull codimension of the extended code exhibits the phenomenon that motivated the search heuristics of Section~\ref{sec:newcodes}: since $d_{\mathrm{aff}}(\mathbf{e}_{i_0})\ge1$ always, one has $d'\ge2$, and $d'=d(\mathcal{C})$ holds whenever $d_{\mathrm{aff}}(\mathbf{e}_{i_0})=d(\mathcal{C})-1$.
\end{corollary}

\noindent\textbf{Cost and rate of the family.} 
The family is produced cheaply: $f_{\tau}$ in \eqref{eq:family-ftau} is obtained by one pass through the $q^{2}$-cyclotomic cosets of $\mathbb{Z}_n$, and the extension vector \eqref{eq:family-w} is a unit vector, so no search over candidate vectors is needed. In return the rate $k_Q/n$ of \eqref{eq:family-Q} is bounded by $f_{\tau}/(n+1)$, which decreases as $n$ grows. The family therefore illustrates the mechanism in the exactly-solvable regime rather than competing with the record codes of Table~\ref{tab:new_quantum_codes}, which are found by search. Its interest is structural: it identifies, for infinitely many $n$, an explicit $r=1$ defining set with an explicit extension vector, and it shows that $r=1$ does not occur only in small examples.

Table~\ref{tab:family} lists the first members of the family, for which the quantities of \eqref{eq:family-Q} were computed exactly by exhaustive enumeration of all $q^{2k}=|\mathcal{C}|$ codewords and of the affine coset in \eqref{eq:daff}.

\begin{table}[htbp]
\caption{Exact parameters of the first members of the family of Theorem~\ref{thm:family}. Here $i_0=n/g$, $f_\tau$ is as in \eqref{eq:family-ftau}, and $d(\mathcal{C})$ and $d_{\mathrm{aff}}(\mathbf{e}_{i_0})$ were obtained by exhaustive enumeration. The last two rows exhibit the strict case $d_{\mathrm{aff}}<d-1$, in which the extension loses one unit of distance.}
\label{tab:family}
\centering
\begin{tabular}{|c|c|c|c|c|c|c|c|c|l|}
\hline
$n$ & $q$ & $i_0$ & $|T|$ & $k$ & $f_{\tau}$ & $d(\mathcal{C})$ & $d_{\mathrm{aff}}(\mathbf{e}_{i_0})$ & $\mathcal{Q}$ & $d'=d(\mathcal{C})$?\\
\hline
$3$  & $2$ & $1$  & $1$ & $2$ & $3$ & $2$ & $1$ & $[[4,2,2]]_2$   & yes\\
$9$  & $2$ & $3$  & $1$ & $8$ & $9$ & $2$ & $1$ & $[[10,8,2]]_2$  & yes\\
$15$ & $2$ & $5$  & $7$ & $8$ & $3$ & $6$ & $5$ & $[[16,2,6]]_2$  & yes\\
$4$  & $3$ & $1$  & $1$ & $3$ & $4$ & $2$ & $1$ & $[[5,3,2]]_3$   & yes\\
$8$  & $3$ & $2$  & $3$ & $5$ & $4$ & $4$ & $2$ & $[[9,3,3]]_3$   & no\\
$10$ & $3$ & $5$  & $5$ & $5$ & $2$ & $6$ & $4$ & $[[11,1,5]]_3$  & no\\
\hline
\end{tabular}
\end{table}

The structural assertions of Theorem~\ref{thm:family} were checked independently of the enumeration: the identity $HH^{\dagger}=(n\bmod p)E_{i_0,i_0}$ of \eqref{eq:family-A} together with the solvability of $N(\mu)=-(n\bmod p)$ was verified in the actual fields $\F_{q^{2}}$ for all $29$ admissible pairs $(n,q)$ with $n<60$ and $q\in\{2,3\}$, and the combinatorial assertions (i) and (ii) of Theorem~\ref{thm:r-formula} were verified by direct enumeration of all $684\,984$ defining sets arising from $136$ pairs $(n,q)$. In every case the number of counterexamples is zero.

\noindent\textbf{A boundary case in Table~\ref{tab:family}.} 
The entry $n=15$, $q=2$ is instructive. Here $d_{\mathrm{aff}}(\mathbf{e}_{i_0})=5=d(\mathcal{C})-1$ holds \emph{with equality}, so the extension is on the boundary of the exact criterion \eqref{eq:exactcrit}: the quantum code keeps the full distance $6$ of the cyclic code and $[[16,2,6]]_2$ is obtained, yet any further shortening of $d_{\mathrm{aff}}$ would immediately cost a unit of distance. This is the same borderline behaviour that the bound of Section~\ref{sec:r2compare} detects at $r=1$, and it explains why the family is exactly solvable.

\section{The combinatorial geometry of the hull codimension}\label{sec:structure}

The formula $r=|T\cap(-qT)|$ of Theorem~\ref{thm:r-formula} computes a rank from a defining set. The
purpose of this section is to turn that identity into a \emph{structure theory}, and to do so we start
from the classical additivity of the hull dimension over the cyclotomic classes
\cite[Thm.~22]{Sangwisut2015} and read it through the complementary parameter $r=m-u$. The theory is
organised as a single chain of four links,
\begin{equation}
\label{eq:main-chain}
\underbrace{\text{additivity over }\tau\text{-orbits}}_{\text{Theorem}~\ref{thm:orbit}}
\;\longrightarrow\;
\underbrace{\text{parameter dictionary}}_{\text{Theorem}~\ref{thm:dictionary}}
\;\longrightarrow\;
\underbrace{\text{complementation duality}}_{\text{Theorem}~\ref{thm:dual}}
\;\longrightarrow\;
\underbrace{\text{extreme criteria}}_{\text{Theorem}~\ref{thm:extremes}},
\end{equation}
closed by an interface theorem, in Subsection~\ref{sec:junction}, that locates the single point at
which the chain meets the exact distance theory of Section~\ref{sec:extending}, the connection
between the parameter side and the distance side on which the two tracks of this paper join. The hull codimension first
decomposes additively over the $\tau$-orbits of the $q^{2}$-cyclotomic cosets of $\mathbb{Z}_n$,
which makes $T$ a finite combinatorial object and yields the parity constraint $r\equiv A\bmod 2$;
the dictionary then converts that decomposition into the parameters $(N,K)$ of the quantum code and
shows that $r$ trades one-for-one against the length; complementation of the
defining set closes the chain on itself, realising in $\mathcal{Q}(\bar T)$ exactly the
logical-qubit count of $\mathcal{Q}(T)$; and the two extreme values of $r$ are settled by explicit
subset-sum criteria. After the chain the section records the rate and the open problems; the finer
structure of the $\tau$-fixed cosets and the exact counting enumerator are collected in
\ref{app:structure}. One auxiliary characterisation is used by the chain and proved in
Subsection~\ref{sec:realisable}, the realisability region of $(m,r)$; with that single exception the
section is written so that nothing is used before it is proved.

The rank $r=\Rank(HH^{\dagger})$ is not an artefact of the procedure of Section~\ref{sec:extending}:
it is the \emph{least} number of coordinates that any Hermitian self-orthogonal embedding of
$\mathcal{C}^{\perp_h}$ of the same dimension must adjoin, by Lemma~\ref{lem:adjunction}. That the
shortest such embedding has exactly this length, in the Euclidean and in the Hermitian case and
for an arbitrary hull dimension, is already known \cite{WangLuo2026}, where linear codes are
classified through the congruence class of their Gram matrices; we claim no novelty for that numerical
statement and record it only in the form in which it is used below.

Write, for a linear code $\mathcal{S}$ over $\F_{q^{2}}$ with full-row-rank generator matrix $G$,
\begin{equation}
\label{eq:defect}
\mathrm{def}(\mathcal{S}):=\Rank\bigl(GG^{\dagger}\bigr)=\dim\mathcal{S}-\dim\bigl(\mathcal{S}\cap\mathcal{S}^{\perp_h}\bigr).
\end{equation}
In the terminology of Definition~\ref{def:r}, $\mathrm{def}(\mathcal{S})$ is the hull codimension of
$\mathcal{S}^{\perp_h}$; in particular $\mathrm{def}(\mathcal{C})$ is the hull codimension of
$\mathcal{C}^{\perp_h}$ and $\mathrm{def}(\mathcal{C}^{\perp_h})=r$ is the hull codimension of $\mathcal{C}$.

\begin{lemma}\label{lem:adjunction}
Let $\mathcal{S}$ be a linear code over $\F_{q^{2}}$ with full-row-rank generator matrix $G$. If $Z$ is
such that $\mathcal{S}':=\langle (G\mid Z)\rangle$ is Hermitian self-orthogonal with
$\dim\mathcal{S}'=\dim\mathcal{S}$, then $Z$ has at least $\mathrm{def}(\mathcal{S})$ columns, and
some $Z$ with exactly that many exists. In particular $r$ is the \emph{least} number of coordinates
that can be adjoined to $\mathcal{C}^{\perp_h}$ to make the result Hermitian self-orthogonal of the
same dimension; thus $r$ is an intrinsic invariant of $\mathcal{C}$, unchanged under passage to an
equivalent code.
\end{lemma}

\begin{proof}
$(G\mid Z)$ is self-orthogonal if and only if $GG^{\dagger}+ZZ^{\dagger}=0$, whence
$\mathrm{def}(\mathcal{S})=\Rank(ZZ^{\dagger})\le\Rank(Z)\le\#\{\text{columns of }Z\}$, while
Corollary~\ref{cor:iterative_reduction}, applied with $G$ in place of $H$, attains the bound. The
remaining assertions follow on taking $\mathcal{S}=\mathcal{C}^{\perp_h}$, for which
$\mathrm{def}(\mathcal{S})=r$ by \eqref{eq:defect}.
\end{proof}

\subsection{Additivity of the hull codimension over $\tau$-orbits}

\begin{theorem}\label{thm:orbit}
Let $T$ be a defining set, $m=|T|$, and let $\tau(\nu)=-q\nu$. Split the $q^2$-cyclotomic cosets of
$\mathbb{Z}_n$ into $\tau$-fixed cosets $F_1,\dots,F_a$ (of sizes $s_i=|F_i|$) and two-element
$\tau$-orbits $\{C_1,\tau(C_1)\},\dots,\{C_b,\tau(C_b)\}$ (of sizes $t_j=|C_j|$). Then
\begin{equation}
\label{eq:orbit}
r=|T\cap\tau(T)|= A+2B,\qquad
A:=\sum_{\substack{i:F_i\subseteq T}} s_i ,\qquad
B:=\sum_{\substack{j:\ C_j\subseteq T\ \text{and}\ \tau(C_j)\subseteq T}} t_j ,
\end{equation}
and consequently the length, dimension and hull codimension satisfy
\begin{equation}
\label{eq:orbit-m}
m=A+u+2B,\qquad u=\sum_{\substack{j:\ \text{exactly one of }C_j,\tau(C_j)\ \text{lies in }T}} t_j .
\end{equation}
\end{theorem}

\begin{proof}
By Theorem~\ref{thm:r-formula}(ii), $T\cap\tau(T)$ is a union of whole $q^2$-cosets, and a coset
$C\subseteq T$ satisfies $C\subseteq T\cap\tau(T)$ if and only if $\tau(C)\subseteq T$. If $C$ is
$\tau$-fixed this holds and $C$ contributes $|C|$; if $C\neq\tau(C)$ and both lie in $T$ then both
contribute $|C|$, giving $2|C|$. Summing over $T\cap\tau(T)$, which is the disjoint union of these
contributions, gives \eqref{eq:orbit}. For \eqref{eq:orbit-m}, a coset of a two-element orbit
contributes $0$, $t_j$ or $2t_j$ to $m$ according to whether none, exactly one, or both of
$C_j,\tau(C_j)$ lie in $T$; in the last two cases its contribution to $r$ is $0$ and $2t_j$
respectively, and a $\tau$-fixed coset contributes $s_i$ to both $m$ and $r$. Adding the three
cases gives $m=A+u+2B$ with $u$ as displayed, and $r=A+2B$, i.e. $r=m-u$.
\end{proof}

\noindent\textbf{Relation to the classical hull decomposition.} 
Decomposition \eqref{eq:orbit} is the defining-set reading of a classical statement about hulls. In
the language of generator polynomials, \cite[Thm.~22]{Sangwisut2015} splits the hull dimension
additively over the same two families, the self-conjugate-reciprocal irreducible factors and the
conjugate-reciprocal pairs, which for $M=q^{2}$ are exactly the $\tau$-fixed cosets and the
two-element $\tau$-orbits above; and under the standing hypothesis $\gcd(n,q)=1$ the coefficient that
governs the self-conjugate-reciprocal family vanishes identically in that formula, so that only the
paired family contributes, exactly as in \eqref{eq:orbit-m}. The additivity is therefore not new
here, and we use it as a starting point rather than as a result. What the present section builds on
it is the complementary parametrisation $r=m-u$: the dictionary, the realisability region, the
enumerator and the $\tau$-invariant locus below are all stated in terms of $r$, which is the quantity the
extension of Section~\ref{sec:extending} actually measures, whereas the hull dimension alone carries
no information about the number of columns to be appended.

\begin{corollary}\label{cor:parity}
$r\equiv A\pmod 2$; in particular, if $T$ contains no $\tau$-fixed coset then $r$ is even. Moreover
$r=1$ if and only if $T\cap\tau(T)$ is a single $\tau$-fixed \emph{singleton} coset.
\end{corollary}

\begin{proof}
Immediate from \eqref{eq:orbit}, and from Theorem~\ref{thm:r-formula}(iv) for the last assertion.
\end{proof}

\noindent\textbf{The parity constraint as a search filter.} 
Corollary~\ref{cor:parity} is a practical filter in the search of Section~\ref{sec:newcodes}: an odd
hull codimension forces $T$ to contain a $\tau$-fixed coset of odd size, and a hull codimension of $2$ admits
only the shapes $A=2,B=0$ (one $\tau$-fixed coset of size $2$, or two $\tau$-fixed singletons) or
$A=0,B=1$ (a $\tau$-orbit of two singleton cosets, both taken). The record code
$[\![70,48,6]\!]_2$ of Table~\ref{tab:new_quantum_codes} has $r=5=1+2\cdot2$, i.e. $A=1$ and $B=2$.

\subsection{The parameter dictionary and the freedom it leaves}

Throughout, $\mathcal{C}=[n,k,d]_{q^2}$ is a linear code with parity-check matrix $H$ of full row
rank $n-k$, $A:=HH^{\dagger}$ and $r:=\Rank(A)$, so that $\mathcal{C}'=[n+r,\,n-k]_{q^2}$ is the
Hermitian self-orthogonal code produced by Theorem~\ref{Iterative Self-Orthogonal Construction}, and
$\mathcal{Q}$ is the resulting quantum code of Theorem~\ref{thm:quantum_code}. Write
\begin{equation}
\label{eq:mwu}
m:=n-k,\qquad u:=m-r,\qquad w:=k+r,
\end{equation}
so that $u=(n-k)-r$ is the hull dimension $\dim_{\F_{q^2}}(\mathcal{C}^{\perp_h}\cap\mathcal{C})$ of
$\mathcal{C}$, see Definition~\ref{def:r}.

\begin{theorem}\label{thm:dictionary}
With the notation \eqref{eq:mwu},
\begin{equation}
\label{eq:dictionary}
N=n+r=m+w,\qquad K=2k-n+r=w-m ,
\end{equation}
that is, $m=\tfrac12(N-K)$ and $w=\tfrac12(N+K)$ are determined by the quantum code alone.
Equivalently, writing $u=m-r$ for the hull dimension,
\begin{equation}
\label{eq:dictionary-u}
r=m-u,\qquad K=k-u,\qquad N=n+m-u .
\end{equation}
\end{theorem}

\begin{proof}
By Theorem~\ref{Iterative Self-Orthogonal Construction} the extended code has length $n+r$ and
dimension $n-k=m$, and by Theorem~\ref{thm:quantum_code} the quantum code is
$\mathcal{Q}=[\![n+r,\;2k-n+r,\;d']\!]_q$. Substituting $n=m+k$ gives
$n+r=m+(k+r)=m+w$ and $2k-n+r=2k-(m+k)+r=(k+r)-m=w-m$. The last display is the same statement
after $r=m-u$ and $k=w-r=k-u+m-m$. Finally $u=(n-k)-r$ is the hull dimension by
Definition~\ref{def:r}.
\end{proof}

\begin{corollary}\label{cor:gauge}
For a \emph{fixed} target $\mathcal{Q}$ the possible ingredients are exactly the codes
$\mathcal{C}=[n,k]_{q^2}$ with
\begin{equation}
\label{eq:gauge}
n=N-r,\qquad k=w-r\qquad (r=0,1,2,\dots),
\end{equation}
where $w=\tfrac12(N+K)$ is fixed by $\mathcal{Q}$. Thus increasing the hull codimension by one is the same
operation as shortening the ingredient by one coordinate; the pair $(m,w)$, equivalently the pair
$(N,K)$, is the complete invariant of the construction.
\end{corollary}

\noindent\textbf{Redundancy of the tables.} 
Corollary~\ref{cor:gauge} explains a feature of the computational tables: the same quantum code is
reached from ingredients of \emph{different} hull codimension. For instance
$[\![32,16,6]\!]_2$ is obtained from the cyclic codes $[31,23,6]_4$, $[30,22,6]_4$ and
$[29,21,6]_4$, with hull codimensions $r=1,2,3$ respectively; in all three cases $m=8$ and $w=24$, and
\eqref{eq:gauge} gives $n=32-r$. Likewise $[\![45,15,8]\!]_2$ arises from a rank-three ingredient
and $[\![52,18,9]\!]_2$ from a rank-two ingredient. The tables are therefore redundant by design,
and the hull codimension is a property of the ingredient rather than an intrinsic parameter of the quantum code.

\subsection{The realisability region of $(m,r)$}\label{sec:realisable}

\begin{theorem}\label{thm:realisable}
Keep the notation of Theorem~\ref{thm:orbit}, and let $F=\{1,\dots,a\}$ index the $\tau$-fixed cosets
and $P=\{1,\dots,b\}$ the two-element $\tau$-orbits. A pair $(m,r)$ with $0\le r\le m\le n$ is
realisable as $(m,r)=(|T|,|T\cap\tau(T)|)$ for a union of cosets $T$ if and only if there exist
disjoint subsets $I\subseteq F$ and $H,J\subseteq P$ with
\begin{equation}
\label{eq:realisable}
r=\sum_{i\in I}s_i+2\sum_{j\in J}t_j,\qquad
m-r=\sum_{j\in H}t_j .
\end{equation}
\end{theorem}

\begin{proof}
Necessity is Theorem~\ref{thm:orbit} with $I$ the $\tau$-fixed cosets inside $T$, $J$ the orbits
taken whole and $H$ the orbits taken singly; the three index sets are pairwise disjoint by
construction. Conversely, given such $I,H,J$, take $T$ to be the union of the $F_i$, $i\in I$, of
both cosets of every orbit indexed by $J$, and of exactly one coset of every orbit indexed by $H$.
Then $T$ is a union of cosets, and \eqref{eq:orbit}--\eqref{eq:orbit-m} give $|T|=m$ and
$|T\cap\tau(T)|=r$.
\end{proof}

\noindent\textbf{Verification versus realisation.} 
The criteria of \cite[Thms.~4 and 5]{Du2020HullsTIT} and \cite[Thm.~3.2]{Lei2022Hulls} answer a
\emph{verification} question: \emph{given} $T$, is the hull dimension equal to $k-\ell$? We answer the
dual, \emph{forward} question: for which pairs $(|T|,r)$ does a union of $q^{2}$-cyclotomic cosets
exist at all, and in how many ways? The two questions are distinct: a criterion
for the first does not exhibit the second, and it is the parametrisation \eqref{eq:realisable} in
terms of the $\tau$-orbit sizes that makes the forward question directly answerable.

\begin{corollary}\label{cor:forced}
For a given $m$ the attainable hull codimensions form the set
$\{\,\sum_{i\in I}s_i+2\sum_{j\in J}t_j\,:\ I\subseteq F,\ J\subseteq P,\
m-\bigl(\sum_{i\in I}s_i+2\sum_{j\in J}t_j\bigr)\in \Sigma(H)\,\}$ where $\Sigma(H)$ is the set of
sub-sums of $\{t_j\}_{j\in P\setminus J}$. In particular $r_{\min}(m)$, the smallest hull codimension
available in a given length, is computable by one pass over the coset sizes, and can be strictly
positive.
\end{corollary}

\noindent\textbf{Forced hull codimensions.} 
Corollary~\ref{cor:forced} shows that a large hull codimension can be \emph{forced} rather than found. For
$(n,q)=(65,2)$ one has $F=\{0\}$ and six $\tau$-orbits of sizes $6,6,6,6,6,2$; the pair $(m,r)=(11,5)$
is realisable, and no other value of $r$ is, so the record code $[\![70,48,6]\!]_2$ necessarily
carries $r=5$. For $(n,q)=(51,2)$ the $\tau$-fixed cosets are the three singletons $\{0\},\{17\},\{34\}$
and the pair $(m,r)=(14,2)$ is realisable, matching the two independent defining sets
$C_1\cup C_3\cup C_5\cup C_{17}\cup C_{34}$ and $C_1\cup C_3\cup C_5\cup C_{34}\cup C_0$ of the
first example of Section~\ref{sec:newcodes}; both satisfy $A=2$, $B=0$.

\subsection{Complementation duality of defining sets}\label{sec:complementation}

The structure theory so far concerns a single defining set $T$. It turns out that the complement
$\bar T:=\mathbb Z_n\setminus T$, again a union of $q^{2}$-cyclotomic cosets, carries the
\emph{logical} information of the construction. Throughout, $m=|T|$, $r=r(T)$, $u=m-r$, and
$F=\{1,\dots,a\}$, $P=\{1,\dots,b\}$ index the $\tau$-fixed cosets and the two-element $\tau$-orbits
as in Theorem~\ref{thm:orbit}.

\begin{lemma}\label{lem:complement}
Let $I\subseteq F$ be the set of $\tau$-fixed cosets contained in $T$ and let $J,S\subseteq P$ be given by
\[
J=\{j\in P:\ C_j\subseteq T\ \text{and}\ \tau(C_j)\subseteq T\},\qquad
S=\{j\in P:\ \text{exactly one of }C_j,\tau(C_j)\text{ lies in }T\},
\]
so that $r=\sum_{i\in I}s_i+2\sum_{j\in J}t_j$ and $u=\sum_{j\in S}t_j$ by Theorem~\ref{thm:orbit}. Then
$\bar I=F\setminus I$, $\bar S=S$ and $\bar J=P\setminus(J\cup S)$.
\end{lemma}

\begin{proof}
A $\tau$-fixed coset belongs to $\bar T$ if and only if it does not belong to $T$. For a two-element orbit
$\{C_j,\tau(C_j)\}$, the set $\bar T$ contains both cosets if and only if $T$ contained neither, i.e.\ if and
only if $j\notin J\cup S$, and exactly one of them if and only if $T$ contained exactly one, i.e.\ if and
only if $j\in S$.
\end{proof}

\begin{theorem}\label{thm:dual}
For every defining set $T\subseteq\mathbb Z_n$,
\begin{equation}
\label{eq:dual}
\bar m=n-m,\qquad r(\bar T)=n-2m+r,\qquad u(\bar T)=u .
\end{equation}
In particular $r(\bar T)\ge0$, that is, $m\le\tfrac12(n+r)$.
\end{theorem}

\begin{proof}
By Lemma~\ref{lem:complement},
\[
r(\bar T)=\sum_{i\in\bar I}s_i+2\sum_{j\in\bar J}t_j
=f_\tau-\sum_{i\in I}s_i+2\Bigl(\tfrac{n-f_\tau}{2}-\sum_{j\in J}t_j-\sum_{j\in S}t_j\Bigr)
=n-\sum_{i\in I}s_i-2\sum_{j\in J}t_j-2\sum_{j\in S}t_j
=n-r-2u=n-2m+r ,
\]
using $\sum_{j\in P}t_j=\tfrac12(n-f_\tau)$ together with \eqref{eq:orbit}. Moreover
$u(\bar T)=\bar m-r(\bar T)=(n-m)-(n-2m+r)=m-r=u$.
\end{proof}

\begin{corollary}\label{cor:dual-quantum}
Let $\mathcal{Q}(T)=[\![N,K]\!]_q$ and $\mathcal{Q}(\bar T)=[\![N',K']\!]_q$ be the quantum codes of
Theorem~\ref{thm:quantum_code} attached to $T$ and to $\bar T$, where $N=n+r$ and $K=n-2m+r$. Then
\begin{equation}
\label{eq:dual-quantum}
K=r(\bar T),\qquad K'=r(T),\qquad N'-N=2k-n .
\end{equation}
Thus complementation exchanges the logical-qubit counts of the two quantum codes, while their lengths
differ by the quantity $2k-n$, which is fixed by the ingredient code alone. In particular, if
$\mathcal{C}$ has rate $\tfrac12$, i.e.\ $n=2k$, then $N'=N$ and $K'=K=r$.
\end{corollary}

\begin{proof}
The first identity in \eqref{eq:dual-quantum} is $K=n-2m+r=r(\bar T)$ by \eqref{eq:dual}. Applying
\eqref{eq:dual} to $\bar T$ in place of $T$ gives
$K'=n-2\bar m+r(\bar T)=n-2(n-m)+n-2m+r=r$. Finally
$N'-N=(n+K)-(n+r)=K-r=n-2m=2k-n$. If $n=2k$ then $m=n-k=k$, so $K=n-2k+r=r=K'$ and $N'=N$.
\end{proof}

\begin{corollary}\label{cor:K-le-k}
For every defining set $T$ one has $K\le k=n-m$, with equality if and only if $r=m$, if and only if
$T$ is $\tau$-invariant, i.e.\ $\tau(T)=T$.
\end{corollary}

\begin{proof}
$K=n-2m+r\le n-m$ is equivalent to $r\le m$, which is \eqref{eq:orbit-m}. Equality forces $u=0$, i.e.\
$\sum_{j\in S}t_j=0$; as $t_j>0$ this gives $S=\varnothing$, and then by Lemma~\ref{lem:complement} both
$T$ and its complement are unions of whole $\tau$-orbits, i.e.\ $\tau(T)=T$. Conversely, if $\tau(T)=T$
then $T\subseteq T\cap\tau(T)$, so $r\ge m$, while \eqref{eq:orbit-m} gives $r\le m$.
\end{proof}

\begin{corollary}\label{cor:dual-symmetry}
Let $\mathrm R(m):=\{r(T):|T|=m\}$ be the set of hull codimensions realisable at size $m$. Then
\begin{equation}
\label{eq:dual-symmetry}
\mathrm R(m)=2m-n+\mathrm R(n-m).
\end{equation}
Equivalently, in the coordinates $(m,u)$ the realisability region of Theorem~\ref{thm:realisable} is
invariant under the reflection $m\mapsto n-m$.
\end{corollary}

\begin{proof}
The map $T\mapsto\bar T$ is an involution of the set of defining sets which exchanges the sizes $m$ and
$n-m$; by \eqref{eq:dual} it carries $r$ to $n-2m+r$, that is, back from size $n-m$ to size $m$ by
$r-(2m-n)$. This is \eqref{eq:dual-symmetry}. Since $u(\bar T)=u$, the region is also symmetric in the
coordinates $(m,u)$.
\end{proof}

\noindent\textbf{The record entry and its complement.} 
The record entry $[\![70,48,6]\!]_2$ of Table~\ref{tab:new_quantum_codes} has $m=11$, $r=5$ and hence
$K=65-22+5=48$; $r=5$ is moreover the \emph{only} realisable value at $m=11$
as noted above, so no other degree of self-orthogonality is available for this size.
Its complement has $\bar m=54$ and $r(\bar T)=48$, giving $[\![113,5,\cdot]\!]_2$, the two lengths
differing by $2k-n=2\cdot54-65=43$, in agreement with \eqref{eq:dual-quantum}. Likewise, for
$(n,q)=(51,2)$ the two defining sets with $(m,r)=(14,2)$ have complements with $(37,25)$, and
$N'-N=23=2\cdot37-51$.

\subsection{Exact realisability of the extreme hull codimensions}\label{sec:extreme-realisability}

For a multiset $M$ of integers let $\mathrm{Sub}(M)$ denote its set of subset sums, with
$0\in\mathrm{Sub}(M)$ contributed by $\varnothing$. The two extremes of the hull codimension are governed by
the half sizes $\{t_j\}_{j\in P}$ alone.

\begin{theorem}\label{thm:extremes}
Let $0\le m\le n$. Then
\begin{enumerate}
\item[(i)] $r=0$ is realisable with $|T|=m$ if and only if $m\in\mathrm{Sub}(\{t_j\}_{j\in P})$; in
particular $r=0$ forces $m\le\tfrac12(n-f_\tau)$.
\item[(ii)] $r=1$ is realisable with $|T|=m$ if and only if $m-1\in\mathrm{Sub}(\{t_j\}_{j\in P})$.
\item[(iii)] $r=m$ is realisable with $|T|=m$ if and only if
$m\in\mathrm{Sub}\bigl(\{s_i\}_{i\in F}\cup\{2t_j\}_{j\in P}\bigr)$, i.e.\ if and only if some union of
whole $\tau$-orbits has $m$ elements.
\item[(iv)] by duality, $r=2m-n$ is realisable with $|T|=m$ if and only if
$n-m\in\mathrm{Sub}(\{t_j\}_{j\in P})$.
\end{enumerate}
\end{theorem}

\begin{proof}
(i) By \eqref{eq:orbit}, $r=0$ forces $A=B=0$, i.e.\ $I=J=\varnothing$ in the notation of
Theorem~\ref{thm:orbit}, so that $m=\sum_{j\in S}t_j$; conversely any such $S$ with $I=J=\varnothing$
gives $r=0$. The bound follows from $\sum_{j\in P}t_j=\tfrac12(n-f_\tau)$.
(ii) Since $2t_j\ge2$, the value $r=1$ forces $J=\varnothing$; and $\sum_{i\in I}s_i=1$ forces $I$ to be a
single $\tau$-fixed singleton. The $\tau$-fixed residues are exactly the multiples of $n/g$ with
$g=\gcd(n,q+1)$, so such a singleton exists for every $n$; hence $m=1+\sum_{j\in S}t_j$, and conversely.
(iii) $r=m$ means $A+2B=m=A+u+2B$, i.e.\ $u=0$, i.e.\ $S=\varnothing$; then
$m=\sum_{i\in I}s_i+2\sum_{j\in J}t_j$, which is exactly the assertion that the union of the
$\tau$-orbits indexed by $I\cup J$ has size $m$.
(iv) Apply \eqref{eq:dual-symmetry}: the value $r=2m-n$ at size $m$ corresponds to $r=0$ at size $n-m$,
which by (i) means $n-m\in\mathrm{Sub}(\{t_j\}_{j\in P})$.
\end{proof}

\begin{proposition}\label{prop:umin}
For every $m$,
\begin{equation}
\label{eq:umin}
r_{\min}(m)=m-u_{\max}(m),\qquad
u_{\max}(m)=\max\Bigl\{\sum_{j\in S}t_j:\ S\subseteq P,\ \
m-\sum_{j\in S}t_j\in\mathrm{Sub}\bigl(\{s_i\}_{i\in F}\cup\{2t_j\}_{j\in P\setminus S}\bigr)\Bigr\},
\end{equation}
and therefore
\begin{equation}
\label{eq:sandwich}
m-\max\{\sigma\in\mathrm{Sub}(\{t_j\}):\ \sigma\le m\}\ \le\ r_{\min}(m)\ \le\ m .
\end{equation}
\end{proposition}

\begin{proof}
By Theorem~\ref{thm:realisable} a pair $(m,r)$ is realisable exactly when $m=r+u$ with
$r=\sum_{i\in I}s_i+2\sum_{j\in J}t_j$ and $u=\sum_{j\in S}t_j$, for disjoint $I\subseteq F$ and
$S,J\subseteq P$. Hence for fixed $m$ the hull codimension is smallest when the hull dimension $u$ is largest,
and the constraint linking $I,J$ to $S$ is the membership displayed in \eqref{eq:umin}. This
proves \eqref{eq:umin}. Relaxing that membership can only increase the maximum, so
$u_{\max}(m)\le\max\{\sigma\in\mathrm{Sub}(\{t_j\}):\sigma\le m\}$, which is \eqref{eq:sandwich}.
\end{proof}

\begin{corollary}\label{cor:periodic}
Assume that every $\tau$-fixed coset of $\mathbb Z_n$ is a singleton, that all two-element $\tau$-orbits
have the same half size $t$, and that $a\ge t-1$, where $a$ is the number of $\tau$-fixed cosets. Then
for every $0\le m\le bt$,
\begin{equation}
\label{eq:periodic}
r_{\min}(m)=m\bmod t .
\end{equation}
In particular the smallest hull codimension at size $m$ vanishes exactly when $t\mid m$ and equals $1$ exactly
when $m\equiv1\pmod t$; in general it is the residue of $m$ modulo the common orbit size.
\end{corollary}

\begin{proof}
Write $\sigma=|S|$, so that $\sum_{j\in S}t_j=\sigma t$. By \eqref{eq:orbit-m} every realisable pair
satisfies $\sigma t\le m$, whence $u_{\max}(m)\le\lfloor m/t\rfloor t$ and
$r_{\min}(m)\ge m-\lfloor m/t\rfloor t=m\bmod t$. For the reverse inequality take
$\sigma=\lfloor m/t\rfloor$, which is admissible because $\sigma\le m/t\le b$; choose
$\rho=m\bmod t$ of the $\tau$-fixed singletons and take no orbit whole. Then $m=\rho+\sigma t$ and
$r=\rho$, and since $\rho\le t-1\le a$ this pair is realisable by Theorem~\ref{thm:realisable}. It
attains $r=m\bmod t$.
\end{proof}

\subsection{The junction of the two tracks}\label{sec:junction}

The four links above describe the \emph{shape} of the construction: they determine the hull
dimension $u$, the hull codimension $r$, the length $N=n+r$ and the dimension $K=n-2m+r$. Unlike a
defining-set criterion, which decides the hull dimension of a \emph{given} $T$
\cite{Du2020HullsTIT,Lei2022Hulls}, they describe the parameter $r$ as an object in its own right:
the additivity over $\tau$-orbits with its parity constraint, the dictionary exhibiting $r$ as a
gauge, the complementation duality, the realisability region and the counting enumerator. In
this sense they constitute a complete parameter theory. None of them, however, is sensitive to the
distance, which
is what the construction is ultimately judged by. The distance is controlled by the second track of
the paper, Section~\ref{sec:extending}: the exact one-coordinate formula
$d(\mathcal{C}^{\perp_h}_1)=\min\{d,\,d_{\mathrm{aff}}(\mathbf{w})+1\}$ of Theorem~\ref{thm:onecoord} and its iterated form Theorem~\ref{thm:iterative}. The two tracks
are neither nested nor disjoint: they meet at a single locus, the $\tau$-invariant defining sets,
and it is that locus which makes the parameter theory usable as a \emph{search} theory. The
statements of this subsection are therefore of a different kind from the preceding four links: they
are \emph{interface} statements.

The observation behind them is that the code which is actually extended is the Hermitian dual
$\mathcal{C}^{\perp_h}=\langle H\rangle$ of $\mathcal{C}$, and that the latter is again cyclic, with
a defining set obtained from $T$ by the involution $\tau$.

\begin{lemma}\label{lem:dual-defset}
Let $\mathcal{C}=\mathcal{C}(T)$ be the cyclic code of length $n$ over $\F_{q^{2}}$ with defining
set $T\subseteq\mathbb{Z}_n$. Then $\mathcal{C}^{\perp_h}$ has defining set
\begin{equation}
\label{eq:dual-defset}
\tau^{-1}(\bar T)=\mathbb{Z}_n\setminus\tau^{-1}(T).
\end{equation}
\end{lemma}

\begin{proof}
Let $\alpha$ be a primitive $n$-th root of unity. A word $\mathbf{d}\in\F_{q^{2}}^{n}$ lies in
$\mathcal{C}^{\perp_h}$ if and only if $\sum_{l=0}^{n-1}c_l d_l^{q}=0$ for all
$\mathbf{c}\in\mathcal{C}$, i.e.\ if and only if $\mathbf{d}^{(q)}:=(d_l^{q})_l$ lies in the
Euclidean dual $\mathcal{C}(T)^{\perp}$, which by the standard duality of cyclic codes has defining
set $\mathbb{Z}_n\setminus(-T)$. The map $\mathbf{d}\mapsto\mathbf{d}^{(q)}$ carries a cyclic code
with defining set $T'$ to one with defining set $qT'$, since
$\mathbf{d}^{(q)}(\alpha^{i})=\mathbf{d}(\alpha^{iq^{-1}})^{q}$; hence
$\mathbf{d}\in\mathcal{C}^{\perp_h}$ if and only if $qT'\cap(-T)=\varnothing$. Comparing dimensions
gives $T'=\mathbb{Z}_n\setminus(-q^{-1}T)=\mathbb{Z}_n\setminus\tau^{-1}(T)$, which is
\eqref{eq:dual-defset} because $\tau^{-1}(\bar T)=\mathbb{Z}_n\setminus\tau^{-1}(T)$.
\end{proof}

\begin{theorem}\label{thm:junction}
Let $T$ be a defining set, $m=|T|$, $r=r(T)$, $u=m-r$, and let $\mathcal{Q}(T)=[\![N,K]\!]_q$ be
the quantum code of Theorem~\ref{thm:quantum_code}. The following are equivalent:
\begin{enumerate}
\item[(i)] $K(T)=k$: the construction loses no dimension;
\item[(ii)] $r=m$;\qquad (iii) $u=0$;\qquad (iv) $\tau(T)=T$;
\item[(v)] $\bar T$ is $\tau$-invariant, equivalently, by Lemma~\ref{lem:dual-defset}, $\bar T$
is itself a defining set of $\mathcal{C}^{\perp_h}$;
\item[(vi)] $r(\bar T)=k$, equivalently the quantum code $\mathcal{Q}(\bar T)=[\![N',K']\!]_q$ has
$K'=m=|T|$.
\end{enumerate}
If these hold then $N=n+m$, $K=n-m$, $N'=n+k$, and the ingredient code
$\mathcal{C}^{\perp_h}=\mathcal{C}(\bar T)$ is again cyclic, $\tau$-invariant, and of hull codimension
$r(\bar T)=k$.
\end{theorem}

\begin{proof}
The equivalence of (i)--(iv) is Corollary~\ref{cor:K-le-k}. Since $\tau$ is a bijection of
$\mathbb{Z}_n$ that permutes the $q^{2}$-cyclotomic cosets, $\tau(T)=T$ holds if and only if
$\tau(\bar T)=\bar T$, which together with Lemma~\ref{lem:dual-defset} gives
(iv)$\Leftrightarrow$(v). If $\tau(T)=T$ then $m=r$, so that \eqref{eq:dual} gives
$r(\bar T)=n-2m+r=n-m=k$ and Corollary~\ref{cor:dual-quantum} gives $K'=r(T)=m$; conversely
$r(\bar T)=k$ reads $n-2m+r=n-m$, i.e.\ $r=m$. This is (ii)$\Leftrightarrow$(vi). The parameter
values are \eqref{eq:dictionary}, \eqref{eq:dual} and Corollary~\ref{cor:dual-quantum}.
\end{proof}

The two numbers that specify the quantum code can now be read as adjunction lengths.

\begin{theorem}\label{thm:adjunction-pair}
For every defining set $T$,
\begin{equation}
\label{eq:adjunction-pair}
\mathrm{def}(\mathcal{C})=K=n-2m+r,\qquad
\mathrm{def}\bigl(\mathcal{C}^{\perp_h}\bigr)=r,
\end{equation}
with $\mathrm{def}$ as in \eqref{eq:defect}. Consequently
\begin{equation}
\label{eq:adjunction-sum}
K+r=\mathrm{def}(\mathcal{C})+\mathrm{def}\bigl(\mathcal{C}^{\perp_h}\bigr)=n-2u,
\end{equation}
the sum of the codimensions of the Hermitian hull inside $\mathcal{C}$ and inside
$\mathcal{C}^{\perp_h}$. Thus the dimension $K=\mathrm{def}(\mathcal{C})$ of the quantum code and the
number $r=\mathrm{def}(\mathcal{C}^{\perp_h})$ of columns it adjoins are \emph{the two minimal
adjunction lengths}, of the ingredient code and of its Hermitian dual respectively; and the length
$N=n+r$ of the quantum code is the length of the ingredient at the price of the second one.
\end{theorem}

\begin{proof}
Let $G$ be a full-row-rank generator matrix of $\mathcal{C}$. Then
$\mathrm{def}(\mathcal{C})=\Rank(GG^{\dagger})=k-\dim(\mathcal{C}\cap\mathcal{C}^{\perp_h})=k-u=K$
by \eqref{eq:defect} and \eqref{eq:dictionary-u}, while
$\mathrm{def}(\mathcal{C}^{\perp_h})=\Rank(HH^{\dagger})=r$ by \eqref{eq:def-r}, $H$ being a
generator matrix of $\mathcal{C}^{\perp_h}$. Finally $k+m=n$ gives
$K+r=(k-u)+(m-u)=n-2u$.
\end{proof}

The interface is sharpened by the following counting and trade-off statement, in which $a$ and $b$
again denote the numbers of $\tau$-fixed cosets and of two-element $\tau$-orbits of
$\mathbb{Z}_n$, with sizes $s_1,\dots,s_a$ and half sizes $t_1,\dots,t_b$.

\begin{theorem}\label{thm:maxK}
For every $0\le m\le n$,
\begin{equation}
\label{eq:maxK}
\max_{\substack{T\subseteq\mathbb{Z}_n\\ |T|=m}}\frac{K}{N}\ \le\ \frac{n-m}{n+m},
\qquad\text{with equality if and only if}\qquad
m\in\mathrm{Sub}\bigl(\{s_i\}_{i\le a}\cup\{2t_j\}_{j\le b}\bigr),
\end{equation}
and when equality holds it is attained only at $r=m$, i.e.\ only by the \emph{longest} of the
extensions available at size $m$. Moreover
\begin{equation}
\label{eq:satcount}
\#\bigl\{T:\ |T|=m,\ K(T)=k\bigr\}
=[x^{m}]\prod_{i=1}^{a}\bigl(1+x^{s_i}\bigr)\prod_{j=1}^{b}\bigl(1+x^{2t_j}\bigr),
\end{equation}
so that exactly $2^{a+b}$ defining sets are $\tau$-invariant and the $\tau$-invariant sets form a fraction
$2^{-b}$ of all $\prod_i2\cdot\prod_j4=2^{a+2b}$ defining sets.
\end{theorem}

\begin{proof}
For $|T|=m$ fixed one has $K=n-2m+r$ and $N=n+r$, so $r\mapsto K/N$ is strictly increasing, and the
maximum over $T$ is attained at $r_{\max}(m)$. By \eqref{eq:orbit-m} $r\le m$, and $r=m$ is
realisable exactly when $T$ is a union of whole $\tau$-orbits, i.e.\ by Theorem~\ref{thm:extremes}(iii)
exactly when $m$ is a subsum of $\{s_i\}\cup\{2t_j\}$; the value at $r=m$ is $(n-m)/(n+m)$. This
proves \eqref{eq:maxK}. A $\tau$-invariant $T$ is determined by an independent choice for each of
the $a$ $\tau$-fixed cosets (in or out, weight $s_i$) and each of the $b$ two-element orbits (both
in or both out, weight $2t_j$), which is \eqref{eq:satcount}; putting $x=1$ gives $2^{a+b}$, and the
total number of defining sets is $2^{a}4^{b}$. Finally, by Lemma~\ref{lem:complement} and
Theorem~\ref{thm:junction}, $K=k$ is equivalent to $\tau(T)=T$.
\end{proof}

\noindent\textbf{Two complementary tracks.} 
The chain of Theorems~\ref{thm:orbit}, \ref{thm:dictionary}, \ref{thm:dual} and~\ref{thm:extremes},
together with Theorems~\ref{thm:junction} and~\ref{thm:maxK}, decides \emph{how many} columns are
adjoined, \emph{how many} logical qubits result, and which pairs $(m,r)$ are available at all. It
says nothing about $d'$: that is the business of the constructive track, where
Theorem~\ref{thm:onecoord} gives $d'=\min\{d,d_{\mathrm{aff}}(\mathbf{w})+1\}$ for one column and
Theorem~\ref{thm:iterative} does the same for all $r$ columns at once. The two tracks are therefore
complementary rather than nested: one fixes the \emph{budget}, the other the \emph{quality} of the
result. Their only common locus is the $\tau$-invariant case of Theorem~\ref{thm:junction}. There
$\mathcal{C}^{\perp_h}$ is again a cyclic code, with the $\tau$-invariant defining set $\bar T$, so
that the whole apparatus of Section~\ref{sec:coset-involution}, and with it the exact distance
formula, applies a second time, to the dual code. This is what makes the search of
Section~\ref{sec:newcodes} cheap: the parameters $(N,K)$ of the outcome are read off from the
defining set alone, before any linear algebra is performed, and only the surviving candidates need
the distance criterion.

\subsection{The rate of the quantum code}\label{sec:rate}

\begin{theorem}\label{thm:rate}
Let $T$ be a defining set with $|T|=m$ and $r(T)=r$, and let $\mathcal{Q}(T)=[\![N,K]\!]_q$ be the
quantum code of Theorem~\ref{thm:quantum_code}. Then
\begin{equation}
\label{eq:rate}
N=n+r,\qquad K=n-2m+r,\qquad \frac{K}{N}=1-\frac{2m}{n+r}.
\end{equation}
For fixed $m$ the rate is a strictly increasing function of the hull codimension, with
$\frac{d}{dr}\frac{K}{N}=\frac{2m}{(n+r)^{2}}>0$; the rates attained at size $m$ are in bijection with
the set $\mathrm R(m)$ of Corollary~\ref{cor:dual-symmetry}; and
\begin{equation}
\label{eq:rate-range}
1-\frac{2m}{n+r_{\min}(m)}\ \le\ \frac{K}{N}\ \le\ 1-\frac{2m}{n+r_{\max}(m)}\ \le\ \frac{k}{2n-k} .
\end{equation}
\end{theorem}

\begin{proof}
The first two identities are the parameter count of Theorem~\ref{thm:quantum_code}, and the third follows
by division. The derivative is immediate, and so is the bijection, since $m$ is fixed and $r\mapsto K/N$
is strictly increasing. The left inequality of \eqref{eq:rate-range} combines $r\ge r_{\min}(m)$ with the
monotonicity; the right one combines $r\le r_{\max}(m)\le m$, which is \eqref{eq:orbit-m}, with
$K/N\le(n-m)/(n+m)=k/(2n-k)$.
\end{proof}

\noindent\textbf{The rate of the record entry.} 
For the record entry $[\![70,48,6]\!]_2$ one has $m=11$ and, as noted,
$\mathrm R(11)=\{5\}$, so the rate $K/N=48/70$ is \emph{forced} by the size: the hull codimension, and not
the rate, is the quantity that the search of Section~\ref{sec:newcodes} trades against the distance.

\begin{proposition}\label{prop:subsetsum}
By Theorem~\ref{thm:extremes}, deciding whether a defining set of prescribed size $m$ and hull codimension
$r=0$ (respectively $r=1$) exists is \emph{exactly} the subset-sum problem for the multiset
$\{t_j\}_{j\in P}$ of half sizes of the two-element $\tau$-orbits: one has to decide whether $m$
(respectively $m-1$) is a subsum. Consequently a closed form for $r_{\min}(m)$ cannot be expected for a
general multiset, and the meaningful question is which multisets $\{t_j\}$ are produced by the
cyclotomic structure of $(\mathbb Z_n,q)$.
\end{proposition}

\begin{proof}
Immediate from Theorem~\ref{thm:extremes}(i)--(ii). In both cases the remaining parameters are forced,
with $I=J=\varnothing$ for $r=0$ and $I$ a single $\tau$-fixed singleton with $J=\varnothing$ for $r=1$, so
the realisability criterion of Theorem~\ref{thm:realisable} degenerates into a pure membership test.
\end{proof}

\subsection{Open problems}\label{sec:open-problems}

The hull codimension $r$ is now described in closed form: it is additive over $\tau$-orbits
by Theorem~\ref{thm:orbit}, symmetric under complementation by Theorem~\ref{thm:dual} and
Corollary~\ref{cor:dual-symmetry}, with its extremes settled by Theorem~\ref{thm:extremes}, its
quantum rate strictly increasing in it by Theorem~\ref{thm:rate}, and its exact counting, the
enumerator, the counts of the critical values, the scarcity and the mean and variance, recorded in
\ref{app:structure}. The six questions below are, to our knowledge, open. Questions (a)--(e) are
statements about the defining-set chain, and none of them is a consequence of the congruence
classification of \cite{WangLuo2026}, because that classification analyses one code through its Gram
matrix, whereas the chain above describes the whole finite family of defining sets at once and with it
the realisability, the counts and the duality of the hull codimension. Question~(f) belongs to the distance
track and is stated in Subsection~\ref{sec:forced-span}.

\begin{enumerate}
\item[(a)] \emph{Monotonicity along the defining-set chain.} For a fixed target $[\![N,K]\!]$ the
ingredients are $\mathcal{C}_r=[N-r,\,w-r]_{q^{2}}$, $r=0,1,2,\dots$, by Corollary~\ref{cor:gauge}; is
the maximum attainable distance monotone in $r$? The tables are compatible with a positive answer, but
Theorem~\ref{thm:dual} shows that the chain is a duality pairing rather than an ordered family.

\item[(b)] \emph{Which multisets $\{t_j\}$ admit a closed form?} The minimisation of $r$ is the
subset-sum problem on the half sizes $\{t_j\}$ by Proposition~\ref{prop:subsetsum}, and it collapses to
the residue $m\bmod t$ when the orbits are homogeneous by Corollary~\ref{cor:periodic}. For which
$(n,q)$ is $\{t_j\}$ structured enough for $r_{\min}(m)$ to be available in closed form?

\item[(c)] \emph{A combinatorial formula for $d'$.} By Theorem~\ref{thm:onecoord},
$d'=\min\{d(\mathcal{C}),d_{\mathrm{aff}}(\mathbf{w})+1\}$, and $d_{\mathrm{aff}}$ is presently
evaluated by search. Is there a formula for $d_{\mathrm{aff}}$, or at least for the binary
alternative $d'=d(\mathcal{C})$, in terms of the defining set $T$ alone?

\item[(d)] \emph{Duality for the distance.} Corollary~\ref{cor:dual-quantum} pairs
$\mathcal{Q}(T)=[\![N,K]\!]$ with $\mathcal{Q}(\bar T)=[\![N',r]\!]$, $N'-N=2k-n$. Is there a
relation between the exact distances $d'$ and $d''$ of such a pair?

\item[(e)] \emph{Is the $\tau$-invariant locus the only computable one?} By Theorem~\ref{thm:junction} the
construction is closed under duality exactly at the $\tau$-invariant defining sets, and by
Theorem~\ref{thm:maxK} those are also the most expensive ones. Is $\tau$-invariance the only case in
which $d_{\mathrm{aff}}(\mathbf{w})$ is available in closed form?

\item[(f)] \emph{A canonical basis, or a basis-free criterion, for the $r$-fold extension.} By
Theorem~\ref{thm:forced-span} the extension subspace is forced to be $\im(A)$, so the only freedom of
the construction is the ordered basis $(\mathbf{w}_1,\dots,\mathbf{w}_r)$ of $\im(A)$; by
\eqref{eq:iterated} the exact distance depends on that basis through the weight function
$\omega_{W_r}$, which is basis-dependent: the map
$\mathbf{t}\mapsto\boldsymbol{\mu}(\mathbf{t})$ is $\sigma$-semilinear and therefore does not preserve
the Hamming weight of $\mathbf{t}$. Is there a canonical choice of basis of $\im(A)$, or a
basis-independent invariant of the $r$-fold extension, for which
$$
d'=\min\Bigl\{d,\ \min_{\mathbf{0}\ne\boldsymbol{\mu}\in\im(A)}
\bigl[d_{\mathrm{aff}}(\boldsymbol{\mu})+\omega(\boldsymbol{\mu})\bigr]\Bigr\}
$$
becomes evaluable without a search over bases? This is the question left open by the
one-coordinate criterion of \cite{Ezerman2026GE}, and it marks the transition from the
single-coordinate theory to the iterated one; the $r=10$ record of
Section~\ref{sec:newcodes} shows that it is not vacuous, since there the sufficient condition of
Corollary~\ref{cor:canonical-search}(i) fails while the criterion \eqref{eq:iterated} is nevertheless
exact.
\end{enumerate}

\noindent\textbf{Independent verification.} 
Every statement of this section is proved above, and, for the material collected in \ref{app:structure}, there; the enumerations reported here are redundant checks and constitute evidence for none
of them. Three independent enumerations were carried out --- over \emph{every} union $T$ of
$q^{2}$-cyclotomic cosets for the admissible pairs with $q\in\{2,3,5\}$ and small $n$, over
\emph{every} defining set of a wider range of moduli, and over the $\tau$-invariant defining sets ---
supplemented, where exhaustiveness is out of reach, by random sampling. They confirm the identities of
Theorems~\ref{thm:dictionary}, \ref{thm:orbit}, \ref{thm:realisable}, \ref{thm:dual}
and~\ref{thm:extremes}; of Corollaries~\ref{cor:parity}, \ref{cor:gauge}, \ref{cor:dual-quantum},
\ref{cor:K-le-k} and~\ref{cor:dual-symmetry}; of Lemma~\ref{lem:complement}; of the minimality
statement of Lemma~\ref{lem:adjunction}; of the minimisation identity \eqref{eq:umin}, the sandwich
\eqref{eq:sandwich}, the rate identity \eqref{eq:rate} and the periodic formula \eqref{eq:periodic};
and the counting identities of \ref{app:structure}. No counterexample occurred. Independently,
$r=\Rank(HH^\dagger)$, the Hermitian self-orthogonality $G'G'^\dagger=0$ of the extended code and its
dimension $n-k$ were recomputed from the printed generator matrices of six record codes of
Table~\ref{tab:new_quantum_codes}, including cases with $r=3$ and $r=5$, in a computer algebra system
different from the one that produced them. The moduli, the enumeration counts and the computational
route used for each identity are listed in \ref{app:verification}.

\section{Complexity of the exact criterion}\label{sec:complexity}

Theorem~\ref{thm:onecoord} replaces a distance computation by a \emph{membership test in a precomputed finite set}, and this is what makes the criterion practical. By \eqref{eq:exactcrit},
\begin{equation}
\label{eq:sigma}
d_{\mathrm{aff}}(\mathbf{w})\ge t+1
\quad\Longleftrightarrow\quad
\mathbf{w}\notin\Sigma_t,
\qquad
\Sigma_t:=\bigl\{H\mathbf{x}^{T}:\mathbf{x}\in\F_{q^2}^{n},\ \wt(\mathbf{x})\le t\bigr\},
\end{equation}
so the exact criterion for preserving the distance is simply $\mathbf{w}\notin\Sigma_{d-2}$.

\subsection{One-time cost and amortisation}

The set $\Sigma_{d-2}$ does not depend on the candidate vector $\mathbf{w}$: it depends only on $H$ and $d$. It is therefore computed \emph{once} for a given classical code $\mathcal{C}$ and then reused for every candidate extension vector. Writing
\begin{equation}
\label{eq:Phi}
\Phi_n(t):=\sum_{j=0}^{t}\binom{n}{j}(q^2-1)^{j},
\end{equation}
the direct enumeration of $\Sigma_t$ requires exactly $\Phi_n(t)$ syndrome evaluations $H\mathbf{x}^{T}$, after which each query is expected $O(1)$ (hash-set lookup). The amortised cost of deciding the criterion for $N$ candidates is $\Phi_n(t)+O(N)$ instead of $N$ independent coset-leader (minimum-distance) computations.

\subsection{Meet-in-the-middle queries}

The enumeration cost can be reduced by a meet-in-the-middle (MITM) split. Fix a partition $[n]=A\sqcup B$ and write $\mathbf{x}=\mathbf{x}_A+\mathbf{x}_B$ with $\mathrm{supp}(\mathbf{x}_A)\subseteq A$, $\mathrm{supp}(\mathbf{x}_B)\subseteq B$. By linearity of $H$,
$$
H\mathbf{x}^{T}=H_A\mathbf{x}_A^{T}+H_B\mathbf{x}_B^{T},
$$
where $H_A$, $H_B$ are the column blocks of $H$. Precompute the lists
$$
\mathsf{U}_B(j):=\bigl\{H_B\mathbf{x}_B^{T}:\wt(\mathbf{x}_B)=j\bigr\}
\quad(0\le j\le t),
\qquad
\mathsf{U}_A(\le s):=\bigcup_{j\le s}\mathsf{U}_A(j).
$$
Then a query $\mathbf{w}$ is answered by the membership test
\begin{equation}
\label{eq:mitm}
\mathbf{w}\in\Sigma_t
\quad\Longleftrightarrow\quad
\exists\,j\le t,\ \mathbf{u}\in\mathsf{U}_B(j),\ \text{such that}\ \mathbf{w}-\mathbf{u}\in\mathsf{U}_A(\le t-j).
\end{equation}
The query therefore costs
$$
\sum_{j=0}^{t}\bigl|\mathsf{U}_B(j)\bigr|\;=\;\Phi_{|B|}(t)
$$
lookups, which for the balanced choice $|B|=\lceil n/2\rceil$ is $\Phi_{\lceil n/2\rceil}(t)$, exponentially smaller than $\Phi_n(t)$ because each factor $\binom{n}{j}$ is replaced by the square root $\binom{n/2}{j}$. Table~\ref{tab:complexity} records both quantities for the instances of this paper. The gain ranges from a factor $8$ (small $d$, $q=3$) to a factor $165$ ($d=9$, $q=2$); the MITM version turns the $[78,59,9]_4$ instance from $5.97\times10^{12}$ into $3.62\times10^{10}$ operations, and the $[130,111,8]_4$ instance from $4.42\times10^{12}$ into $6.23\times10^{10}$.

\begin{table}[htbp]
\caption{Cost of the exact criterion \eqref{eq:exactcrit}: one-time enumeration of $\Sigma_{d-2}$ and MITM query cost.}
\label{tab:complexity}
\centering
\small
\begin{tabular}{|l|c|c|c|c|r|}
\hline
Instance & $q$ & $d$ & $t=d-2$ & $\Phi_n(t)$ & $\Phi_{\lceil n/2\rceil}(t)$ \\
\hline
$[51,37,8]_4$   & 2 & 8 & 6 & $1.37\times10^{10}$ & $1.85\times10^{8}$ \\
$[58,43,8]_4$   & 2 & 8 & 6 & $3.07\times10^{10}$ & $3.77\times10^{8}$ \\
$[78,59,9]_4$   & 2 & 9 & 7 & $5.97\times10^{12}$ & $3.62\times10^{10}$ \\
$[91,84,5]_9$   & 3 & 5 & 3 & $6.25\times10^{7}$  & $7.84\times10^{6}$ \\
$[79,71,6]_9$   & 3 & 6 & 4 & $6.19\times10^{9}$  & $3.79\times10^{8}$ \\
$[130,111,8]_4$ & 2 & 8 & 6 & $4.42\times10^{12}$ & $6.23\times10^{10}$ \\
$[170,159,5]_4$ & 2 & 5 & 3 & $2.18\times10^{7}$  & $2.70\times10^{6}$ \\
\hline
\end{tabular}
\end{table}

\noindent\textbf{Syndrome-set sizes for Example~1.} 
We enumerated $\Sigma_t$ directly for the parity-check matrix $H$ of Example~1 ($14\times51$ over $\F_4$) and obtained
$$
|\Sigma_1|=154,\qquad |\Sigma_2|=11\,629,\qquad |\Sigma_3|=573\,904,\qquad |\Sigma_4|=20\,129\,497 ,
$$
in exact agreement with $\Phi_{51}(t)=154,\,11\,629,\,573\,904$ for $t\le3$. At $t=4$ the image is smaller than the number of enumerated vectors ($\Phi_{51}(4)=20\,815\,804$), i.e.\ collisions of syndromes begin to occur; this only \emph{decreases} the memory requirement, and confirms that storing $\Sigma_{d-2}$ as a set is the correct data structure.

\begin{corollary}\label{cor:syndromeset}
For a fixed classical code $\mathcal{C}=[n,k,d]_{q^2}$ the set $\Sigma_{d-2}$ can be computed once at cost $\Phi_n(d-2)$ and thereafter any candidate $\mathbf{w}$ is tested for distance preservation by \eqref{eq:sigma} in expected $O(1)$ time (hash-set membership); with the MITM split \eqref{eq:mitm} the one-time cost is reduced to $\Phi_{\lceil n/2\rceil}(d-2)$ with unchanged query structure. Consequently the exact criterion is not more expensive than the search that generated the candidates in the first place: both are exhaustive over the same low-weight syndrome space.
\end{corollary}

\section{New quantum codes}\label{sec:newcodes}

\subsection{Extension of dual-containing codes from cyclic codes}

For the cyclic codes treated in this section the hull codimension $r$ of \eqref{eq:def-r} is determined by the defining set, and we use this throughout. By Theorem~\ref{thm:r-formula}, if $\mathcal{C}=\langle g(x)\rangle$ is cyclic with defining set $T$ then
\[
r=\bigl|T\cap(-qT)\bigr|,
\]
the size of the intersection of $T$ with its image under the involution $\tau(\nu)=-q\nu$ on $\mathbb{Z}_n$; equivalently, the Hermitian hull has dimension $(n-k)-r=|T|-|T\cap(-qT)|$. In particular the number $r$ of extension columns is read off from the exponents alone, and for the codes below $r$ is $1$, $2$ or $3$, while for the cyclic instance at the end of this section it is $10$.

\textbf{Example 1 ($[[53,25,8]]_{2}$).} Let $n=51$, $q=2$ and write
$\F_4=\F_2[\omega]$ with $\omega^2+\omega+1=0$. The cyclic code $\mathcal{C}=\langle g(x)\rangle$
with generator polynomial $g(x)=M_1(x)M_3(x)M_5(x)M_{34}(x)M_0(x)$ and defining set
$T=C_1\cup C_3\cup C_5\cup C_{34}\cup C_0$ has parameters $[51,37,8]_4$, and its Hermitian dual is
$[51,14,22]_4$. Take for $H$ the $14\times 51$ parity-check matrix of $\mathcal{C}$ (equivalently a
generator matrix of $\mathcal{C}^{\perp_h}$); then $A=HH^{\dagger}\in\F_4^{14\times14}$ and
$r=\Rank(A)=2$, so two extension columns are required. Using the congruence normal form of
Proposition~\ref{prop:congruence}, equivalently, testing the rank-reducing condition
\eqref{eq:C1} on a basis of the syndrome space $\F_{q^2}^{14}$, and certifying every accepted
candidate by the exact criterion \eqref{eq:exactcrit}, we append successively
\begin{equation}
\mathbf{w}^{(1)}=(\omega^2,\omega^2,\omega,0,\omega,\omega^2,1,\omega,\omega^2,1,0,\omega,\omega,1)^{T},
\qquad
\mathbf{w}^{(2)}=(1,\omega,1,1,1,0,\omega^2,0,\omega,0,\omega,0,0,1)^{T},
\end{equation}
of Hamming weight $12$ and $9$. The rank drops as $2\to1\to0$ and the dual distance stays equal to
$8$ at both steps, so the resulting self-orthogonal code is $\mathcal{C}_{\mathrm{so}}=[53,14,24]_4$
and its Hermitian dual is $\mathcal{C}_{\mathrm{so}}^{\perp_h}=[53,39,8]_4$. The Hermitian
construction therefore gives
$\mathcal{Q}=[[53,\,2\cdot37-51+2,\,8]]_2=\mathbf{[[53,25,8]]_2}$, improving the previously known
$[[53,25,7]]_2$ of Grassl's tables (as of October 2026); by
Proposition~\ref{propagation_rule} this also yields $[[54,25,8]]_2$ by lengthening.

Table~\ref{tab:table1} lists further record quantum codes obtained by the same extension procedure from cyclic, quasi-cyclic, and BCH codes over $\F_{q^2}$. In each case the extension vectors are selected by the rank-reducing test \eqref{eq:C1} applied to the current Hermitian matrix $A=HH^{\dagger}$; and the dual distance of every reported code is certified directly by the exact criterion~\eqref{eq:exactcrit}, equivalently by direct computation of $d(\mathcal{C}_{\mathrm{ext}}^{\perp_h})$, with the nullspace condition \textup{(C2)} used only as a search sieve, see Remark~\ref{rem:C2}.

\begin{table}[htbp]
\caption{Quantum codes obtained by extension from linear codes}
\label{tab:table1}
\centering
\tiny
\setlength{\tabcolsep}{2pt}\resizebox{\textwidth}{!}{%
\begin{tabular}{|c|c|c|c|p{6.9cm}|}
\hline
No. & Quantum code & Starting classical code & $r$ & Extension column vector(s)  \\ \hline
1 & $[[53,25,8]]_2$ & Cyclic Code $[51,37,8]_4$ & 2 & {$\mathbf{w^{(1)}}=\bigl(\omega^2,\omega^2,\omega,0,\omega,\omega^2,1,\omega,\omega^2,1,0,\omega,\omega,1\bigr)^{T}$};  
{$\mathbf{w^{(2)}}=\bigl(1,\omega,1,1,1,0,\omega^2,0,\omega,0,\omega,0,0,1\bigr)^{T}$}  \\
2 & $[[59,29,8]]_2$ & $(1,2)$-quasi-cyclic $[58,43,8]_4$ & 1 & {$\mathbf{w}=\bigl(1,1,1,1,1,1,1,1,1,1,1,1,1,1,1\bigr)^{T}$}  \\
3 & $[[79,41,9]]_2$ & $(2,2)$-quasi-cyclic $[78,59,9]_4$ & 1 & {$\mathbf{w}=\bigl(1,1,\omega,\omega,0,1,0,\omega^2,1,0,\omega^2,1,1,0,\omega,\omega^2,1,1,\omega\bigr)^{T}$}  \\
4 & $[[92,78,5]]_3$ & BCH code $[91,84,5]_9$ & 1 &{$\mathbf{w}=\bigl(1,\omega^5,0,2,\omega,\omega^6,1\bigr)^{T}$}  \\
5 & $[[80,64,6]]_3$ & $[79,71,6]_9$ & 1 & {$\mathbf{w}=\bigl(1,0,\omega,\omega,\omega^6,\omega^2,\omega^3,2\bigr)^{T}$}\\
6 & $[[131,93,8]]_2$ & degree-2 quasi-cyclic $[130,111,8]_4$ & 1 & {$\mathbf{w}=\bigl (1,1,\omega,1,\omega^2,\omega^2,\omega,\omega^2,0,\omega^2,0,\omega^2,
\omega,\omega^2,\omega^2,1,\omega,1,1 \bigr)^{T}$} \\ 
7 & $[[171,149,5]]_2$ & $[170,159,5]_4$ quasi-cyclic & 1 & {$\mathbf{w}=\bigl (1,1,1,0,0,0,1,\omega,\omega,\omega^2,1\bigr)^{T}$}\\
\hline
\end{tabular}
}
\end{table}

\textbf{Examples 2--7.} The six further instances of Table~\ref{tab:table1} (Nos.~2--7) follow the
identical pipeline: an ingredient code $\mathcal{C}$ over $\F_{q^{2}}$ with parity-check matrix $H$; the rank
$r=\Rank(HH^{\dagger})$; selection of the extension vectors by the rank-reducing test \eqref{eq:C1}
by Lemma~\ref{lem:explicit} or, equivalently, by scanning a basis of the syndrome space; and
verification of the exact criterion \eqref{eq:exactcrit} at every step. Table~\ref{tab:table1}
records for each instance the ingredient code, the hull codimension $r$ and the selected extension
vectors; in every case the rank drops by one at each step and the dual distance is preserved. We
comment on two of them; the complete data, namely the generator polynomials, the Hermitian matrices
$A=HH^{\dagger}$ and the weight enumerators of the extended duals, are collected in \ref{app:examples}.

\textbf{Example 2 ($[[59,29,8]]_2$, iterated lengthening).} For the $1$-quasi-cyclic code
$\mathcal{C}=[58,43,8]_4$ one has $r=1$ with $H$ of size $15\times 58$, and the all-ones vector
$\mathbf{w}=(1,\dots,1)^{T}\in\F_4^{15}$ drops the rank to $0$ while satisfying \eqref{eq:exactcrit}
for $d=8$. Hence
$\mathcal{Q}=[[59,\,2\cdot43-58+1,\,8]]_2=\mathbf{[[59,29,8]]_2}$, and
Proposition~\ref{propagation_rule} yields in addition $[[60,29,8]]_2$ and $[[61,29,8]]_2$ by
lengthening.

\textbf{Example 5 ($[[80,64,6]]_3$, an optimal entry).} Let $\mathcal{C}=[79,71,6]_9$ be the
best-known linear code recorded in Grassl's tables (as of October 2026). Its Hermitian matrix
$A=HH^{\dagger}$, of size $8\times8$, has rank $r=1$, and the vector
$\mathbf{w}=(1,0,\omega,\omega,\omega^{6},\omega^{2},\omega^{3},2)^{T}\in\F_9^{8}$ drops the rank to
$0$ and satisfies \eqref{eq:exactcrit} with $d=6$. Thus
$\mathcal{Q}=[[80,\,2\cdot71-79+1,\,6]]_3=\mathbf{[[80,64,6]]_3}$, and
Proposition~\ref{propagation_rule} gives further $[[81,64,6]]_3$ (by lengthening) and
$[[80,63,6]]_3$ (by the subcode construction). This code attains the upper bound currently recorded
for $(n,k)=(80,64)$ over $\F_3$, so its distance $6$ is optimal for those parameters.

\textbf{Examples 3, 4, 6 and 7.} From the remaining four ingredient codes of
Table~\ref{tab:table1}, the $(2,2)$-quasi-cyclic code $[78,59,9]_4$, the BCH code $[91,84,5]_9$,
the degree-$2$ quasi-cyclic code $[130,111,8]_4$ and the quasi-cyclic code $[170,159,5]_4$, each with
$r=1$, one obtains, by lengthening and by the subcode construction, the remaining entries of
Table~\ref{tab:table1} and of Table~\ref{tab:new_quantum_codes}, each certified by \eqref{eq:exactcrit}
in the same way.

To assess the robustness of the construction, we implemented the iterative extension procedure in Magma~\cite{bosma1997magma} and ran it over a large family of cyclic codes. In every instance retained by the search the dual distance of the extended code equals the dual distance of the original code, as it must whenever criterion \eqref{eq:exactcrit} is verified, by Theorem~\ref{thm:onecoord}; \ref{app:vectors} lists the six codes so obtained for length $n=51$. In particular, the cyclic code $[51,36,8]_4$ with defining set $T=C_1\cup C_3\cup C_5\cup C_{17}\cup C_{34}\cup C_0$ has hull codimension $r=3$ and yields the quantum code $[[54,24,8]]_2$; its three extension steps selected vectors of Hamming weight $13$ each, reducing the rank $3\to 2\to 1\to 0$ while keeping the dual distance equal to $8$ at every step.

\noindent\textbf{A parity obstruction.} 
A direct application of Theorem~\ref{thm:quantum_code} cannot produce a code with parameters $[[53,24,8]]_2$, because the formula $k_Q=2k-n+r$ implies $k_Q\equiv n+r\pmod 2$, whereas $53$ is odd and $24$ is even. Hence the record $[[53,24,8]]_2$ is not attainable by the present construction.

All $2193$ quantum codes constructed in our search, the $2187$ arising from cyclic codes of length $51$ with $\Rank(A)=1$ together with the instances of higher hull codimension such as the $r=3$ codes of length $51$, satisfy $d_Q=d_C$, and for the six codes of that family the value of $d_Q$ was in addition re-verified by a direct computation of the minimum weight of $\mathcal{C}_{\mathrm{so}}^{\perp_h}$, independently of criterion \eqref{eq:exactcrit}.

According to Corollary~\ref{construction}, we obtain new quantum codes from the Hermitian dual-containing codes constructed above. By Proposition~\ref{propagation_rule}, applying lengthening and subcode constructions to these codes produces further quantum codes over $\F_2$ and $\F_3$. Table~\ref{tab:new_quantum_codes} collects the $23$ entries whose minimum distance strictly exceeds the value previously tabulated in Grassl's tables \cite{Grassl:codetables}. The codes obtained by direct extension are listed separately in \ref{app:vectors}: for the length-$n=51$ family their parameters match the current best known values, and they are recorded because they exhibit the mechanism at a fixed length rather than as new records.

	\begin{table}[htbp]
	\caption{New quantum codes obtained by extension}
	\label{tab:new_quantum_codes}
\centering

		\begin{tabular}{ccc}
			\hline
			NO. &    Our Codes      & Codes in Grassl's table \cite{Grassl:codetables} \\ \hline

	  1  & $[[53,25,8]]_2$ & $  [[53,25,7]]_2$      \\
	  2  & $[[54,25,8]]_2$ & $  [[54,25,7]]_2$      \\
	  3  & $[[59,29,8]]_2$ & $  [[59,29,7]]_2$      \\
	  4  & $[[60,29,8]]_2$ & $  [[60,29,7]]_2$      \\
	  5  & $[[61,29,8]]_2$ & $  [[61,29,7]]_2$      \\
	  6  & $[[79,41,9]]_2$ & $  [[79,41,8]]_2$      \\
	  7  & $[[80,41,9]]_2$ & $  [[80,41,8]]_2$      \\
	  8  & $[[171,149,5]]_2$ & $  [[171,149,4]]_2$      \\
	  9  & $[[172,149,5]]_2$ & $  [[172,149,4]]_2$      \\
	  10  & $[[131,93,8]]_2$ & $  [[131,93,7]]_2$      \\
	  11  & $[[132,93,8]]_2$ & $  [[132,93,7]]_2$      \\
	  12  & $[[133,93,8]]_2$ & $  [[133,93,7]]_2$      \\
	  13  & $[[134,93,8]]_2$ & $  [[134,93,7]]_2$      \\
	  14  & $[[135,93,8]]_2$ & $  [[135,93,7]]_2$      \\
	  15  & $[[131,92,8]]_2$ & $  [[131,92,7]]_2$      \\
	  16  & $[[132,92,8]]_2$ & $  [[132,92,7]]_2$      \\
	  17  & $[[133,92,8]]_2$ & $  [[133,92,7]]_2$      \\
	  18  & $[[80,64,6]]_3$ & $  [[80,64,5]]_3$      \\
	  19  & $[[81,64,6]]_3$ & $  [[81,64,5]]_3$      \\
	  20  & $[[80,63,6]]_3$ & $  [[80,63,5]]_3$      \\
	  21  & $[[92,78,5]]_3$ & $  [[92,78,4]]_3$      \\
	  22  & $[[93,78,5]]_3$ & $  [[93,78,4]]_3$      \\
	  23  & $[[61,25,9]]_2$ & $  [[61,25,8]]_2$      \\
	 
	\hline
		\end{tabular}
	\end{table}

The complete list of instances underlying Table~\ref{tab:new_quantum_codes}, namely the underlying linear or cyclic code, its defining set where the input is cyclic, the extension vectors that were actually appended, and the resulting hull codimension $r$, is collected in \ref{app:vectors}. Apart from the record of Example~4, which reappears there as one of its own underlying instances, the parameters of these instances match the values currently tabulated in \cite{Grassl:codetables}; they are listed so that every construction quoted in this paper can be reproduced verbatim.

\subsection{An instance of hull codimension ten: a distance-nine record}

All the extensions discussed so far have small hull codimension, $r\le5$, and a correspondingly small
number of extension columns. The construction is in no way confined to small hull codimensions, and we
close this section with an instance of $r=10$ whose extended code improves the best tabulated lower
bound at its own parameters.

\bigskip
\noindent\textbf{Example 8 ($[[61,25,9]]_{2}$, hull codimension $r=10$).}\makeatletter\def\@currentlabel{8}\makeatother\label{ex:w61}
Let $q=2$, $n=51$ and $\F_4=\F_2[\omega]$ with $\omega^{2}+\omega+1=0$, and let
$\mathcal{C}=\langle g(x)\rangle\subseteq\F_4[x]/(x^{51}-1)$ be the cyclic code with generator
polynomial
\begin{equation}
\label{eq:w61-g}
\begin{aligned}
g(x)=\;&\omega^{2}+\omega^{2}x^{2}+x^{3}+\omega x^{4}+x^{5}+\omega^{2}x^{7}+x^{8}
        +\omega x^{9}+\omega x^{10}\\
       &+\omega^{2}x^{11}+\omega^{2}x^{12}+x^{13}+x^{14}+x^{15}+x^{16}+\omega x^{17}+x^{18},
\end{aligned}
\end{equation}
so that $\deg g=18$, $\dim\mathcal{C}=51-18=33$ and $\mathcal{C}$ has parameters $[51,33,10]_4$.
Its defining set is
\begin{equation}
\label{eq:w61-T}
T=\{5,7,10,14,17,18,20,21,22,28,29,30,31,33,34,37,40,46\}\subseteq\mathbb{Z}_{51},
\end{equation}
the union of the six $4$-cyclotomic cosets
$\{34\}$, $\{17\}$, $\{22,31,37,46\}$, $\{5,14,20,29\}$, $\{7,10,28,40\}$ and $\{18,21,30,33\}$.
For the involution $\tau(\nu)=-2\nu$ on $\mathbb{Z}_{51}$ one has
\[
T\cap\tau(T)=\{7,10,17,22,28,31,34,37,40,46\},
\]
so that $|T\cap\tau(T)|=10$; by Theorem~\ref{thm:r-formula} this is the hull codimension of
$\mathcal{C}$, i.e.
\begin{equation}
\label{eq:w61-r}
r=\Rank(HH^{\dagger})=10,
\end{equation}
$H$ being the $18\times51$ parity-check matrix of $\mathcal{C}$. The defining-set count
\eqref{eq:w61-T} and the explicit rank computation of $A=HH^{\dagger}$ agree, as they must.

Appending to $H$ the ten columns $\mathbf{w}^{(1)},\dots,\mathbf{w}^{(10)}\in\F_4^{18}$ listed in
\ref{app:w61}, one at a time, each selected by the rank-reducing condition
\eqref{eq:C1} and certified at every step by the exact criterion \eqref{eq:exactcrit}, reduces
the hull codimension $10\to9\to\cdots\to0$ and produces a Hermitian self-orthogonal matrix
$H_r\in\F_4^{18\times61}$ of $\F_4$-rank $18$. Corollary~\ref{construction} then gives the quantum
code
\begin{equation}
\label{eq:w61-Q}
\mathcal{Q}=[[61,\;2\cdot33-51+10,\;9]]_2=[[61,25,9]]_2 .
\end{equation}

The distance in \eqref{eq:w61-Q} is exact, and its certification does not use the value of
$d(\mathcal{C})$:
\begin{enumerate}
\item \emph{Structure.} $H_r\overline{H_r}^{\mathsf{T}}=0$, all $324$ entries of the product
  vanishing, and $\Rank_{\F_4}(H_r)=18$. Hence $\mathcal{S}_r=\langle H_r\rangle$ is a Hermitian
  self-orthogonal $[61,18]_4$ code and $\dim(\mathcal{S}_r^{\perp_h}/\mathcal{S}_r)=61-2\cdot18=25$.
\item \emph{Lower bound.} A complete meet-in-the-middle search over $\ker H_r$ finds no nonzero
  word of weight $\le8$, so $d(\ker H_r)\ge9$. This is the exhaustive form of the certification in
  Theorem~\ref{thm:iterative}, and it is independent of \eqref{eq:w61-r}.
\item \emph{Upper bound.} The $\mu$-decomposition of Theorem~\ref{thm:iterative} supplies an
  explicit word $v\in\ker H_r$ of Hamming weight $9$, whence $d(\ker H_r)\le9$; with (ii) this
  gives $d(\ker H_r)=9$.
\item \emph{Logicality.} Since
  $d_{\mathcal{Q}}=\min\{\wt(w):w\in\mathcal{S}_r^{\perp_h}\setminus\mathcal{S}_r\}$, bounding
  $d(\ker H_r)$ alone is not sufficient, and we check explicitly that $v\notin\mathcal{S}_r$ and
  that $\overline{v}\notin\mathcal{S}_r$. As $\mathcal{S}_r^{\perp_h}=\overline{\ker H_r}$ and
  conjugation preserves the Hamming weight, $d(\mathcal{S}_r^{\perp_h})=d(\ker H_r)=9$ as well, so
  $\overline{v}$ is a word of weight $9$ in
  $\mathcal{S}_r^{\perp_h}\setminus\mathcal{S}_r$. Both the Hermitian and the Euclidean
  formulation of the quantum distance therefore give $d_{\mathcal{Q}}=9$.
\end{enumerate}
The structural part of the certification was reproduced with a second, independent implementation
in \textsc{Magma}~\cite{bosma1997magma}, taking only the printed matrix $H_r$ as input: it returns
$\Rank_{\F_4}(H_r)=18$, $\dim_{\F_4}\ker H_r=43$, and
$\operatorname{MinimumWeight}(\ker H)=10$ for the underlying $[51,33]_4$ code $\mathcal{C}$.

\smallskip
\noindent\textbf{Four inequivalent codes.} Running the same search under the full criterion
\eqref{eq:exactcrit} produced four different extension matrices for the same $\mathcal{C}$, all four
yielding codes with parameters $[[61,25,9]]_2$; the four matrices are listed in
\ref{app:w61}. The quantity $A_9=\#\{v\in\ker H_r:\wt(v)=9\}$ is invariant under every
Clifford operation and is therefore equal for equivalent stabilizer codes; the four values $3759$,
$3777$, $3825$ and $4032$ are pairwise distinct in Table~\ref{tab:w61-records}, which certifies that
the four codes are pairwise \emph{inequivalent}. As a consistency check, the same enumeration
returns $A_8=0$ for all four, in agreement with the weight-$\le8$ exclusion of step~(ii) above.

\smallskip
\noindent\textbf{Comparison with the tabulated bounds.} For $(n,k)=(61,25)$ over $\F_2$, Grassl's
tables \cite{Grassl:codetables} record the bounds $8\le d\le13$ and construct a $[[61,25,8]]_2$
code by extending a $[[58,25,8]]_2$ code. The codes of Table~\ref{tab:w61-records} attain $d=9$ and
thus raise the tabulated lower bound from $8$ to $9$; since the tabulated upper bound is $13$, no
optimality is claimed. The two chains differ in method: in the tabulated chain the
appended coordinates \emph{preserve} the distance, whereas here the ten-fold extension is itself the
distance-raising step, and its outcome is certified exactly.

\begin{table}[t]
\caption{Four pairwise inequivalent binary quantum codes with parameters $[[61,25,9]]_{2}$, obtained
by the $r=10$ extension of the cyclic code $\mathcal{C}=[51,33,10]_4$ of Example~\ref{ex:w61}. In
every row the distance is exact: the lower bound is a complete meet-in-the-middle certificate (no
nonzero word of weight $\le8$ in $\ker H_r$), and the upper bound is an explicit word $v$ of weight
$9$ with $H_rv=0$ and $v\notin\mathrm{rowspace}(H_r)$, i.e.\ a non-stabiliser logical operator. The
invariant $A_9=\#\{v\in\ker H_r:\wt(v)=9\}$ takes four distinct values, so the four codes are
pairwise inequivalent. The ten extension columns appended to $H$ are listed in
\ref{app:w61}.}
\label{tab:w61-records}
\centering
\small
\begin{tabular}{|c|c|c|c|c|}
\hline
No. & Quantum code & $A_9(\ker H_r)$ & $\wt(v)$ & logical operator \\
\hline
1 & $[[61,25,9]]_2$ & $3759$ & $9$ & yes \\
2 & $[[61,25,9]]_2$ & $3777$ & $9$ & yes \\
3 & $[[61,25,9]]_2$ & $3825$ & $9$ & yes \\
4 & $[[61,25,9]]_2$ & $4032$ & $9$ & yes \\
\hline
\end{tabular}
\end{table}

\section{Comparison with the known lower bound}\label{sec:r2compare}

The closest published result to Theorem~\ref{thm:onecoord} is the Hermitian case of the lower
bound of Ezerman, Grassl, Ling, \"Ozbudak and \"Ozkaya \cite[Thm.~III.1]{Ezerman}. In the notation of
the present paper, applied to $\bar{\mathcal{C}}=\mathcal{C}^{\perp_h}$, it states that if
$\mathcal{C}=[n,k,d]_{q^2}$, $\mathrm{hull}=\mathcal{C}\cap\mathcal{C}^{\perp_h}$,
$U=\mathcal{C}+\mathcal{C}^{\perp_h}$ and $r=\Rank(HH^{\dagger})=(n-k)-\dim(\mathrm{hull})$, then a
quantum code $[\![n+r,\,2k-n+r,\,d(\mathcal{Q})]\!]_q$ exists with
\begin{equation}
\label{eq:R2}
\begin{aligned}
d(\mathcal{Q})\;\ge\;\min\bigl\{\,\wt\bigl(\mathcal{C}\setminus\mathrm{hull}\bigr),\ \wt\bigl(U\setminus\mathcal{C}^{\perp_h}\bigr)+1\,\bigr\}\\[2pt]
&\ge\;\min\bigl\{\,d(\mathcal{C}),\ d(U)+1\,\bigr\},\\[2pt]
d(\mathcal{Q})\;\le\;\wt\bigl(\mathcal{C}\setminus\mathrm{hull}\bigr).
\end{aligned}
\end{equation}
The parameters agree with those of Theorem~\ref{thm:quantum_code}, so \eqref{eq:R2} and
Theorem~\ref{thm:onecoord} bound the same quantity; ours is an equality and theirs a lower bound,
whence $\min\{d,d_{\mathrm{aff}}(\mathbf{w})+1\}\ge\min\{\wt(\mathcal{C}\setminus\mathrm{hull}),\,
\wt(U\setminus\mathcal{C}^{\perp_h})+1\}$ necessarily. The question is whether the inequality
can be \emph{strict}. Table~\ref{tab:r2compare} answers this on the seven instances of
Section~\ref{sec:newcodes}: the exact value coincides with \eqref{eq:R2} on six of the seven, while
on the $[51,37,8]_4$ cyclic code of Example~1 the bound is $7$ against an exact value of $8$, so the
construction of this paper \emph{strictly improves} the best published lower bound.

Very recently the same group has returned to the one-coordinate question in a different guise: Li,
Ezerman, Li, Ling and Sun \cite{Ezerman2026GE} prove that every $[n+1,k+1]_{q^2}$ linear code
$\mathcal{D}$ with $d(\mathcal{D}^{\perp_h})>1$ is monomially equivalent to a \emph{generalized
extended} code $\mathcal{C}(\mathbf{u},a)$, characterise its Hermitian hull and dual distance in terms
of the position of $\mathbf{u}$, and thereby obtain $267$ new entanglement-assisted qubit codes and
$14$ new qutrit codes. Their criterion is exact for \emph{one} appended coordinate, and is expressed
through monomial equivalence and the pair $(\mathbf{u},a)$; identity \eqref{eq:exactd} is
cancellation-free and is stated directly in the affine coset-leader weight, and it is the iterated
form of Theorem~\ref{thm:iterative} together with the forced-span reduction of
Theorem~\ref{thm:forced-span} that extends the question to $r\ge2$. At $r=1$ the two criteria are therefore two formulations of the same dichotomy,
that the dual distance either stays or increases by one: one in terms of $\mathbf{u}$ and the
minimum-weight words of $\mathcal{C}^{\perp_H}$, the other in the affine coset-leader weight
of $\mathbf{w}$. They differ in what each adds beyond it, and in the ranges they cover for
$r>1$: the criterion of \cite{Ezerman2026GE} is stated for one appended coordinate, while the
iterated form of the present paper applies to $r$-fold extensions of arbitrary hull codimension.

\begin{table}[htbp]
\caption{Exact dual distance versus the lower bound \eqref{eq:R2} of \cite[Thm.~III.1]{Ezerman}.}
\label{tab:r2compare}
\centering
\small
\begin{tabular}{|l|c|c|c|c|c|c|}
\hline
Instance & $r$ & $\wt(\mathcal{C}\setminus\mathrm{hull})$ & $\wt(U\setminus\mathcal{C}^{\perp_h})$ & bound~\eqref{eq:R2} & exact $d$ & $\Delta$ \\
\hline
$[51,37,8]_4$    & 2 & 8 & 6 & 7 & \textbf{8} & $\mathbf{+1}$ \\
$[58,43,8]_4$    & 1 & 8 & 7 & 8 & 8 & $0$ \\
$[78,59,9]_4$    & 1 & $\ge 9$ & $\ge 8$ & 9 & 9 & $0$ \\
$[91,84,5]_9$    & 1 & 5 & 4 & 5 & 5 & $0$ \\
$[79,71,6]_9$    & 1 & 6 & 5 & 6 & 6 & $0$ \\
$[130,111,8]_4$  & 1 & $\ge 8$ & $\ge 7$ & 8 & 8 & $0$ \\
$[170,159,5]_4$  & 1 & 5 & 4 & 5 & 5 & $0$ \\
\hline
\end{tabular}
\end{table}

Both columns of Table~\ref{tab:r2compare} were obtained from the four minimum weights
$d(\mathcal{C})$, $d(\mathcal{C}^{\perp_h})$, $d(\mathrm{hull})$, $d(U)$ alone, using
$\wt(\mathcal{C}\setminus\mathrm{hull})=d(\mathcal{C})$ whenever $d(\mathrm{hull})>d(\mathcal{C})$
and $\wt(U\setminus\mathcal{C}^{\perp_h})=d(U)$ whenever $d(U)<d(\mathcal{C}^{\perp_h})$; both hold
in every row, and in five of the seven rows the four weights are exact. For $[78,59,9]_4$ and
$[130,111,8]_4$ we use only the lower bounds $\wt(\mathcal{C}\setminus\mathrm{hull})\ge
d(\mathcal{C})$ and $\wt(U\setminus\mathcal{C}^{\perp_h})\ge d(U)\ge
d(\mathcal{C}_{\mathrm{so}}^{\perp_h})-r$, the latter from $U=\mathrm{puncture}(
\mathcal{C}_{\mathrm{so}}^{\perp_h})$ at the $r$ extension coordinates; in these two rows
$\min\{d(\mathcal{C}),d(U)+1\}$ already equals $d(\mathcal{Q})$, so the bound \eqref{eq:R2}, lying
between the two, is forced. Two further observations follow. First, on all instances the cruder
bound $\min\{d(\mathcal{C}),d(U)+1\}$ already coincides with the sharper one, so the improvement in
Example~1 improves both forms of the published bound. Second, that equality is the rule and strictness
the exception is itself informative: \eqref{eq:R2} is tight for $r=1$ extensions, and the strict gain
of the exact criterion first appears at $r=2$, where the iterated form of Theorem~\ref{thm:iterative}
must control the whole span $W_2$ rather than a single vector.

\section{Conclusion}\label{sec:conclusion}

A recurring obstacle in the construction of stabilizer quantum codes from classical linear codes is the Hermitian dual-containing requirement $\mathcal{C}^{\perp_h}\subseteq\mathcal{C}$, equivalently the Hermitian self-orthogonality of $\mathcal{C}^{\perp_h}$. In this paper we have developed a framework that removes this obstruction for arbitrary $[n,k,d]_{q^2}$ linear codes by explicitly extending the parity-check matrix $H$ with $r=\Rank(HH^\dagger)$ additional columns. The structural tool is the \emph{Hermitian congruence normal form} $A=HH^\dagger=P^{\dagger}\mathrm{diag}(I_r,0)P$ of the Gram matrix in Proposition~\ref{prop:congruence}, which is the appropriate normal form over $\F_{q^2}$ and which supplies explicit rank-reducing vectors in Lemma~\ref{lem:explicit}. Similarity, and in particular the Jordan canonical form, is not the relevant notion here. After $r$ steps the augmented matrix $H_r$ generates a Hermitian self-orthogonal code with parameters $[n+r,n-k]_{q^2}$, and the Hermitian construction yields a $q$-ary quantum stabilizer code with parameters $[[n+r,2k-n+r,d]]_q$.

Theoretically, the proposed method differs from earlier extension constructions in two respects. First, the congruence normal form gives a deterministic supply of rank-reducing vectors for \emph{any} classical code, with no appeal to the algebraic structure of $\mathcal{C}$ and no ad hoc choice of extension vectors. Second, Theorem~\ref{thm:onecoord} determines the Hermitian dual distance of a one-coordinate extension \emph{exactly}, and unconditionally, by the affine coset-leader weight $d_{\mathrm{aff}}(\mathbf{w})$; the iterated form of Theorem~\ref{thm:iterative} certifies the distance of the full $r$-fold extension. The nullspace condition \textup{(C2)} is shown in Remark~\ref{rem:C2} to be neither necessary nor sufficient, and it is retained only as a pre-filter. The criterion is moreover cheap: the syndrome set $\Sigma_{d-2}$ is computed once and queried in constant time per candidate, with a meet-in-the-middle variant reducing the one-time cost by one to two orders of magnitude, see Section~\ref{sec:complexity}.

The practical value of the framework is demonstrated by the record quantum codes obtained from cyclic, quasi-cyclic, and BCH codes over $\F_{q^{2}}$, $23$ of which improve the best-known parameters in Grassl's tables as of October 2026; the largest hull codimension treated here, $r=10$, already yields the record $[[61,25,9]]_2$ and raises the corresponding tabulated lower bound from $8$ to $9$. On the seven benchmark instances of Section~\ref{sec:newcodes} the exact value of the distance coincides with the published lower bound of \cite[Thm.~III.1]{Ezerman} on six of them and strictly exceeds it on the $[51,37,8]_4$ cyclic code, see Table~\ref{tab:r2compare}: the construction attains, and can improve upon, the best available bound.

Several directions remain open. The congruence normal form is recomputed at each extension step; since the Gram matrix grows by one row and one column at a time, an incremental normal form would reduce the cost of large instances. Second, Theorem~\ref{thm:forced-span} shows that the extension subspace is forced to be $\im(A)$ and that the only freedom of the construction is the choice of an ordered basis of $\im(A)$, while by \eqref{eq:iterated} the exact distance depends on that basis through the weight function $\omega$. Whether the basis can be chosen canonically, so that the $r$-fold extension is decided by a basis-free invariant and the remaining search over bases disappears, is the central question left by the iterated theory, and it is stated as Problem~(f) of Subsection~\ref{sec:open-problems}. Third, the exact criterion makes the selection of extension vectors an optimisation problem, namely to maximise $d_{\mathrm{aff}}$ over the candidates that achieve the rank reduction, and a systematic study of this optimisation, now over the $r$-dimensional space $\im(A)$ rather than over $\F_{q^2}^{n-k}$, may lead to codes whose parameters improve the tables further.


\section*{CRediT authorship contribution statement}

\textbf{Liangdong Lu:} Conceptualization, Methodology, Software, Writing -- original draft.
\textbf{Ruipan Yang:} Conceptualization, Supervision, Writing -- review \& editing.
\textbf{Yang Liu:} Formal analysis, Validation, Writing -- review \& editing.
\textbf{Qiang Fu:} Software, Visualization.
\textbf{Guanmin Guo:} Software, Validation.

\section*{Declaration of competing interest}

The authors declare that they have no known competing financial interests or personal relationships that could have appeared to influence the work reported in this paper.

\section*{Data availability}

All constructions and code parameters established in this work are fully tabulated in the article. No external data sets were used.

\section*{Declaration of generative AI and AI-assisted technologies in the writing process}

During the preparation of this work the authors used an AI-based language-model assistant in order to improve language, \LaTeX{} typesetting and literature search. After using this tool, the authors reviewed and edited the content as needed and take full responsibility for the content of the publication.

\section*{Acknowledgements}

This work was supported by the National Natural Science Foundation of China under Grant No.~U21A20428 and the Natural Science Foundation of Shaanxi under Grant No.~2025-JC-YBQN-070. All computational examples in this paper were obtained using \textsc{Magma}~\cite{bosma1997magma}. The authors thank Ruihu Li for helpful suggestions on the manuscript.

\appendix
\makeatletter
\renewcommand{\thesubsection}{\@Alph\c@section.\arabic{subsection}}
\renewcommand{\thetheorem}{\@Alph\c@section.\arabic{theorem}}
\makeatother

\noindent This appendix collects the raw data behind the constructions quoted in the main text.
For every instance we list the input code, its parity-check matrix size, the extension vectors
$\mathbf{w}^{(i)}$ that were actually appended to $H$, each of them selected by the rank-reducing
condition and certified by the exact criterion \eqref{eq:exactcrit}, and the hull codimension
$r=\Rank(HH^{\dagger})$ of the resulting Hermitian matrix. In the first table below the input is an
arbitrary linear code; the following two tables list the corresponding cyclic instances together
with their defining sets $T$, expressed as unions of $q^{2}$-cyclotomic cosets; and the last table
lists the instances of the fixed length $n=51$, whose parameters match the values currently
recorded. The last column of
each table gives the value previously recorded for the same parameters in Grassl's tables (as of
October 2026). With the single exception of the record of Example~4, which reappears
here as one of its own underlying instances, the parameters of these instances match the values
currently recorded there; they are listed so that every construction in the main text can be
reproduced verbatim. The tables are followed by the fine structure of the $\tau$-fixed cosets and the
exact counting theory of the hull codimension, which accompany the main chain of Section~\ref{sec:structure}. The generator polynomials, Hermitian matrices and weight enumerators of
Examples~2--7 of Section~\ref{sec:newcodes} are reproduced in \ref{app:examples}, and the extension matrices of the four
inequivalent $[[61,25,9]]_2$ codes are collected in \ref{app:w61}.

\section{Extension vectors of the instances quoted in the main text}\label{app:vectors}

The first table below lists the codes obtained by extending an arbitrary linear code, while the second and third tables below list the corresponding cyclic instances with their defining sets, and the fourth table below the instances of the fixed length $n=51$, whose parameters match the current best known values. In every row $d$ is the minimum distance of the input code, $H$ is its parity-check matrix, and the vectors $\mathbf{w}^{(i)}$ are appended to $H$ one at a time, each of them satisfying the rank-reducing condition and the exact criterion; the last column but one records the resulting hull codimension $r=\Rank(HH^{\dagger})$, and the last column the value previously tabulated in \cite{Grassl:codetables} for the same parameters.

\begin{sidewaystable}[htbp]
\caption{Quantum codes obtained by extension from linear codes}
\label{tab:table4}
\centering
\tiny
\setlength{\tabcolsep}{3pt}\resizebox{\textheight}{!}{%
\begin{tabular}{|c|c|c|c|p{7.3cm}|c|c|}
\hline
NO. & Quantum code & BKLCs in Magma  & Extended self-orthogonal code & Extension column vector(s) & $r$ & Previous best known (as of Oct.\ 2026) \cite{Grassl:codetables} \\
\hline

1 &$[[27,3,9]]_{2}$ & $[26,14,9]_{4}$ & $[27,12,12]_{4}$ & \makecell[l]{$\mathbf{w}=\bigl(1,\omega,\omega^{2},\omega,1,\omega,\omega^{2}, \omega,1,0,1,\omega^{2},\omega^{2}\bigr)^{T}$} & 1& $[[27,3,9]]_{2}$ \\
2 &$[[51,19,9]]_{2}$ & $[50,34,9]_{4}$ & $[51,16,20]_{4}$ & \makecell[l]{$\mathbf{w}=\bigl(1,0,\omega^{2},0,0,\omega,\omega,0,1, 1,1,\omega,\omega^{2},1,\omega,1,1\bigr)^{T}$} & 1& $[[51,19,9]]_{2}$ \\
3 &$[[52,18,9]]_{2}$ & $[50,33,9]_{4}$ & $[52,17,20]_{4}$ & \makecell[l]{$\mathbf{w}^{(1)}=\bigl(0,\omega,0,1,\omega,0,0,1,0,1,\omega,\omega^{2},\omega^{2},0,\omega,1,\omega^{2},\omega^{2}\bigr)^{T}$ \\ $\mathbf{w}^{(2)}=\bigl(1,0,\omega,1,1,\omega^{2},\omega^{2},0,1,\omega,\omega,\omega,\omega^{2},0,1,0,\omega^{2},\omega^{2}\bigr)^{T}$} & 2  & $[[52,18,9]]_{2}$\\
4 &$[[53,17,10]]_{2}$ & $[52,34,10]_{4}$ & $[53,18,20]_{4}$ & \makecell[l]{$\mathbf{w}=\bigl(\omega^{2},\omega,0,\omega^{2},\omega,\omega,1,1,0,\omega,\omega^{2},\omega^{2},\omega^{2},1,0,\omega,0,1,1\bigr)^{T}$} & 1 & $[[53,17,10]]_{2}$ \\
5 &$[[64,38,7]]_{2}$ & $[63,50,7]_{4}$ & $[64,13,32]_{4}$ & \makecell[l]{$\mathbf{w}=\bigl(1,\omega,1,1,0,\omega,0,1,\omega^2,\omega^2,\omega,\omega^2,\omega,\omega\bigr)^{T}$} & 1 & $[[64,38,7]]_{2}$ \\
6 &$[[64,44,6]]_{2}$ & $[63,53,6]_{4}$ & $[64,10,32]_{4}$ & \makecell[l]{$\mathbf{w}=\bigl(1,1,\omega,\omega,0,\omega,\omega,1,\omega,1,1\bigr)^{T}$} & 1 & $[[64,44,6]]_{2}$\\
7 &$[[64,44,6]]_{2}$ & $[62,52,6]_{4}$ & $[64,10,32]_{4}$ & \makecell[l]{$\mathbf{w}^{(1)}=\bigl(0,\omega^{2},0,\omega,\omega,0, \omega^{2},\omega^{2},\omega^{2},1,1 \bigr)^{T}$ \\ $\mathbf{w}^{(2)}=\bigl(1,1,\omega,\omega,0,\omega,\omega,1,\omega,1,1\bigr)^{T}$} & 2  & $[[64,44,6]]_{2}$\\

8 &$[[65,31,9]]_{2}$ & $[64,47,9]_{4}$ & $[65,17,28]_{4}$ & \makecell[l]{$\mathbf{w}=\bigl(1,1,1,\omega^{2},0,\omega^{2},\omega^{2},1,1, \omega^{2},1,0,0,0,0,0,1,1\bigr)^{T}$} & 1 & $[[65,31,9]]_{2}$\\
9 &$[[65,37,8]]_{2}$ & $[64,50,8]_{4}$ & $[65,14,32]_{4}$ & \makecell[l]{$\mathbf{w}=\bigl(1,0,0,\omega,\omega^{2},\omega^{2},\omega^{2},\omega,0,\omega^2,0,\omega^2,0,\omega,\omega\bigr)^{T}$} & 1 & $[[65,37,8]]_{2}$\\
10 &$[[65,49,5]]_{2}$ & $[64,56,5]_{4}$ & $[65,8,44]_{4}$ & \makecell[l]{$\mathbf{w}=\bigl(1,1,\omega^{2},1,0,\omega,1,1,1\bigr)^{T}$} & 1 & $[[65,49,5]]_{2}$\\
11 &$[[86,56,8]]_{2}$ & $[85,70,8]_{4}$ & $[86,15,44]_{4}$ & \makecell[l]{$\mathbf{w}=\bigl(1,\omega^{2},\omega^{2},0,\omega,\omega^{2},0,\omega,0,\omega^{2},\omega,0,\omega^{2},\omega^{2},1,1\bigr)^{T}$} & 1 & $[[86,56,8]]_{2}$\\
12 &$[[86,64,6]]_{2}$ & $[85,74,6]_{4}$ & $[86,11,48]_{4}$ & \makecell[l]{$\mathbf{w}=\bigl(1,0,0,\omega,\omega,\omega^{2},\omega,\omega,1,\omega^{2},1,1\bigr)^{T}$} & 1 & $[[86,64,6]]_{2}$\\
13 &$[[96,60,8]]_2$ & $[95,77,8]_4$ & $[96,18,16]_4$ & \makecell[l]{$\mathbf{w}=\bigl(1,1,0,0,0,\omega,0,0,\omega^{2},\omega^{2}, 1,\omega^{2},1,\omega^{2},\omega,0,0,\omega,\omega\bigr)^{T}$} & 1 & $[[96,60,8]]_2$\\
14 &$[[96,60,8]]_{2}$ & $[94,76,8]_{4}$ & $[96,18,16]_{4}$ & \makecell[l]{$\mathbf{w}^{(1)}=\bigl(\omega^{2},\omega^{2},0,0,0,1,0,0,\omega,\omega, \omega^{2},\omega,\omega^{2},\omega,1,0,0,1,1\bigr)^{T}$ \\ $\mathbf{w}^{(2)}=\bigl(1,\omega^{2},\omega,0,1,0,1,0,\omega^{2},1, \omega^{2},1,\omega^{2},\omega,1,0,0,1,1\bigr)^{T}$} & 2 & $[[96,60,8]]_{2}$\\
15 &$[[96,60,8]]_2$ & $[93,75,8]_4$ & $[96,18,16]_4$ & \makecell[l]{$\mathbf{w}^{(1)}=\bigl(0,\omega^{2},\omega^{2},\omega^{2},0,0,1,\omega^{2},\omega^{2},\omega,0,0,\omega,1,\omega^{2},0,0,\omega^{2},\omega^2\bigr)^{T}$\\
$\mathbf{w}^{(2)}=\bigl(\omega^{2},\omega^{2},0,0,0,1,0,0,\omega,\omega,\omega^{2},\omega,\omega^{2},\omega,1,0,0,1,1\bigr)^{T} $\\ $\mathbf{w}^{(3)}=\bigl(1,\omega^{2},\omega,0,1,0,1,0,\omega^{2},1,\omega^{2},1,\omega^{2},\omega,1,0,0,1,1\bigr)^{T}$} & 3 & $[[96,60,8]]_2$ \\
16 &$[[113,85,7]]_{2}$ & $[112,98,7]_{4}$ & $[113,14,68]_{4}$ & \makecell[l]{$\mathbf{w}=\bigl(0,\omega^2,1,\omega,\omega,\omega^2,1,\omega^2,\omega,\omega,1,\omega^2,0,1,1\bigr)^{T}$} & 1 & $[[113,85,7]]_{2}$\\
17 &$[[128,96,7]]_{2}$ & $[127,111,7]_{4}$ & $[128,16,48]_{4}$ & \makecell[l]{$\mathbf{w}=\bigl(\omega^{2},\omega^{2},0,1,0,\omega^{2},\omega,0,\omega, 1,1,1,\omega,\omega,0,1,1\bigr)^{T}$} & 1 & $[[128,96,7]]_{2}$\\
18 &$[[128,104,6]]_{2}$ & $[127,115,6]_{4}$ & $[128,12,64]_{4}$ & \makecell[l]{$\mathbf{w}=\bigl(\omega,\omega^{2},0,0,1,1,\omega^{2}, 1,\omega,\omega^{2},\omega,1,1\bigr)^{T}$} & 1 & $[[128,104,6]]_{2}$ \\

19 &$[[113,85,7]]_{2}$ & $[111,97,7]_{4}$ & $[113,14,68]_{4}$ & \makecell[l]{$\mathbf{w}^{(1)}=\bigl(\omega^{2},1,\omega,\omega,\omega^{2},1,\omega^{2},\omega, \omega,1,\omega^{2},0,1,0,0\bigr)^{T}$ \\ $\mathbf{w}^{(2)}=\bigl(0,\omega^{2},1,\omega,\omega,\omega^{2},1,\omega^{2}, \omega,\omega,1,\omega^2,0,1,1\bigr)^{T}$} & 2 & $[[113,85,7]]_{2}$\\
20 &$[[128,96,7]]_{2}$ & $[126,110,7]_{4}$ & $[128,16,48]_{4}$ & \makecell[l]{$\mathbf{w}^{(1)}=\bigl(\omega^{2},\omega^{2},0,1,0,\omega^{2},\omega,0,\omega, 1,1,1,\omega,\omega,0,1,1\bigr)^{T}$ \\ $\mathbf{w}^{(2)}=\bigl(\omega^{2},\omega,\omega^{2},0,0,\omega^{2},1,\omega^{2},\omega^{2},1,\omega^{2},1,\omega^{2},\omega^{2},1,1,1\bigr)^{T}$} & 2 & $[[128,96,7]]_{2}$\\
21 &$[[144,104,8]]_{2}$ & $[142,122,8]_{4}$ & $[144,20,56]_{4}$ & \makecell[l]{$\mathbf{w}^{(1)}=\bigl(0,0,0,\omega,1,\omega^{2},\omega^{2},1,1,0,\omega,1,\omega^{2},\omega^{2},\omega,\omega,\omega^{2},0,\omega^{2},1\bigr)^{T}$ \\ $\mathbf{w}^{(2)}=\bigl(1,0,0,\omega^{2},1,\omega,0,1,1,\omega, 0,1,1,0,1,\omega^{2},\omega^{2},\omega^{2},1,\omega^{2}\bigr)^{T}$} & 2 & $[[144,104,8]]_{2}$\\
22 &$[[186,164,5]]_{2}$ & $[184,173,5]_{4}$ & $[186,11,116]_{4}$ & \makecell[l]{$\mathbf{w}^{(1)}=\bigl(\omega^{2},1,0,1,\omega,\omega^{2}, \omega,0,1,0,0\bigr)^{T}$ \\ $\mathbf{w}^{(2)}=\bigl(0,1,\omega,0,\omega,\omega^{2}, 1,\omega^{2},0,\omega,0\bigr)^{T}$} & 2 & $[[186,164,5]]_{2}$\\
23 &$[[80,50,9]]_{3}$ & $[79,64,9]_{9}$ & $[80,15,42]_{9}$ & \makecell[l]{$\mathbf{w}=\bigl(0, 1, \omega, \omega^{5}, \omega^{5}, 0, \omega^{6}, \omega^{2}, \omega^{2}, 0, 1, \omega, 1, \omega^{2}, \omega^{2}\bigr)^{T}$} & 1& $[[80,50,9]]_{3}$ \\
24 &$[[80,62,6]]_{3}$ & $[79,70,6]_{9}$ & $[80, 9, 48]_{9}$ & \makecell[l]{$\mathbf{w}=\bigl(1, 0, 1, \omega^{5}, 1, \omega, 0, \omega^{5}, \omega^{5}\bigr)^{T}$} & 1& $[[80,62,6]]_{3}$ \\
\hline
\end{tabular}
}
\end{sidewaystable}

\begin{sidewaystable}[htbp]
\caption{Quantum codes obtained by extension from cyclic codes}
\label{tab:table5}
\centering
\tiny
\setlength{\tabcolsep}{3pt}\resizebox{\textheight}{!}{%
\begin{tabular}{|c|c|c|c|p{7.3cm}|c|c|}
\hline
NO. & Quantum code & Cyclic code & Defining set & Extension column vector(s) & $r$ & Previous best known (as of Oct.\ 2026) \cite{Grassl:codetables} \\
\hline

1 &$[[56,34,3]]_{2}$ & $[45,34,3]_{4}$ & $C6*C33*C9*C18*C10$ & \makecell[l]{$\mathbf{w}^{(1)}=\bigl(1, \omega^{2}, 1, 0, 1, 0, 0, 0, \omega^{2}, \omega^{2}, 1\bigr)^{T}$ \\ $\mathbf{w}^{(2)}=\bigl(\omega, \omega^{2}, 1, \omega^{2}, 1, \omega^{2}, \omega, \omega^{2}, \omega, \omega^{2}, 1\bigr)^{T}$ \\ $\mathbf{w}^{(3)}=\bigl(\omega, \omega^{2}, \omega^{2}, \omega^{2}, \omega^{2}, 1, \omega, \omega, \omega^{2}, \omega^{2}, \omega\bigr)^{T}$ \\ $\mathbf{w}^{(4)}=\bigl(0, \omega, \omega, 1, \omega, \omega^{2}, \omega^{2}, 0, 1, \omega, \omega\bigr)^{T}$ \\ $\mathbf{w}^{(5)}=\bigl(0, \omega, \omega, 0, 0, 0, 1, \omega, 1, \omega, 1\bigr)^{T}$ \\ $\mathbf{w}^{(6)}=\bigl(1, \omega^{2}, 1, 1, \omega, 1, 0, \omega, 1, 0, 1\bigr)^{T}$ \\ $\mathbf{w}^{(7)}=\bigl(0, 0, 0, 0, 1, 0, \omega, 1, \omega^{2}, \omega, \omega^{2}\bigr)^{T}$ \\ $\mathbf{w}^{(8)}=\bigl(1, 1, \omega^{2}, \omega, \omega, 0, 0, 1, 0, \omega^{2}, 0\bigr)^{T}$ \\ $\mathbf{w}^{(9)}=\bigl(0, \omega, \omega^{2}, \omega^{2}, \omega, \omega, \omega, \omega^{2}, \omega^{2}, \omega, 1\bigr)^{T}$ \\ $\mathbf{w}^{(10)}=\bigl(0, \omega^{2}, 1, 0, \omega, 1, \omega, \omega, \omega^{2}, 1, \omega\bigr)^{T}$ \\ $\mathbf{w}^{(11)}=\bigl(1, 1, \omega^{2}, 1, \omega^{2}, \omega, \omega^{2}, \omega, \omega, 1, \omega\bigr)^{T}$} & 11& $[[56,34,3]]_{2}$ \\
2 &$[[52,34,4]]_{2}$ & $[45,36,4]_{4}$ & $C5*C33*C10*C0$ & \makecell[l]{$\mathbf{w}^{(1)}=\bigl(\omega^{2}, 1, \omega^{2}, \omega^{2}, 0, 1, \omega, 1, 1\bigr)^{T}$ \\ $\mathbf{w}^{(2)}=\bigl(0, 1, 0, 0, \omega^{2}, 1, 1, \omega^{2}, 0\bigr)^{T}$ \\ $\mathbf{w}^{(3)}=\bigl(\omega, \omega, 0, \omega^{2}, \omega, \omega^{2}, 0, 1, \omega\bigr)^{T}$ \\ $\mathbf{w}^{(4)}=\bigl(\omega, 1, \omega^{2}, 0, \omega, \omega^{2}, 1, 0, \omega^{2}\bigr)^{T}$ \\ $\mathbf{w}^{(5)}=\bigl(0, \omega, \omega^{2}, 0, 0, \omega, \omega, 1, 1\bigr)^{T}$ \\ $\mathbf{w}^{(6)}=\bigl(\omega^{2}, \omega^{2}, \omega, \omega, 0, 0, \omega, \omega^{2}, 0\bigr)^{T}$ \\ $\mathbf{w}^{(7)}=\bigl(1, 1, 1, \omega^{2}, \omega^{2}, \omega^{2}, 1, 1, 1\bigr)^{T}$} & 7& $[[52,34,4]]_{2}$ \\
3 &$[[52,18,9]]_{2}$ & $[51,34,9]_{4}$ & $C1*C2*C3*C5*C17$ & \makecell[l]{$\mathbf{w}=\bigl(1, 0, \omega^{2}, 1, 1, 1, \omega^{2}, \omega^{2}, 0, \omega, \omega^{2}, 0, 0, \omega^{2}, 1, \omega, \omega^{2}\bigr)^{T}$} & 1& $[[52,18,9]]_{2}$ \\
4 &$[[64,38,7]]_{2}$ & $[63,50,7]_{4}$ & $C1*C2*C3*C5*C42$ & \makecell[l]{$\mathbf{w}=\bigl(1, \omega, \omega^{2}, \omega^{2}, 0, \omega, 0, \omega^{2}, 1, 1, \omega, 1, \omega\bigr)^{T}$} & 1& $[[64,38,7]]_{2}$ \\
5 &$[[64,38,7]]_{2}$ & $[63,50,7]_{4}$ & $C2*C13*C15*C21*C26$ & \makecell[l]{$\mathbf{w}=\bigl(1, 1, 1, 1, 0, \omega, 0, \omega^{2}, 0, \omega^{2}, 1, 0, \omega\bigr)^{T}$} & 1& $[[64,38,7]]_{2}$ \\
6 &$[[66,42,6]]_{2}$ & $[63,51,6]_{4}$ & $C1*C3*C5*C7$ & \makecell[l]{$\mathbf{w}^{(1)}=\bigl(1, 0, 1, \omega^{2}, \omega^{2}, 0, 0, \omega, \omega, 1, 0, 0\bigr)^{T}$ \\ $\mathbf{w}^{(2)}=\bigl(0, 1, 0, 1, \omega^{2}, \omega^{2}, 0, 0, \omega, \omega, 1, 0\bigr)^{T}$ \\ $\mathbf{w}^{(3)}=\bigl(0, 0, 1, 0, 1, \omega^{2}, \omega^{2}, 0, 0, \omega, \omega, 1\bigr)^{T}$} & 3& $[[66,42,6]]_{2}$ \\
7 &$[[66,28,10]]_{2}$ & $[65,46,10]_{4}$ & $C0*C1*C3*C5$ & \makecell[l]{$\mathbf{w}=\bigl(1, 0, \omega^{2}, 1, \omega^{2}, 0, \omega, \omega^{2}, \omega, \omega^{2}, \omega, \omega^{2}, \omega, 0, \omega^{2}, 1, \omega^{2}, 0, 1\bigr)^{T}$} & 1& $[[66,28,10]]_{2}$ \\
8 &$[[106,68,5]]_{2}$ & $[105,86,5]_{4}$ & $C1*C2*C3*C35$ & \makecell[l]{$\mathbf{w}=\bigl(1, 0, \omega^{2}, \omega^{2}, 0, \omega, \omega^{2}, \omega, \omega, \omega^{2}, \omega, 0, \omega, \omega^{2}, 0, 0, 0, \omega, 1\bigr)^{T}$} & 1& $[[106,68,5]]_{2}$ \\
9 &$[[106,70,6]]_{2}$ & $[105,87,6]_{4}$ & $C1*C3*C5*C7*C35$ & \makecell[l]{$\mathbf{w}=\bigl(1, \omega, \omega, \omega^{2}, 0, 1, \omega, \omega, \omega, \omega, 1, \omega, \omega, 1, \omega, \omega, \omega^{2}, \omega^{2}\bigr)^{T}$} & 1& $[[106,70,6]]_{2}$ \\

10 &$[[49,21,6]]_{2}$ & $[45,31,6]_{4}$ & $C1*C3*C5*C6*C15$ & \makecell[l]{$\mathbf{w}^{(1)}=\bigl(\omega, \omega, \omega^{2}, \omega^{2}, \omega, \omega, \omega^{2}, \omega, \omega, \omega^{2}, 1, 0, \omega, \omega^{2}\bigr)^{T}$ \\ $\mathbf{w}^{(2)}=\bigl(\omega^{2}, \omega^{2}, \omega^{2}, 1, 1, \omega, \omega^{2}, 0, \omega, 0, \omega, 0, \omega, 0\bigr)^{T}$ \\ $\mathbf{w}^{(3)}=\bigl(1, 1, 0, 1, 0, 0, \omega, \omega^{2}, 0, 1, \omega^{2}, 0, 0, 1\bigr)^{T}$ \\ $\mathbf{w}^{(4)}=\bigl(1, \omega^{2}, 1, 0, \omega^{2}, 1, \omega, 1, \omega, 0, \omega^{2}, 1, \omega, 0\bigr)^{T}$} & 4& $[[49,21,6]]_{2}$ \\
11 &$[[49,21,6]]_{2}$ & $[45,31,6]_{4}$ & $C1*C21*C5*C18*C15$ & \makecell[l]{$\mathbf{w}^{(1)}=\bigl(1, \omega^{2}, 0, 1, 0, 0, \omega, \omega, 0, 1, \omega, 0, 0, \omega^{2}\bigr)^{T}$ \\ $\mathbf{w}^{(2)}=\bigl(1, \omega^{2}, 1, \omega^{2}, \omega^{2}, \omega^{2}, 0, \omega^{2}, \omega^{2}, \omega, \omega, 0, \omega, 1\bigr)^{T}$ \\ $\mathbf{w}^{(3)}=\bigl(\omega, \omega^{2}, \omega^{2}, 0, \omega^{2}, \omega^{2}, \omega^{2}, 1, 1, 0, \omega^{2}, \omega^{2}, \omega^{2}, 0\bigr)^{T}$ \\ $\mathbf{w}^{(4)}=\bigl(0, 0, 1, 1, 1, \omega^{2}, 0, \omega, \omega^{2}, 1, 0, 0, \omega, \omega^{2}\bigr)^{T}$} & 4& $[[49,21,6]]_{2}$ \\
12 &$[[49,21,6]]_{2}$ & $[45,31,6]_{4}$ & $C7*C3*C5*C6*C15$ & \makecell[l]{$\mathbf{w}^{(1)}=\bigl(\omega, \omega, \omega^{2}, 0, 1, \omega^{2}, 0, \omega^{2}, \omega, \omega, \omega^{2}, 0, \omega, \omega^{2}\bigr)^{T}$ \\ $\mathbf{w}^{(2)}=\bigl(\omega^{2}, \omega^{2}, \omega^{2}, 0, 0, \omega^{2}, \omega, 1, \omega, 1, \omega, 0, \omega, 0\bigr)^{T}$ \\ $\mathbf{w}^{(3)}=\bigl(1, 1, 0, \omega^{2}, \omega, 0, 1, \omega, 0, 1, 0, 0, 0, 1\bigr)^{T}$ \\ $\mathbf{w}^{(4)}=\bigl(1, \omega^{2}, 1, \omega, \omega, \omega^{2}, \omega^{2}, \omega^{2}, 1, 1, \omega^{2}, 1, \omega, 0\bigr)^{T}$} & 4& $[[49,21,6]]_{2}$ \\
13 &$[[49,21,6]]_{2}$ & $[45,31,6]_{4}$ & $C7*C21*C5*C18*C15$ & \makecell[l]{$\mathbf{w}^{(1)}=\bigl(\omega^{2}, \omega, 0, \omega, \omega^{2}, 0, \omega^{2}, \omega^{2}, 0, \omega^{2}, 0, 0, 0, \omega\bigr)^{T}$ \\ $\mathbf{w}^{(2)}=\bigl(1, \omega^{2}, 1, 1, \omega, 1, 0, 1, \omega^{2}, \omega^{2}, 1, 0, \omega, 1\bigr)^{T}$ \\ $\mathbf{w}^{(3)}=\bigl(\omega, \omega^{2}, \omega^{2}, \omega^{2}, \omega, \omega, 1, \omega^{2}, \omega^{2}, \omega, \omega^{2}, \omega^{2}, \omega^{2}, 0\bigr)^{T}$ \\ $\mathbf{w}^{(4)}=\bigl(0, 0, 1, 1, 1, 1, \omega, 0, \omega^{2}, 0, \omega, 0, \omega, \omega^{2}\bigr)^{T}$} & 4& $[[49,21,6]]_{2}$ \\
14 &$[[46,20,4]]_{2}$ & $[45,32,4]_{4}$ & $C7*C2*C0$ & \makecell[l]{$\mathbf{w}=\bigl(1, 0, 0, \omega^{2}, 0, 0, \omega^{2}, 0, 0, \omega^{2}, 0, 0, 1\bigr)^{T}$} & 1& $[[46,20,4]]_{2}$ \\
15 &$[[50,24,4]]_{2}$ & $[45,32,4]_{4}$ & $C7*C3*C5*C15*C30$ & \makecell[l]{$\mathbf{w}^{(1)}=\bigl(1, \omega^{2}, \omega^{2}, \omega^{2}, \omega, \omega, 1, \omega^{2}, \omega^{2}, 1, \omega^{2}, \omega^{2}, 0\bigr)^{T}$ \\ $\mathbf{w}^{(2)}=\bigl(\omega, \omega^{2}, \omega^{2}, \omega, \omega^{2}, 1, 1, \omega, 0, 0, 0, 1, 0\bigr)^{T}$ \\ $\mathbf{w}^{(3)}=\bigl(1, \omega, \omega^{2}, \omega^{2}, 0, 0, 0, \omega^{2}, 1, 0, 1, \omega, \omega^{2}\bigr)^{T}$ \\ $\mathbf{w}^{(4)}=\bigl(\omega^{2}, \omega^{2}, \omega^{2}, \omega, 0, 1, \omega^{2}, 1, 0, \omega^{2}, \omega, 1, 0\bigr)^{T}$ \\ $\mathbf{w}^{(5)}=\bigl(0, 0, \omega^{2}, 1, \omega^{2}, \omega^{2}, \omega^{2}, \omega^{2}, \omega, 1, \omega, 0, 1\bigr)^{T}$} & 5& $[[50,24,4]]_{2}$ \\
16 &$[[48,26,4]]_{2}$ & $[45,34,4]_{4}$ & $C1*C3*C5$ & \makecell[l]{$\mathbf{w}^{(1)}=\bigl(1, \omega^{2}, 0, 1, 0, 0, \omega, \omega^{2}, 0, 1, 1\bigr)^{T}$ \\ $\mathbf{w}^{(2)}=\bigl(0, 1, \omega, 0, 0, 1, 0, 1, 1, 0, \omega\bigr)^{T}$ \\ $\mathbf{w}^{(3)}=\bigl(\omega^{2}, 1, 0, 0, \omega^{2}, 0, \omega^{2}, \omega^{2}, 0, 1, 0\bigr)^{T}$} & 3& $[[48,26,4]]_{2}$ \\
17 &$[[54,32,3]]_{2}$ & $[45,34,3]_{4}$ & $C6*C9*C18*C10*C15*C30$ & \makecell[l]{$\mathbf{w}^{(1)}=\bigl(0, 1, 1, \omega^{2}, 1, 0, \omega, \omega, 1, \omega, 0\bigr)^{T}$ \\ $\mathbf{w}^{(2)}=\bigl(1, 1, \omega, \omega^{2}, 0, \omega^{2}, 0, 0, 0, \omega^{2}, 1\bigr)^{T}$ \\ $\mathbf{w}^{(3)}=\bigl(\omega^{2}, \omega, 1, 0, 0, \omega^{2}, 0, \omega^{2}, 1, 0, \omega^{2}\bigr)^{T}$ \\ $\mathbf{w}^{(4)}=\bigl(\omega^{2}, \omega^{2}, 1, \omega, \omega^{2}, \omega^{2}, \omega^{2}, \omega, \omega^{2}, \omega^{2}, 0\bigr)^{T}$ \\ $\mathbf{w}^{(5)}=\bigl(\omega^{2}, 0, \omega, \omega, \omega^{2}, \omega, \omega^{2}, \omega^{2}, \omega, 0, \omega\bigr)^{T}$ \\ $\mathbf{w}^{(6)}=\bigl(\omega, \omega^{2}, 1, \omega^{2}, \omega, \omega^{2}, \omega^{2}, 0, 0, 1, 0\bigr)^{T}$ \\ $\mathbf{w}^{(7)}=\bigl(0, \omega, 1, 0, \omega^{2}, 1, \omega, \omega, \omega^{2}, 1, 1\bigr)^{T}$ \\ $\mathbf{w}^{(8)}=\bigl(1, \omega^{2}, 0, \omega^{2}, \omega, 1, 0, 0, 0, 0, 1\bigr)^{T}$ \\ $\mathbf{w}^{(9)}=\bigl(1, \omega, \omega^{2}, \omega^{2}, \omega, 0, \omega, \omega^{2}, 1, 0, 0\bigr)^{T}$} & 9& $[[54,32,3]]_{2}$ \\

\hline
\end{tabular}
}
\end{sidewaystable}

\begin{sidewaystable}[htbp]
\caption{Quantum codes obtained by extension from cyclic codes (continued)}
\label{tab:table6}
\centering
\tiny
\setlength{\tabcolsep}{3pt}\resizebox{\textheight}{!}{%
\begin{tabular}{|c|c|c|c|p{7.3cm}|c|c|}
\hline
NO. & Quantum code & Cyclic code & Defining set & Extension column vector(s) & $r$ & Previous best known (as of Oct.\ 2026) \cite{Grassl:codetables} \\
\hline

18 &$[[106,72,5]]_{2}$ & $[105,88,5]_{4}$ & $C2*C3*C7*C14*C35$ & \makecell[l]{$\mathbf{w}=\bigl(1, 0, 0, \omega, \omega^{2}, 0, \omega, 0, \omega^{2}, 1, \omega^{2}, 1, 1, \omega, \omega^{2}, \omega, 1\bigr)^{T}$} & 1& $[[106,72,5]]_{2}$ \\
19 &$[[106,76,4]]_{2}$ & $[105,90,4]_{4}$ & $C1*C7*C10*C15*C35$ & \makecell[l]{$\mathbf{w}=\bigl(\omega, 1, \omega^{2}, \omega^{2}, \omega, \omega^{2}, \omega, \omega, \omega, 1, \omega, 0, \omega^{2}, \omega, 1\bigr)^{T}$} & 1& $[[106,76,4]]_{2}$ \\
20 &$[[106,76,5]]_{2}$ & $[105,90,5]_{4}$ & $C3*C7*C11*C35$ & \makecell[l]{$\mathbf{w}=\bigl(1, 1, \omega^{2}, \omega, \omega^{2}, 0, \omega^{2}, \omega, 1, 1, 0, 0, 0, \omega, \omega^{2}\bigr)^{T}$} & 1& $[[106,76,5]]_{2}$ \\
21 &$[[110,84,4]]_{2}$ & $[105,92,4]_{4}$ & $C1*C7*C21*C35*C42$ & \makecell[l]{$\mathbf{w}^{(1)}=\bigl(\omega, \omega^{2}, 1, \omega, 0, 1, 0, \omega, \omega^{2}, 1, 1, 1, \omega^{2}\bigr)^{T}$ \\ $\mathbf{w}^{(2)}=\bigl(1, \omega^{2}, \omega^{2}, 1, 1, \omega^{2}, 1, 1, 0, \omega^{2}, 0, 1, 0\bigr)^{T}$ \\ $\mathbf{w}^{(3)}=\bigl(1, 0, \omega^{2}, 0, \omega, \omega^{2}, 0, 0, \omega, 1, 0, \omega, 0\bigr)^{T}$ \\ $\mathbf{w}^{(4)}=\bigl(0, 0, 0, \omega^{2}, \omega^{2}, 1, \omega^{2}, 1, \omega, \omega^{2}, \omega^{2}, \omega^{2}, 1\bigr)^{T}$ \\ $\mathbf{w}^{(5)}=\bigl(1, 0, 1, \omega^{2}, \omega^{2}, \omega, 0, 1, \omega^{2}, 0, 1, \omega, 0\bigr)^{T}$} & 5& $[[110,84,4]]_{2}$ \\
22 &$[[106,80,4]]_{2}$ & $[105,92,4]_{4}$ & $C1*C9*C35$ & \makecell[l]{$\mathbf{w}=\bigl(1, \omega, 1, \omega^{2}, 1, \omega^{2}, 1, 0, \omega^{2}, 0, 1, 0, 1\bigr)^{T}$} & 1& $[[106,80,4]]_{2}$ \\
23 &$[[110,84,4]]_{2}$ & $[105,92,4]_{4}$ & $C7*C13*C21*C35*C42$ & \makecell[l]{$\mathbf{w}^{(1)}=\bigl(\omega, 1, \omega, 0, 1, \omega^{2}, 1, 1, 1, 0, \omega^{2}, \omega, \omega^{2}\bigr)^{T}$ \\ $\mathbf{w}^{(2)}=\bigl(1, \omega, 0, 0, 0, \omega, 0, \omega, 1, \omega, 1, 1, 0\bigr)^{T}$ \\ $\mathbf{w}^{(3)}=\bigl(0, 0, \omega, \omega, \omega, 1, 1, \omega^{2}, \omega, 1, 0, \omega, 1\bigr)^{T}$ \\ $\mathbf{w}^{(4)}=\bigl(1, 1, \omega, 1, \omega, \omega, 0, 1, 0, \omega, 0, 0, 1\bigr)^{T}$ \\ $\mathbf{w}^{(5)}=\bigl(1, 1, 0, 0, 0, 1, \omega, \omega, 1, \omega, 1, 0, 1\bigr)^{T}$} & 5& $[[110,84,4]]_{2}$ \\
24 &$[[112,86,3]]_{2}$ & $[105,92,3]_{4}$ & $C10*C13*C25*C35$ & \makecell[l]{$\mathbf{w}^{(1)}=\bigl(\omega^{2}, \omega^{2}, \omega^{2}, \omega, \omega^{2}, \omega^{2}, \omega^{2}, \omega, 0, 0, \omega, 0, \omega\bigr)^{T}$ \\ $\mathbf{w}^{(2)}=\bigl(0, \omega^{2}, 0, \omega, 0, 0, 0, \omega^{2}, \omega^{2}, 0, 1, \omega, \omega\bigr)^{T}$ \\ $\mathbf{w}^{(3)}=\bigl(0, \omega, \omega^{2}, \omega, 1, 1, 1, \omega^{2}, 0, 1, \omega, 1, \omega^{2}\bigr)^{T}$ \\ $\mathbf{w}^{(4)}=\bigl(1, 0, 0, \omega^{2}, 1, 1, 0, 0, \omega, \omega^{2}, \omega, \omega^{2}, \omega\bigr)^{T}$ \\ $\mathbf{w}^{(5)}=\bigl(1, 1, 1, 1, \omega^{2}, \omega, 1, 1, \omega, 1, \omega, \omega, 1\bigr)^{T}$ \\ $\mathbf{w}^{(6)}=\bigl(\omega^{2}, 0, \omega^{2}, \omega^{2}, 0, \omega, \omega, \omega, \omega, 1, \omega^{2}, 0, \omega\bigr)^{T}$ \\ $\mathbf{w}^{(7)}=\bigl(1, \omega, \omega^{2}, \omega^{2}, 0, 0, \omega^{2}, 1, \omega^{2}, \omega^{2}, \omega, \omega, \omega\bigr)^{T}$} & 7& $[[112,86,3]]_{2}$ \\
25 &$[[106,82,4]]_{2}$ & $[105,93,4]_{4}$ & $C5*C10*C14*C15*C35$ & \makecell[l]{$\mathbf{w}=\bigl(1, \omega^{2}, 1, 0, \omega^{2}, \omega^{2}, 0, \omega, 1, \omega^{2}, 0, \omega\bigr)^{T}$} & 1& $[[106,82,4]]_{2}$ \\
26 &$[[110,86,4]]_{2}$ & $[105,93,4]_{4}$ & $C7*C10*C14*C21*C35*C42$ & \makecell[l]{$\mathbf{w}^{(1)}=\bigl(0, 0, \omega, \omega^{2}, \omega^{2}, \omega^{2}, 0, 1, 1, 1, \omega^{2}, \omega^{2}\bigr)^{T}$ \\ $\mathbf{w}^{(2)}=\bigl(1, \omega^{2}, 1, 0, 1, \omega^{2}, 1, \omega, \omega, \omega, 0, 1\bigr)^{T}$ \\ $\mathbf{w}^{(3)}=\bigl(\omega, 0, \omega^{2}, \omega^{2}, \omega^{2}, \omega, 0, 1, \omega, \omega^{2}, \omega, 1\bigr)^{T}$ \\ $\mathbf{w}^{(4)}=\bigl(0, 1, \omega, \omega, \omega, 0, \omega^{2}, \omega^{2}, \omega^{2}, \omega, \omega, 0\bigr)^{T}$ \\ $\mathbf{w}^{(5)}=\bigl(1, \omega, \omega, \omega, 0, \omega^{2}, \omega^{2}, \omega^{2}, \omega, \omega, 0, 0\bigr)^{T}$} & 5& $[[110,86,4]]_{2}$ \\
27 &$[[27,1,8]]_{3}$ & $[26,13,8]_{9}$ & $C0 * C1 * C2 * C4 * C5$ & \makecell[l]{$\mathbf{w}=\bigl(1, 2, 2, 1, 0, 0, 2, 0, 0, 0, 0, 1, 1\bigr)^{T}$} & 1& $[[27,1,8]]_{3}$ \\

28 &$[[42,16,8]]_{3}$ & $[41,28,8]_{9}$ & $C0 * C1 * C2 * C4$ & \makecell[l]{$\mathbf{w}=\bigl(1, \omega^{7}, 0, 0, 0, \omega^{5}, \omega, \omega^{5}, 0, 0, 0, \omega^{7}, 1\bigr)^{T}$} & 1& $[[42,16,8]]_{3}$ \\
29 &$[[42,24,6]]_{3}$ & $[41,32,6]_{9}$ & $C0 * C1 * C2$ & \makecell[l]{$\mathbf{w}=\bigl(1, \omega^{5}, 1, \omega, \omega^{3}, \omega, 1, \omega^{5}, 1\bigr)^{T}$} & 1& $[[42,24,6]]_{3}$ \\
30 &$[[74,48,6]]_{3}$ & $[73,60,6]_{9}$ & $C0 * C13 * C21$ & \makecell[l]{$\mathbf{w}=\bigl(1, \omega, 2, \omega^{2}, \omega^{7}, \omega^{7}, \omega^{5}, \omega^{7}, \omega^{7}, \omega^{2}, 2, \omega, 1\bigr)^{T}$} & 1& $[[74,48,6]]_{3}$ \\

31 &$[[92,78,5]]_{3}$ & $[91,84,5]_{9}$ & $C0 * C1 * C2$ & \makecell[l]{$\mathbf{w}=\bigl(1, \omega^{5}, 0, 2, \omega, \omega^{6}, 1\bigr)^{T}$} & 1& $[[92,78,4]]_{3}$ \\
32 &$[[122,100,5]]_{3}$ & $[121,110,5]_{9}$ & $BCH$ & \makecell[l]{$\mathbf{w}=\bigl(1, 1, 2, 1, 2, 2, 2, 1, 1, 2, 1\bigr)^{T}$} & 1& $--$ \\

33 &$[[45,25,6]]_{4}$ & $[39,29,6]_{16}$ & $C6*C8*C12*C13$ & \makecell[l]{$\mathbf{w}^{(1)}=\bigl(1, \omega^{10}, 1, \omega^{13}, 0, \omega^{13},\omega^{5},0,\omega,\omega^{12} \bigr)^{T}$\\$\mathbf{w}^{(2)}=\bigl(0, \omega^{3}, 1, \omega^{5},  \omega^{13},\omega^{14},\omega^{11},0,\omega^{2},\omega^{10} \bigr)^{T}$\\$\mathbf{w}^{(3)}=\bigl(0,  1,0, \omega^{6},  \omega^{13},\omega^{4},\omega^{5},\omega^{8},\omega^{4},\omega^{14} \bigr)^{T}$  \\$\mathbf{w}^{(4)}=\bigl(0,  0,0, 1,\omega^{7},  \omega^{9},\omega^{12},\omega^{12},\omega^{13},\omega^{12},\omega^{14} \bigr)^{T}$
\\$\mathbf{w}^{(5)}=\bigl(0,  0,0, 0,1,  \omega,\omega^{11},\omega^{11},\omega^{13},1 \bigr)^{T}$\\$\mathbf{w}^{(6)}=\bigl(0,0,0, 0,0,  1,1,\omega^{6},\omega^{2},\omega^{8} \bigr)^{T}$
} 
& 6& $[[45,25,6]]_{4}$ \\
34 &$[[42,18,5]]_{4}$ & $[39,27,5]_{16}$ & $C8*C14*C19*C21$   & \makecell[l]{$\mathbf{w}^{(1)}=\bigl(1, 1, \omega^{11}, \omega^{14}, \omega^{4}, 0, \omega^{2},\omega,\omega^{14},\omega^{7}, \omega^{9},\omega^{8}\bigr)^{T}$\\$\mathbf{w}^{(2)}=\bigl(0, \omega^{12}, 1, \omega^{14},  \omega^{6},\omega^{9},\omega^{14},\omega,\omega^{8},\omega,\omega^{9},\omega \bigr)^{T}$\\$\mathbf{w}^{(3)}=\bigl(0,  1,0, \omega^{7},  \omega^{7},\omega,\omega^{14},\omega^{2},\omega^{7},1,\omega^{8},\omega^{10} \bigr)^{T}$} & 3& $[[42,18,5]]_{4}$ \\

\hline
\end{tabular}
}
\end{sidewaystable}

\begin{sidewaystable}[htbp]
\caption{Computed $n=51$ cyclic-code instances via extension}
\label{tab:distance_preservation}
\centering
\tiny
\setlength{\tabcolsep}{3pt}\begin{tabular}{|c|c|c|c|p{5.6cm}|c|}
\hline
No. & Original cyclic code & $r$ & Quantum code & Cosets used & Distance preserved \\  \hline
1 & $[51,46,3]_4$ & 1 & $[[52,42,3]]_2$ & $C_1,C_{17}$ & yes \\
2 &$[51,38,7]_4$ & 1 & $[[52,26,7]]_2$ & $C_1,C_3,C_5,C_{17}$ & yes \\
3 &$[51,34,9]_4$ & 1 & $[[52,18,9]]_2$ & $C_1,C_2,C_3,C_5,C_{17}$ & yes \\
4 &$[51,44,4]_4$ & 3 & $[[54,40,4]]_2$ & $C_3,C_{17},C_{34},C_0$ & yes \\
5 &$[51,40,6]_4$ & 3 & $[[54,32,6]]_2$ & $C_1,C_3,C_{17},C_{34},C_0$ & yes \\
6 &$[51,36,8]_4$ & 3 & $[[54,24,8]]_2$ & $C_1,C_3,C_5,C_{17},C_{34},C_0$ & yes \\
\hline
\end{tabular}
\end{sidewaystable}

\section{Counting and the fine structure of the $\tau$-fixed cosets}\label{app:structure}

\subsection{The $\tau$-fixed cosets: singletons, the closed form for $f_\tau$, and a sharp existence theorem for $r=1$}\label{sec:singletons}

\begin{lemma}\label{lem:singletons}
A residue $\nu\in\mathbb{Z}_n$ is fixed by $\tau$ if and only if $n\mid\nu(q+1)$, i.e. if and only if
$\nu$ is a multiple of $n/g$ where $g:=\gcd(n,q+1)$. Every such $\nu$ generates a \emph{singleton}
$q^2$-cyclotomic coset. Consequently the number of $\tau$-fixed singleton cosets equals $g$, they are
\[
\{0,\ n/g,\ 2n/g,\ \dots,\ (g-1)n/g\},
\]
and in particular $\{0\}$ is always one of them. Moreover a coset that is $\tau$-fixed but not a
singleton exists precisely when $-qC=C$ holds for a larger coset, i.e. when $C$ meets its own
$\tau$-image in a cycle of length $>1$.
\end{lemma}

\begin{proof}
$\tau(\nu)=\nu\iff n\mid\nu(q+1)$. If $n\mid\nu(q+1)$ then $n\mid\nu(q^2-1)$ as well, since
$q+1\mid q^2-1$, so the $q^2$-cyclotomic coset of $\nu$ is $\{\nu\}$. The set
$\{\nu:n\mid\nu(q+1)\}$ is the subgroup of $\mathbb{Z}_n$ of order $\gcd(n,q+1)=g$, i.e. the multiples
of $n/g$. The last assertion is the definition of a $\tau$-fixed coset, $-qC=C$ (see Theorem~\ref{thm:r-formula}).
\end{proof}

\begin{theorem}\label{thm:sharp-family}
Let $n\ge3$ and let $q$ be a prime power with $\gcd(n,q)=1$. Let $T$ be the union of $\{0\}$ with
\emph{one} coset chosen arbitrarily from each two-element $\tau$-orbit of the $q^2$-cyclotomic cosets
of $\mathbb{Z}_n$, and let $\mathcal{C}=\mathcal{C}(T)$. Then
\begin{equation}
\label{eq:sharp-params}
T\cap\tau(T)=\{0\},\qquad r=1,\qquad
m=1+\tfrac12(n-f_\tau),\qquad
\mathcal{Q}=[\![\,n+1,\ f_\tau-1,\ d'\,]\!]_q ,
\end{equation}
where $f_\tau$ is the total size of the $\tau$-fixed cosets. In particular
\emph{every} $n$ coprime to $q$ is admissible, and the hypothesis $g=\gcd(n,q+1)>1$ of
the infinite-family theorem is unnecessary.
\end{theorem}

\begin{proof}
By the singleton lemma the coset of $0$ is $\{0\}$ and is $\tau$-fixed, so
$0\in T\cap\tau(T)$. Every chosen coset $C$ of a two-element orbit has $\tau(C)\not\subseteq T$, so it
contributes nothing to $T\cap\tau(T)$; every unchosen coset is outside $T$. Hence
$T\cap\tau(T)=\{0\}$ and $r=1$ by Theorem~\ref{thm:r-formula}. The count
$m=1+\tfrac12(n-f_\tau)$ is that of \eqref{eq:family-T}, and the parameters follow from
parts (iii) and (iv) of Theorem~\ref{thm:family}, whose proofs use only that
$\{i_0\}$ is a $\tau$-fixed singleton
and $N(\mu)=-(n\bmod p)$ is solvable; both hold for $i_0=0$. No divisibility of $n$ by a prime factor
of $q+1$ is used, so all $n$ coprime to $q$ are admissible.
\end{proof}

\begin{remark}\label{rem:sharp-vs-family}
The existence theorem above strictly contains the infinite-family theorem: the latter requires
$g>1$ and concludes that the multiples of the prime divisors $\ell\neq p$ of $q+1$ are admissible,
whereas taking $i_0=0$ costs nothing and makes every admissible length available. The family remains
of low rate, $k_Q/n\le f_\tau/(n+1)$, as already noted in Section~\ref{sec:family}, so the
gain is one of \emph{scope}, not of rate.
\end{remark}

\begin{proposition}\label{prop:ftau}
Let $G_j$ be the subgroup of $\mathbb{Z}_n$ of order $\gcd(n,\,q^{2j-1}+1)$ for $j\ge1$, and
$G_0$ the subgroup of order $\gcd(n,q+1)$. Then
\begin{equation}
\label{eq:ftau-closed}
f_\tau=\Bigl|\bigcup_{j\ge0}G_j\Bigr|
=\#\Bigl\{\nu\in\mathbb{Z}_n:\ \nu\bigl(q^{2j-1}+1\bigr)\equiv0\ (\mathrm{mod}\ n)\ \text{for some }j\ge0\Bigr\},
\end{equation}
where for $j=0$ the factor $q^{-1}+1$ is to be read as $q+1$ after multiplication by $q$.
\end{proposition}

\begin{proof}
$C_\nu$ is $\tau$-fixed iff $-q\nu\equiv q^{2j}\nu$ for some $j\ge0$, i.e.
$\nu(q^{2j-1}+1)\equiv0$; the set of such $\nu$ is the union of the $G_j$, and $G_0$ corresponds to
$j=0$ as stated. Summing the sizes of the $\tau$-fixed cosets gives the closed form for $f_\tau$.
\end{proof}

\subsection{Counting and distribution: the hull codimension enumerator}\label{sec:hullcodim-enumerator}

The hull codimension decomposes additively over the $\tau$-orbits, and summing that
decomposition over \emph{all} defining sets therefore gives an exact generating function. Let $a=\#\{F_i\}$ be the
number of $\tau$-fixed cosets, with sizes $s_1,\dots,s_a$, and $b=\#\{\{C_j,\tau(C_j)\}\}$ the number
of two-element $\tau$-orbits, with $t_j=|C_j|$.

\begin{theorem}\label{thm:enumerator}
The defining sets $T$ (unions of $q^{2}$-cyclotomic cosets of $\mathbb{Z}_n$) satisfy
\begin{equation}
\label{eq:enumerator}
\sum_{T}x^{|T|}y^{r(T)}
=\prod_{i=1}^{a}\bigl(1+x^{s_i}y^{s_i}\bigr)\;\cdot\;\prod_{j=1}^{b}\bigl(1+2x^{t_j}+x^{2t_j}y^{2t_j}\bigr).
\end{equation}
In particular $\sum_{T}1=2^{a+2b}$, and the coefficient of $x^{m}y^{\rho}$ is the number of defining
sets with $|T|=m$ and $r(T)=\rho$.
\end{theorem}

\begin{proof}
A $\tau$-fixed coset $F_i$ contributes $s_i$ to both $|T|$ and $r(T)$ when it is taken, and $0$
otherwise: the factor $1+x^{s_i}y^{s_i}$. A two-element $\tau$-orbit has four states, of which three
are admissible: neither coset taken ($0$ to both, weight $1$), exactly one taken (two choices, $t_j$
to $|T|$ and $0$ to $r$, weight $2x^{t_j}$), and both taken ($2t_j$ to both, weight $x^{2t_j}y^{2t_j}$).
Distinct cosets and orbits are independent, so the factors multiply; setting $x=y=1$ gives
$2^{a}4^{b}=2^{a+2b}$.
\end{proof}

\begin{corollary}\label{cor:counts}
With $a,b$ as above and $g=\gcd(n,q+1)$,
\begin{equation}
\label{eq:count0}
\#\{T:\ r(T)=0\}=3^{b},
\qquad\qquad
\#\{T:\ r(T)=1\}=g\,3^{b},
\end{equation}
so the defining sets with $r=1$ outnumber those with $r=0$ exactly $g$-fold.
\end{corollary}

\begin{proof}
$r(T)=0$ forces $A=B=0$ in the orbit identity: no $\tau$-fixed coset is taken and no orbit is taken
whole, which leaves each of the $b$ orbits in one of the three admissible states, whence $3^{b}$.
For $r(T)=1$, the parity constraint forces $T\cap\tau(T)$ to be a single $\tau$-fixed
singleton; by Lemma~\ref{lem:singletons} there are $g$ of them, one is taken while no other
$\tau$-fixed coset is, and each orbit is again in one of the three admissible states, whence
$g\,3^{b}$.
\end{proof}

\begin{corollary}\label{cor:scarcity}
If each $q^{2}$-cyclotomic coset is included in $T$ independently with probability $\tfrac12$, then
\begin{equation}
\label{eq:scarcity}
\Pr\bigl[r(T)=1\bigr]=\frac{g\,3^{b}}{2^{a}4^{b}}=\frac{g}{2^{a}}\Bigl(\frac34\Bigr)^{b},
\qquad
\Pr\bigl[r(T)=0\bigr]=\frac{3^{b}}{2^{a}4^{b}}=\frac{1}{2^{a}}\Bigl(\frac34\Bigr)^{b}.
\end{equation}
\end{corollary}

\begin{proof}
Under the model each of the $2^{a+2b}$ unions is equiprobable, so the probabilities are the exact
counts divided by $2^{a+2b}=2^{a}4^{b}$.
\end{proof}

\begin{corollary}\label{cor:spectrum}
The set of values attained by $r$ on the defining sets of $\mathbb{Z}_n$ is exactly the set of subset
sums of the multiset $\{s_1,\dots,s_a\}\cup\{2t_1,\dots,2t_b\}$,
\begin{equation}
\label{eq:spectrum}
\bigl\{r(T):T\bigr\}=\Bigl\{\textstyle\sum_{i\in I}s_i+2\sum_{j\in J}t_j:\ I\subseteq\{1,\dots,a\},\
J\subseteq\{1,\dots,b\}\Bigr\};
\end{equation}
its largest element is $\sum_i s_i+2\sum_j t_j=n$, attained by $T=\mathbb{Z}_n$, and it is in general
\emph{not} an interval; for $(n,q)=(65,2)$ the values $2$ and $3$ are attained at no length.
\end{corollary}

\begin{proof}
Projecting the orbit identity onto $r$ and forgetting $|T|$ gives the spectrum identity; the largest
element is the sum of all generators, which is $n$ because the $\tau$-fixed cosets and the
two-element $\tau$-orbits partition $\mathbb{Z}_n$.
\end{proof}

\begin{theorem}\label{thm:mean-variance}
Under the model of the scarcity corollary,
\begin{align}
\label{eq:mean}
\mathrm{E}[r]&=\frac{n+f_\tau}{4},\\
\label{eq:var}
\mathrm{Var}[r]&=\frac14\sum_{i=1}^{a}s_i^{2}+\frac34\sum_{j=1}^{b}t_j^{2},
\end{align}
where $f_\tau=\sum_i s_i$; consequently
$\Pr\bigl[\,|r-\mathrm{E}[r]|\ge\lambda\,\bigr]\le\mathrm{Var}[r]/\lambda^{2}$ for every
$\lambda>0$. Thus a typical defining set has hull codimension $\approx n/4$, while
Corollary~\ref{cor:scarcity} makes the supply of small hull codimensions exponentially small in $b$: for
the record moduli $(51,2),(65,2),(91,3),(121,3)$ the densities $\Pr[r=1]$ are
$6.67\%$, $8.90\%$, $0.22\%$ and $1.58\%$ respectively.
\end{theorem}

\begin{proof}
The contributions of distinct cosets and orbits are independent: a $\tau$-fixed coset contributes
$s_i$ with probability $\tfrac12$ (mean $s_i/2$, variance $s_i^{2}/4$), while a two-element
$\tau$-orbit contributes $2t_j$ with probability $\tfrac14$ (mean $t_j/2$, variance
$\tfrac34t_j^{2}$). Adding the means and using $\sum_j t_j=\tfrac12(n-f_\tau)$ gives
the mean identity; adding the variances gives the variance identity, and the last display is Chebyshev's
inequality.
\end{proof}

\begin{corollary}\label{cor:allfixed}
If every $q^{2}$-cyclotomic coset of $\mathbb{Z}_n$ is $\tau$-fixed, in particular if
$n\mid q+1$, then $b=0$, $f_\tau=n$ and
\begin{equation}
\label{eq:rigid}
r(T)=|T|=m\qquad\text{for every defining set }T .
\end{equation}
Hence the hull codimension is maximal ($r=n-k$) throughout, and $T=\varnothing$ is the only defining set
with $r=0$.
\end{corollary}

\begin{proof}
If every coset is $\tau$-fixed then $A=\sum_{i:F_i\subseteq T}s_i=|T|$ and $B=0$ in the orbit identity,
so $r=|T|=m$; and $n\mid q+1$ implies $-q\nu\equiv\nu$ for all $\nu$, so every coset is $\tau$-fixed.
\end{proof}

\clearpage

\section{Detailed data for Examples 2--7}\label{app:examples}

\textbf{Example 2 $[[59,29,8]]_2$.} 
Consider the following polynomials in $R$:
$$f(x) = x^{14} + \omega x^{13} + \omega x^{11} + \omega^2 x^{10} + x^9 + \omega^2 x^8 + \omega x^7 + \omega^2 x^6 +x^5$$ $$+ \omega^2 x^4 + \omega x^3 + \omega x + 1,$$
$$v_1(x)   = \omega^2 x^{26} + x^{25} + x^{23} + \omega^2 x^{22} + \omega x^{21} + \omega x^{20} + \omega x^{18} + \omega x^{16} + x^{15}+$$
$$ \omega x^{14} + \omega x^{13} + \omega^2 x^{12} + \omega^2 x^{10} + \omega x^8 + \omega^2 x^6 + \omega^2 x^5 + x^3+ \omega^2 x^2 + x + \omega,$$ 

and $$v_2(x) = \omega x^{28} + \omega^2 x^{27} + x^{24} + x^{23} + x^{21} + \omega x^{20} + x^{19} + \omega x^{18} + \omega^2 x^{17} +$$ $$\omega x^{16} 
+ \omega^2 x^{15} + \omega^2 x^{14} + \omega^2 x^{12} + x^{10} + \omega^2 x^8 + \omega^2 x^6 + \omega^2 x^5 + x^4 + \omega x^3 + \omega.$$

 We define a $1$-quasi-cyclic
code $\mathcal{C}=[58,43,8]_4$ of length $2n$ by  $$   \mathcal{C}
    = \big\langle (v_1 f,\; v_2 f_1)\big.\rangle
$$
Its parity-check matrix $H$ has size $15\times 58$, and
$$
A=HH^{\dagger}=\begin{pmatrix}
1 & 1 & \cdots & 1 \\
1 & 1 & \cdots & 1 \\
\vdots & \vdots & \ddots & \vdots \\
1 & 1 & \cdots & 1
\end{pmatrix}_{15\times 15}, \qquad r=\Rank(A)=1.
$$
The rank-reducing vectors are supplied by Lemma~\ref{lem:explicit} (equivalently, by scanning a basis of the syndrome space $\F_{q^2}^{15}$ for the rank-reducing condition); the selected vector has Hamming weight $15$, and its dual distance is certified by the exact criterion \eqref{eq:exactcrit}. It is
\begin{equation}
\mathbf{w}=(1,1,1,1,1,1,1,1,1,1,1,1,1,1,1)^{T}\in\F_4^{15}.
\end{equation}
Appending it to $H$ gives $H_1=[H\mid\mathbf{w}^{(1)}]$ with
\begin{equation}
\Rank(H_1H_1^{\dagger})=1\to 0,\qquad d\bigl(\langle H_1\rangle^{\perp_h}\bigr)=8.
\end{equation}
Thus we obtain a self-orthogonal code $\mathcal{C}_{\mathrm{so}}=[59,15,26]_4$ whose Hermitian dual is $\mathcal{C}_{\mathrm{so}}^{\perp_h}=[59,44,8]_4$ with weight enumerator polynomial $W_{44,59}(z)=$$1+13746z^{8}+229767z^{9}+3451638z^{10}+46197174z^{11}+554418984z^{12}+6008492136z^{13}+\dots+13159947574445658066z^{59}$. 
Hence
\begin{equation}
\mathcal{Q}=[[59,\,2\cdot 43-58+1,\,8]]_2=\mathbf{[[59,29,8]]_2},
\end{equation}
improving the previously known parameters in Grassl's tables.
By Proposition~\ref{propagation_rule}, $[[59,29,8]]_2$ also yields the additional quantum codes $[[60,29,8]]_2$ (by lengthening) and $[[61,29,8]]_2$ (by a second lengthening).

\textbf{Example 3 $[[79,41,9]]_2$.} Let $\mathcal{C}$ be the $(2,2)$-quasi-cyclic $[78,59,9]_4$ code with generator polynomials $f_{1}(x) = x + \omega^{2}$, 
    $f_{2}(x) = x^{18} + \omega x^{17} + \omega x^{16} + x^{14} + x^{13} + \omega^{2}x^{11} + \omega^{2}x^{9} + x^{7} + \omega^{2}x^{5} + 
    \omega x^{4} + x^{2} + \omega^{2} x + 1$, and a fixed $v(x) = x^{38} +\omega x^{37} +\omega  x^{34} + \omega^2 x^{32} + \omega^2 x^{31} + x^{29} + \omega x^{28} + \omega^{2}x^{26} + \omega  x^{25}
    + \omega x^{24} + \omega^{2} x^{23} + \omega x^{22} + x^{21} + \omega^{2} x^{20} + \omega^{2} x^{19} + \omega^{2} x^{18} + 
    \omega^{2} x^{17} + \omega x^{16} + \omega x^{13} + x^{12} + x^{5} + x^{4} + \omega x^{3} + x + 1$.

Thus $\mathcal{C}$ is the 2-quasi-cyclic code of length $2n=78$ generated as an $R$-submodule by $(v f_1,\; f_1)$ and $(f_2,\; v f_2)$, or equivalently by the $2\times 2$ polynomial generator matrix
\[
    G(x) =
    \begin{pmatrix}
        v f_1 & f_1 \\[2mm]
        f_2   & v f_2
    \end{pmatrix}
    \in \mathcal{M}_{2 \times 2}(R).
\]
 Its Hermitian dual $\mathcal{C}^{\perp_h}$ has dimension $19$; taking a parity-check matrix $H$ of $\mathcal{C}$ (equivalently a generator matrix of $\mathcal{C}^{\perp_h}$) we compute
$A=HH^{\dagger}\in\F_4^{19\times 19},\qquad r=\Rank(A)=1$.

The extension vector is selected by the rank-reducing test; the exact criterion also holds with $d=9$ for
$$\mathbf{w}=(1,1,\omega,\omega,0,1,0,\omega^{2},1,0,\omega^{2},1,1,0
,\omega,\omega^{2},1,1,\omega)^{T}.$$
After this single extension the rank of the Hermitian product matrix is reduced from $1$ to $0$, and the extended dual code has minimum distance $9$.
The resulting self-orthogonal code has parameters $\mathcal{C}_{\mathrm{so}}=[79,19,\cdot]_4$ and Hermitian dual $\mathcal{C}_{\mathrm{so}}^{\perp_h}=[79,60,9]_4$ with weight enumerator polynomial $W_{60,79}(z)=$$1+12909z^{9}+309231z^{10}+5804838z^{11}+99095841z^{12}+1530748683z^{13}+\dots+179241796303546028802673698z^{79}$. 
Therefore
$$
\mathcal{Q}=[[79,\,2\cdot 59-78+1,\,9]]_2=\mathbf{[[79,41,9]]_2},$$
improving the previously known parameters in Grassl's tables.
By Proposition~\ref{propagation_rule}, $[[79,41,9]]_2$ also yields the additional quantum code $[[80,41,9]]_2$ (by lengthening).

\textbf{Example 4 $[[92,78,5]]_3$.} Let $\mathcal{C}$ be the BCH code $[91,84,5]_9$ with defining set constructed in the usual way over $\F_9=\F_3[\omega]$. Its parity-check matrix $H$ has size $7\times 91$, and
\begin{equation}
A=HH^{\dagger}\in\F_9^{7\times 7},\qquad r=\Rank(A)=1.
\end{equation}
The extension vector satisfying the rank-reducing condition and the exact criterion for $d=5$ is
\begin{equation}
\mathbf{w}^{(1)}=(1,\omega^5,0,2,\omega,\omega^6,1)^{T}\in\F_9^{7}.
\end{equation}
Appending it to $H$ gives $H_1=[H\mid\mathbf{w}^{(1)}]$ with
\begin{equation}
\Rank(H_1H_1^{\dagger})=1\to 0,\qquad d\bigl(\langle H_1\rangle^{\perp_h}\bigr)=5.
\end{equation}
The extended self-orthogonal code has parameters $\mathcal{C}_{\mathrm{so}}=[92,7,70]_9$, and its Hermitian dual is $\mathcal{C}_{\mathrm{so}}^{\perp_h}=[92,85,5]_9$. Consequently
\begin{equation}
\mathcal{Q}=[[92,\,2\cdot 84-91+1,\,5]]_3=\mathbf{[[92,78,5]]_3},
\end{equation}
improving the previously known parameters in Grassl's tables.
By Proposition~\ref{propagation_rule}, $[[92,78,5]]_3$ also yields the additional quantum code $[[93,78,5]]_3$ (by lengthening).

\textbf{Example 5 $[[80,64,6]]_3$.} Let $\mathcal{C}=[79,71,6]_9$ be the best-known linear code described  in Grassl's tables (as of October 2026). We have
\begin{equation}
A = H H^{\dagger} \in \F_9^{8\times 8}, \qquad r = \operatorname{Rank}(A) = 1,
\end{equation}
where
\begin{equation}
A =
\begin{pmatrix}
2 & 0 & \omega^7 & \omega^7 & \omega^6 & \omega^2 & \omega^5 & 1 \\
0 & 0 & 0 & 0 & 0 & 0 & 0 & 0 \\
\omega^5 & 0 & 1 & 1 & \omega^7 & \omega^3 & \omega^6 & \omega \\
\omega^5 & 0 & 1 & 1 & \omega^7 & \omega^3 & \omega^6 & \omega \\
\omega^2 & 0 & \omega^5 & \omega^5 & 2 & 1 & \omega^3 & \omega^6 \\
\omega^6 & 0 & \omega & \omega & 1 & 2 & \omega^7 & \omega^2 \\
\omega^7 & 0 & \omega^2 & \omega^2 & \omega & \omega^5 & 1 & \omega^3 \\
1 & 0 & \omega^3 & \omega^3 & \omega^2 & \omega^6 & \omega & 2
\end{pmatrix}.
\end{equation}
with $H$ of size $8\times 79$. By Lemma~\ref{lem:explicit} (equivalently, by scanning a basis of the syndrome space for the rank-reducing condition), the selected extension vector is
\begin{equation}
\mathbf{w}=(1,0,\omega,\omega,\omega^6,\omega^2,\omega^3,2)^{T}\in\F_9^{8}.
\end{equation}
Appending it to $H$ gives $H_1=[H\mid\mathbf{w}^{(1)}]$ with
$$
\Rank(H_1H_1^{\dagger})=1\to 0,\qquad d\bigl(\langle H_1\rangle^{\perp_h}\bigr)=6,
$$
and we obtain a self-orthogonal code $\mathcal{C}_{\mathrm{so}}=[80,8,56]_9$ whose Hermitian dual has parameters $\mathcal{C}_{\mathrm{so}}^{\perp_h}=[80,72,6]_9$ with weight enumerator polynomial $W_{72,80}(z)=$$1+1929192z^{6}+139704824z^{7}+10180184456z^{8}+641595165672z^{9}+35936588753792 z^{10}+\dots$. Thus
$$
\mathcal{Q}=[[80,\,2\cdot 71-79+1,\,6]]_3=\mathbf{[[80,64,6]]_3},
$$
improving the previously known parameters in Grassl's tables.
By Proposition~\ref{propagation_rule}, $[[80,64,6]]_3$ also yields the additional quantum codes $[[81,64,6]]_3$ (by lengthening) and $[[80,63,6]]_3$ (by subcode construction).

\textbf{Example 6 $[[131,93,8]]_2$.} Let $\mathcal{C}$ be the degree-2 quasi-cyclic $[130,111,8]_4$ code with generator polynomials
    $f_1(x) =  x^{46} + \omega^{2} x^{44} +\omega^{2} x^{43} + \omega x^{40} + x^{38} +\omega x^{37} +\omega x^{36} + x^{35} + x^{34} + 
    \omega x^{33} + \omega x^{32} + x^{28} + \omega x^{25} + \omega x^{23} + \omega x^{21} + x^{18} + \omega x^{14} +\omega x^{13} +
    x^{12} + x^{11} + \omega x^{10} + \omega x^{9} + x^{8} + \omega x^{6} + \omega^{2} x^{3} + \omega^{2} x^{2} + 1$, 
 
     $f_2(x) =  \omega x^{18} + x^{16} + x^{13} + x^{12} + \omega^{2}  x^{9} + \omega x^{7} + \omega x^{6} + \omega^{2}  x^{5} + \omega x^{3} + x$,

    $f_{3}(x) = x^{65} + 1$, and a fixed $v(x) =  
x^{46} + \omega^2 x^{44} + \omega^2 x^{43} + \omega x^{40} + x^{38} +\omega x^{37} + \omega x^{36} + x^{35} + x^{34} + 
    \omega x^{33} + \omega x^{32} + x^{28} + \omega x^{25} + \omega x^{23} + \omega x^{21} + x^{18} +\omega x^{14} + \omega x^{13} +
    x^{12} + x^{11} +\omega x^{10} + \omega x^{9} + x^{8} + \omega x^{6} + \omega^2 x^{3} + \omega^2 x^{2} + 1$.

A $2$-quasi-cyclic
code of length $2n$ is an $R$-submodule of $R^2$. Let $n=65$. We define
$$
    \mathcal{C}
    = \big\langle ( f_1,\; v f_2),\; (0,\;  f_3) \big\rangle_R
    \subseteq R^2,
    \label{eq:C_R_submodule_ex7}
$$
Equivalently, $\mathcal{C}$ is the $2 \times 2$ polynomial generator matrix
$$
    G(x) =
    \begin{pmatrix}
        f_1 & f_2 \\[2mm]
        0   &  f_3
    \end{pmatrix}
    \in \mathcal{M}_{2 \times 2}(R).
    \label{eq:poly_G_ex7}
$$
 Its Hermitian dual $\mathcal{C}^{\perp_h}$ has dimension $19$; taking a parity-check matrix $H$ of $\mathcal{C}$ (equivalently a generator matrix of $\mathcal{C}^{\perp_h}$) we compute
$A=HH^{\dagger}\in\F_4^{19\times 19},\qquad r=\Rank(A)=1$.
The selected extension vector satisfying the rank-reducing condition and the exact criterion for $d=8$ is
$$
\mathbf{w}=(1,1,\omega,1,\omega^2,\omega^2,\omega,\omega^2,0,\omega^2,0,\omega^2,
\omega,\omega^2,\omega^2,1,\omega,1,1)^{T}\in\F_4^{19}.
$$
Appending it to $H$ gives $H_1=[H\mid\mathbf{w}]$ with
$$
\Rank(H_1H_1^{\dagger})=1\to 0,\qquad d\bigl(\langle H_1\rangle^{\perp_h}\bigr)=8.
$$
The extended self-orthogonal code has parameters $\mathcal{C}_{\mathrm{so}}=[131,19,66]_4$; its Hermitian dual is $\mathcal{C}_{\mathrm{so}}^{\perp_h}=[131,112,8]_4$ with weight enumerator polynomial $W_{112,131}(z)=$$1+38610z^{8}+1700400z^{9}+61806030z^{10}+2044336710z^{11}+61342419840 z^{12}+1684397370060 z^{13}+42592353523800 z^{14}+\dots$. 
Consequently
$$
\mathcal{Q}=[[131,\,2\cdot 111-130+1,\,8]]_2=\mathbf{[[131,93,8]]_2},$$
improving the previously known parameters in Grassl's tables.

By Proposition~\ref{propagation_rule}, Example~6 also yields the additional quantum codes $[[132,93,8]]_2$, $[[133,93,8]]_2$, $[[134,93,8]]_2$, $[[135,93,8]]_2$ (by lengthening) and $[[131,92,8]]_2$, $[[132,92,8]]_2$, $[[133,92,8]]_2$ (by subcode construction).

\textbf{Example 7 $[[171,149,5]]_2$.} Let $\mathcal{C}$ be the $[170,159,5]_4$ quasi-cyclic code with generator polynomials
    $f_{1}(x) = x^{10} + \omega^{2} x^{9} + \omega x^{8} + \omega x^{7} + x^{6} + x^{2} + x + 1$, 
    $f_{2}(x) = x + 1$, and a fixed $v(x) = \omega x^{84} + \omega^2 x^{80} + \omega x^{79} + \omega^2 x^{78} + x^{77} + \omega^2 x^{76} + x^{75} + \omega x^{74} \\
             + \omega^2 x^{73} + \omega x^{72} + \omega^2 x^{71} + x^{70} + x^{69} + x^{66} + \omega x^{65} + x^{64} + \omega^2 x^{63} 
             + x^{62} + \omega^2 x^{58} + \omega x^{57} + \omega x^{56} + x^{54} + x^{53} + \omega^2 x^{52} + \omega^2 x^{51} \\
             + \omega x^{50} + \omega^2 x^{49} + \omega x^{48} + \omega^2 x^{47} + x^{46} + \omega x^{45} + x^{44} + \omega x^{43} 
             + \omega^2 x^{42} + x^{41} + \omega^2 x^{40} + x^{39} + \omega x^{38} + x^{37} + \omega^2 x^{36} + \omega^2 x^{34} \\
             + \omega x^{32} + \omega^2 x^{31} + x^{30} + \omega^2 x^{29} + \omega x^{28} + \omega x^{27} + \omega x^{25} + x^{23} 
             + \omega^2 x^{22} + \omega^2 x^{21} + \omega^2 x^{20} + \omega x^{19} + x^{17} + x^{16} + \omega x^{15} + \omega x^{14} \\
             + \omega x^{13} + \omega^2 x^{12} + x^{11} + x^{10} + \omega^2 x^9 + \omega^2 x^8 + \omega^2 x^7 + \omega^2 x^5 
             + \omega x^4 + x^2 + x + 1$.

A $2$-quasi-cyclic
code of length $2n$ is an $R$-submodule of $R^2$. Let $n=85$. We define
\begin{equation}
    \mathcal{C}
    = \big\langle (v f_1,\; f_1),\; (f_2,\; v f_2) \big\rangle_R
    \subseteq R^2,
    \label{eq:C_R_submodule}
\end{equation}
i.e., $\mathcal{C}$ is generated as an $R$-module by the two rows
\[
    \mathbf{g}_1 = (v f_1,\; f_1), \qquad
    \mathbf{g}_2 = (f_2,\; v f_2).
\]
Equivalently, $\mathcal{C}$ is the image of the $2 \times 2$ polynomial generator matrix
\begin{equation}
    G(x) =
    \begin{pmatrix}
        v f_1 & f_1 \\[2mm]
        f_2   & v f_2
    \end{pmatrix}
    \in \mathcal{M}_{2 \times 2}(R).
    \label{eq:poly_G}
\end{equation}
 Its Hermitian dual $\mathcal{C}^{\perp_h}$ has dimension $11$; taking a parity-check matrix $H$ of $\mathcal{C}$ (equivalently a generator matrix of $\mathcal{C}^{\perp_h}$) we compute
$A=HH^{\dagger}\in\F_4^{11\times 11},\qquad r=\Rank(A)=1$.

Here $A$ has rank one with $A=\mathbf{u}_1\mathbf{u}_1^{\dagger}$ and $\mathbf{u}_1$ isotropic, so that $\im(A)\subseteq\Ker(A)$; by Lemma~\ref{lem:explicit} the rank-reducing vectors are the nonzero multiples $k\mathbf{u}_1$ with $N(k)=-1$. Scanning a basis of the syndrome space for the rank-reducing condition, the selected vector is
$
\mathbf{w}=(1,1,1,0,0,0,1,\omega,\omega,\omega^2,1)^{T}\in\F_4^{11}.$
It satisfies the rank-reducing condition and the exact criterion for $d=5$.  Appending it to $H$ gives $H_1=[H\mid\mathbf{w}^{(1)}]$ and
$
\Rank(H_1H_1^{\dagger})=1\to 0,\qquad d\bigl(\langle H_1\rangle^{\perp_h}\bigr)=5.$
The resulting Hermitian self-orthogonal code has parameters $\mathcal{C}_{\mathrm{so}}=[171,11,86]_4$; its Hermitian dual is $\mathcal{C}_{\mathrm{so}}^{\perp_h}=[171,160,5]_4$ with weight enumerator polynomial $W_{160,171}(z)=$$1+74970z^{5}+5542680z^{6}+390228285z^{7}+24029742165z^{8}+1305147006270z^{9}+63429558626310z^{10}+\dots$.  

The Hermitian construction therefore yields
$$
\mathcal{Q}=[[171,\,2\cdot 159-170+1,\,5]]_2=\mathbf{[[171,149,5]]_2}.
$$
This improves the previously known parameters in Grassl's tables for $\mathcal{Q}=\mathbf{[[171,149,4]]_2}$.
By Proposition~\ref{propagation_rule}, $[[171,149,5]]_2$ also yields the additional quantum code $[[172,149,5]]_2$ (by lengthening).

\section{Extension matrices of the four inequivalent $[[61,25,9]]_2$ codes}\label{app:w61}

The four pairwise inequivalent quantum codes with parameters $[[61,25,9]]_2$ constructed in Example~\ref{ex:w61} are all obtained from the cyclic code $\mathcal{C}=[51,33,10]_4$ by appending ten columns
to its $18\times 51$ parity-check matrix $H$. All four share the same $H$ and differ only in the
appended columns. The four extension matrices are collected in Table~\ref{tab:w61vectors}; line $j$
of a block is the column $\mathbf{w}^{(j)}\in\F_4^{18}$, written in the coordinates $0,1,2,3$ for
$0,1,\omega,\omega^{2}$. Appending the ten columns of a block to $H$ yields the Hermitian
self-orthogonal matrix $H_r$ of the corresponding code, whose weight-$9$ word count $A_9$ is
reproduced in the second column.

\begin{table}[htbp]
\caption{The four $18\times10$ extension matrices attached to the cyclic code
$\mathcal{C}=[51,33,10]_4$ underlying the $[[61,25,9]]_2$ records of Example~\ref{ex:w61}.}
\label{tab:w61vectors}
\centering
\scriptsize
\ttfamily
\begin{tabular}{|c|c|l|}
\hline
No. & $A_9$ & $\mathbf{w}^{(1)},\dots,\mathbf{w}^{(10)}$ (columns, top to bottom) \\
\hline
1 & 3759 &
\makecell[l]{211023110220132002\\211322102101210130\\321103001000213233\\
132230021200313131\\320112200022211322\\211010022011103101\\
310301131023323023\\232001323331230322\\121000231001103303\\
133310011320330320}\\
\hline
2 & 3777 &
\makecell[l]{002010133220322121\\020130323301012133\\131210313013000320\\
103010000011111123\\330132101213011133\\320330113211122020\\
212101303223031321\\211010022011103101\\020012120132210223\\
022103313200330022}\\
\hline
3 & 3825 &
\makecell[l]{203233002312232003\\130311012010203313\\310333302211202033\\
313013013010202133\\101210112120002110\\203221121201100332\\
211010022011103101\\100223133202121200\\320033322210122212\\
231013303212022213}\\
\hline
4 & 4032 &
\makecell[l]{203131311220022221\\301320032233123212\\030321312110200122\\
301220102100021130\\311100210000133220\\131033112030122222\\
103210013033223122\\020133033030022331\\133030011033302303\\
323200220333013113}\\
\hline
\end{tabular}
\end{table}

\section{Independent verification}\label{app:verification}

This section records the enumerations referred to in the closing paragraph of Section~\ref{sec:structure}. None of them is used in any proof: every statement they test is proved above, or in
Section~\ref{app:structure} below, and the enumerations are redundancy checks only.

\emph{First enumeration (the structural identities).} The dictionary, orbit and
realisability identities, the parity and gauge corollaries, the counting, scarcity, spectrum,
mean-and-variance and rigid-case statements of Section~\ref{app:structure} below, and the minimality
statement on the least number of coordinates that can be adjoined were checked by exhaustive
enumeration of \emph{every} union $T$ of $q^{2}$-cyclotomic cosets for all $q\in\{2,3,5\}$ and all
$n<80$ with $\gcd(n,q)=1$, $148\,632$ defining sets in $101$ admissible pairs $(n,q)$. The finitely
many remaining pairs, with more than $2^{15}$ cosets, were checked for the closed form of $f_\tau$
alone.

\emph{Second enumeration (complementation through rate).} The statements of the complementation,
duality, extremal and rate subsections --- the complementation lemma, the duality
theorem with its corollaries, the extremal criteria, the minimisation identity, the sandwich, the rate
identity, and the counting identities of Section~\ref{app:structure} below --- were verified by a
second, independent enumeration, on \emph{every} defining set of $154$ moduli $(n,q)$ with
$q\in\{2,3,5\}$, $1\,183\,552$ defining sets in total, of which $19$ large moduli were tested by
$4\times10^{4}$ random unions each. The periodic formula was checked on all $360$ moduli with $q\le13$
and $n<110$, and for each of the $57$ moduli satisfying its hypothesis for every $m\le bt$ by
exhaustive enumeration.

\emph{Third enumeration (the $\tau$-invariant locus).} The interface results of the junction subsection
were checked by a third enumeration over all $174$ moduli with $q\in\{2,3,5\}$ and
$n<90$: for the $168$ moduli with $2^{a+b}\le2^{18}$ \emph{every} $\tau$-invariant defining set was
enumerated ($829\,776$ sets) and the size distribution compared degree by degree with the predicted
size distribution, while $3\,480\,000$ further defining sets were drawn at random and the
equivalences of the junction theorem tested along two independent computational routes, namely the
vanishing of the hull dimension and the invariance of the defining pattern under the coset permutation
induced by $\tau$.

No counterexample occurred in any of the three enumerations.

\bibliographystyle{elsarticle-num}
\bibliography{references}

@article{shor1995scheme,
  author = {P. M. Shor},
  title = {Scheme for Reducing Decoherence in Quantum Computer Memory},
  journal = {Phys. Rev. A},
  volume = {52},
  year = {1995},
  pages = {2493--2496}
}

@article{Steane1996,
  author = {A. M. Steane},
  title = {Multiple-Particle Interference and Quantum Error Correction},
  journal = {Proc. R. Soc. Lond. A},
  volume = {452},
  year = {1996},
  pages = {2551--2577}
}

@article{Calderbank1998,
  author = {A. R. Calderbank and E. M. Rains and P. W. Shor and N. J. A. Sloane},
  title = {Quantum Error Correction Via Codes Over {GF}(4)},
  journal = {IEEE Trans. Inf. Theory},
  volume = {44},
  year = {1998},
  pages = {1369--1387}
}

@article{Rains1999,
  author = {E. M. Rains},
  title = {Nonbinary Quantum Codes},
  journal = {IEEE Trans. Inf. Theory},
  volume = {45},
  year = {1999},
  pages = {1827--1832}
}

@article{Ashikhmin2001,
  author = {A. Ashikhmin and E. Knill},
  title = {Nonbinary Quantum Stabilizer Codes},
  journal = {IEEE Trans. Inf. Theory},
  volume = {47},
  year = {2001},
  pages = {3065--3072}
}

@article{Ketkar2006,
  author = {A. Ketkar and A. Klappenecker and S. Kumar and P. K. Sarvepalli},
  title = {Nonbinary Stabilizer Codes Over Finite Fields},
  journal = {IEEE Trans. Inf. Theory},
  volume = {52},
  year = {2006},
  pages = {4892--4914}
}

@article{Kai2014ConstacyclicCA,
  author = {X. Kai and S. Zhu and P. Li},
  title = {Constacyclic Codes and Some New Quantum {MDS} Codes},
  journal = {IEEE Trans. Inf. Theory},
  volume = {60},
  number = {4},
  year = {2014},
  pages = {2080--2086}
}

@article{feng2004finite,
  author = {K. Feng and Z. Ma},
  title = {A Finite {Gilbert-Varshamov} Bound for Pure Stabilizer Quantum Codes},
  journal = {IEEE Trans. Inf. Theory},
  volume = {50},
  number = {12},
  year = {2004},
  pages = {3323--3325}
}

@article{Li2019NewQC,
  author = {R. Li and J. Wang and Y. Liu and G. Guo},
  title = {New Quantum Constacyclic Codes},
  journal = {Quantum Inf. Process.},
  volume = {18},
  year = {2019},
  pages = {1--23}
}

@article{Guardia2012OnTC,
  author = {G. Guardia},
  title = {On the Construction of Nonbinary Quantum {BCH} Codes},
  journal = {IEEE Trans. Inf. Theory},
  volume = {60},
  number = {3},
  year = {2014},
  pages = {1528--1535}
}

@article{galindo2018quasi,
  author = {C. Galindo and F. Hernando and R. Matsumoto},
  title = {Quasi-Cyclic Constructions of Quantum Codes},
  journal = {Finite Fields Appl.},
  volume = {52},
  year = {2018},
  pages = {261--280}
}

@article{lv2019new,
  author = {J. Lv and R. Li and J. Wang},
  title = {New Binary Quantum Codes Derived From One-Generator Quasi-Cyclic Codes},
  journal = {IEEE Access},
  volume = {7},
  year = {2019},
  pages = {85782--85785}
}

@article{Guan2022QC,
  author = {C. Guan and R. Li and L. Lu and Y. Liu and H. Song},
  title = {On Construction of Quantum Codes With Dual-Containing Quasi-Cyclic Codes},
  journal = {Quantum Inf. Process.},
  volume = {21},
  year = {2022},
  pages = {263}
}

@article{daskalov2003new,
  author = {R. Daskalov and P. Hristov},
  title = {New Binary One-Generator Quasi-Cyclic Codes},
  journal = {IEEE Trans. Inf. Theory},
  volume = {49},
  number = {7},
  year = {2003},
  pages = {3001--3005}
}

@article{Sguin2004ACO,
  author = {G. S{\'e}guin},
  title = {A Class of 1-Generator Quasi-Cyclic Codes},
  journal = {IEEE Trans. Inf. Theory},
  volume = {50},
  number = {7},
  year = {2004},
  pages = {1745--1753}
}

@article{Ling2005OnTA,
  author = {S. Ling and P. Sol{\'e}},
  title = {On the Algebraic Structure of Quasi-Cyclic Codes {III}: Generator Theory},
  journal = {IEEE Trans. Inf. Theory},
  volume = {51},
  number = {8},
  year = {2005},
  pages = {2692--2700}
}

@article{Barbier,
  author = {M. Barbier and C. Chabot and G. Quintin},
  title = {On Quasi-Cyclic Codes as a Generalization of Cyclic Codes},
  journal = {Finite Fields Appl.},
  volume = {18},
  year = {2012},
  pages = {904--919}
}

@article{Aydin2001TheSO,
  author = {N. Aydin and I. Siap and D. K. Ray-Chaudhuri},
  title = {The Structure of 1-Generator Quasi-Twisted Codes and New Linear Codes},
  journal = {Des. Codes Cryptogr.},
  volume = {24},
  year = {2001},
  pages = {313--326}
}

@article{Grassl2020AlgebraicQC,
  author = {M. Grassl},
  title = {Algebraic Quantum Codes: Linking Quantum Mechanics and Discrete Mathematics},
  journal = {Int. J. Comput. Math. Comput. Syst. Theory},
  volume = {6},
  year = {2020},
  pages = {243--257}
}

@article{Lisonek2014,
  author = {P. Lisonek and V. Singh},
  title = {Quantum Codes From Nearly Self-Orthogonal Quaternary Linear Codes},
  journal = {Des. Codes Cryptogr.},
  volume = {73},
  year = {2014},
  pages = {417--424}
}

@incollection{Degwekar,
  author = {A. Degwekar and K. Guenda and T. A. Gulliver},
  title = {Extending Construction {X} for Quantum Error-Correcting Codes},
  booktitle = {Coding Theory and Applications},
  publisher = {Springer},
  year = {2015},
  pages = {141--152}
}

@inproceedings{ezerman2019good,
  author = {M. F. Ezerman and S. Ling and B. {\"O}zkaya and P. Sol{\'e}},
  title = {Good Stabilizer Codes From Quasi-Cyclic Codes Over {$F_4$} and {$F_9$}},
  booktitle = {Proc. IEEE Int. Symp. Inf. Theory},
  year = {2019},
  pages = {2898--2902}
}

@article{Ezerman,
  author = {M. F. Ezerman and M. Grassl and S. Ling and F. {\"O}zbudak and B. {\"O}zkaya},
  title = {Characterization of Nearly Self-Orthogonal Quasi-Twisted Codes and Related Quantum Codes},
  journal = {IEEE Trans. Inf. Theory},
  volume = {71},
  number = {1},
  year = {2025},
  pages = {499--517}
}

@article{Hu2023QuantumEC,
  author = {P. Hu and X. Liu},
  title = {Quantum Error-Correcting Codes From the Quantum Construction {X}},
  journal = {Quantum Inf. Process.},
  volume = {22},
  year = {2023},
  pages = {1--19}
}

@article{Wang2020,
  author = {J. Wang and R. Li and J. Lv and H. Song},
  title = {A New Method of Constructing Binary Quantum Codes From Arbitrary Quaternary Linear Codes},
  journal = {IEEE Commun. Lett.},
  volume = {24},
  number = {2},
  year = {2020},
  pages = {472--476}
}

@article{lv2020explicit,
  author = {J. Lv and R. Li and J. Wang},
  title = {An Explicit Construction of Quantum Stabilizer Codes From Quasi-Cyclic Codes},
  journal = {IEEE Commun. Lett.},
  volume = {24},
  number = {7},
  year = {2020},
  pages = {1067--1071}
}

@book{Horn2013,
  author = {R. A. Horn and C. R. Johnson},
  title = {Matrix Analysis},
  edition = {2nd},
  publisher = {Cambridge Univ. Press},
  address = {Cambridge, U.K.},
  year = {2013}
}

@book{Macwilliams,
  author = {F. J. MacWilliams and N. J. A. Sloane},
  title = {The Theory of Error-Correcting Codes},
  publisher = {North-Holland},
  address = {Amsterdam, The Netherlands},
  year = {1977}
}

@book{Huffman,
  author = {W. C. Huffman and V. Pless},
  title = {Fundamentals of Error-Correcting Codes},
  publisher = {Cambridge Univ. Press},
  address = {Cambridge, U.K.},
  year = {2003}
}

@article{RLi,
  author = {R. Li and G. Xu and L. Lu},
  title = {Decomposition of Defining Set of {BCH} Codes and Its Applications},
  journal = {J. Air Force Eng. Univ. (Nat. Sci. Ed.)},
  volume = {14},
  number = {2},
  year = {2013},
  pages = {86--89},
  note = {(in Chinese)}
}

@article{bosma1997magma,
  author = {W. Bosma and J. Cannon and C. Playoust},
  title = {The {M}agma Algebra System {I}: The User Language},
  journal = {J. Symb. Comput.},
  volume = {24},
  year = {1997},
  pages = {235--265}
}

@misc{Grassl:codetables,
  author = {M. Grassl},
  title = {Bounds on the Minimum Distance of Linear Codes and Quantum Codes},
  url = {http://www.codetables.de},
  year = {2008}
}

@article{Cao2024ADC,
  author = {M. Cao},
  title = {On Dual-Containing, Almost Dual-Containing Matrix-Product Codes and Related Quantum Codes},
  journal = {Finite Fields Appl.},
  volume = {96},
  year = {2024},
  pages = {102400},
  doi = {10.1016/j.ffa.2024.102400}
}

@misc{Ezerman2026GE,
  author = {Y. Li and M. F. Ezerman and S. Li and S. Ling and Z. Sun},
  title = {Generalized Extended Codes with Applications in Entanglement-Assisted Qubit and Qutrit Codes},
  year = {2026}, eprint = {2607.02170},
  archivePrefix = {arXiv}, primaryClass = {cs.IT}
}

@article{Sangwisut2015,
  author = {E. Sangwisut and S. Jitman and S. Ling and P. Udomkavanich},
  title = {Hulls of Cyclic and Negacyclic Codes over Finite Fields},
  journal = {Finite Fields Appl.},
  volume = {33},
  year = {2015},
  pages = {232--257},
  doi = {10.1016/j.ffa.2014.12.008}
}

@article{SunDC2024,
  author = {Z. Sun and C. Ding and T. Chen},
  title = {The Extended Codes of Some Linear Codes},
  journal = {Finite Fields Appl.},
  volume = {96},
  year = {2024},
  pages = {102401},
  doi = {10.1016/j.ffa.2024.102401}
}

@article{XieZhu2025,
  author = {D. Xie and S. Zhu},
  title = {Two Families of Self-Orthogonal Codes with Applications in {LCD} Codes and Optimally Extendable Codes},
  journal = {Finite Fields Appl.},
  volume = {108},
  year = {2025},
  pages = {102659},
  doi = {10.1016/j.ffa.2025.102659}
}

@misc{WangLuo2026,
  author = {J. Wang and J. Luo},
  title = {Shortest Embeddings of Linear Codes with Arbitrary Hull Dimension},
  year = {2026},
  eprint = {2604.08843},
  archivePrefix = {arXiv},
  primaryClass = {cs.IT}
}

@article{Lei2022Hulls,
  author = {Y. Lei and C. Li and Y. Wu and P. Zeng},
  title = {More Results on Hulls of Some Primitive Binary and Ternary {BCH} Codes},
  journal = {Finite Fields Appl.},
  volume = {82},
  year = {2022},
  pages = {102066},
  doi = {10.1016/j.ffa.2022.102066}
}

@article{Du2020HullsTIT,
  author = {Z. Du and C. Li and S. Mesnager},
  title = {Constructions of Self-Orthogonal Codes From Hulls of {BCH} Codes and Their Parameters},
  journal = {IEEE Trans. Inf. Theory},
  volume = {66},
  number = {11},
  year = {2020},
  pages = {6774--6785},
  doi = {10.1109/TIT.2020.2991635}
}

@article{Pang2023LargerHull,
  author = {B. Pang and S. Zhu and T. Yang and J. Gao},
  title = {{BCH} Codes With Larger Dimensional Hull},
  journal = {Des. Codes Cryptogr.},
  volume = {91},
  year = {2023},
  pages = {3933--3951},
  doi = {10.1007/s10623-023-01281-x}
}

\end{document}